\documentclass[a4paper,11pt]{article}
\usepackage{jheppub} 

\pdfoutput = 1
\usepackage[T1]{fontenc} % if not used, title becomes bold

\usepackage{amsmath}
\usepackage{amssymb}
\usepackage{amsthm}
\usepackage{bm}
\usepackage{braket}
\usepackage[svgnames,psnames]{xcolor}
\usepackage{graphicx}
\usepackage[caption=false]{subfig}
\usepackage{adjustbox}

\usepackage[bbgreekl]{mathbbol} %change mathbb style

\usepackage{mathdots}

\newcommand{\orcid}[1]{\href{https://orcid.org/#1}{\includegraphics[width=8pt]{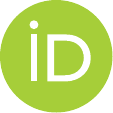}}} %orcid 

\usepackage{tikz}
\usetikzlibrary{positioning,intersections}
\usetikzlibrary{calc}
\usetikzlibrary{shapes.geometric}
\usepackage{braids}
\usepackage{tikz-cd}
\usetikzlibrary{shapes.geometric,shapes.misc}
\usetikzlibrary{arrows,matrix,calc,scopes,decorations.markings,snakes}
\usetikzlibrary{arrows.meta}
\usetikzlibrary{intersections, patterns,fit} 
\usetikzlibrary{decorations.pathreplacing,angles,quotes}
\definecolor{quantumviolet}{HTML}{53257F} %Quantum violet
\definecolor{quantumgray}{HTML}{555555} %Quantum gray
\definecolor{quantumgreen}{HTML}{007474} %Quantum green
\definecolor{quantumblue}{HTML}{002366} %Quantum gray
\definecolor{quantumpurple}{HTML}{66023C} %Quantum purple
\definecolor{quantumdarkviolet}{HTML}{5D3954} %Quantum dark violet

\theoremstyle{theorem}
\newtheorem{theorem}{Theorem}

\newtheorem{definition}{Definition}

\newtheorem{proposition}{Proposition}

\theoremstyle{remark}
\newtheorem{remark}{Remark}[section]

\newcommand{\Av}{\mathfrak{A}}
\newcommand{\Bf}{\mathfrak{B}}

\newcommand{\Zbb}{\mathbb{Z}}
\newcommand{\Cbb}{\mathbb{C}}
\newcommand{\Nbb}{\mathbb{N}}

\newcommand{\one}{\mathbb{1}}

\newcommand{\Rep}{\mathsf{Rep}}

\newcommand{\XR}{\overset{\rightarrow}{X}}
\newcommand{\XL}{\overset{\leftarrow}{X}}

\newcommand{\cone}{(1)}
\newcommand{\ctwo}{(2)}
\newcommand{\cthree}{(3)}

\newcommand\Tr{\operatorname{Tr}}

\newcommand{\id}{\operatorname{id}}
\newcommand{\Irr}{\operatorname{Irr}}
\newcommand{\Hom}{\operatorname{Hom}}
\newcommand{\End}{\operatorname{End}}

\newcommand{\cA}{ {\cal A} }

\usepackage{euscript}

\newcommand\eB           {\EuScript{B}}
\newcommand\eC           {\EuScript{C}}

\newcommand\eN         {\EuScript{N}}

\begin{document}

\flushbottom

\title{Generalized cluster states from Hopf algebras: non-invertible symmetry and Hopf tensor network representation}

\author[a,b]{Zhian Jia\orcid{0000-0001-8588-173X}}

\affiliation[a]{Centre for Quantum Technologies, National University of Singapore, Singapore 117543, Singapore}
\affiliation[b]{Department of Physics, National University of Singapore, Singapore 117543, Singapore}

\emailAdd{giannjia@foxmail.com}

\abstract{
Cluster states are crucial resources for measurement-based quantum computation (MBQC). It exhibits symmetry-protected topological (SPT) order, thus also playing a crucial role in studying topological phases.
We present the construction of cluster states based on Hopf algebras. By generalizing the finite group valued qudit to a Hopf algebra valued qudit and introducing the generalized Pauli-X operator based on the regular action of the Hopf algebra, as well as the generalized Pauli-Z operator based on the irreducible representation action on the Hopf algebra, we develop a comprehensive theory of Hopf qudits. We demonstrate that non-invertible symmetry naturally emerges for Hopf qudits.
Subsequently, for a bipartite graph termed the cluster graph, we assign the identity state and trivial representation state to even and odd vertices, respectively. Introducing the edge entangler as controlled regular action, we provide a general construction of Hopf cluster states.
To ensure the commutativity of the edge entangler, we propose a method to construct a cluster lattice for any triangulable manifold.
We use the 1d cluster state as an example to illustrate our construction. As this serves as a promising candidate for SPT phases, we construct the gapped Hamiltonian for this scenario and provide a detailed discussion of its non-invertible symmetries. We demonstrate that the 1d cluster state model is equivalent to the quasi-1d Hopf quantum double model with one rough boundary and one smooth boundary. We also discuss the generalization of the Hopf cluster state model to the Hopf ladder model through symmetry topological field theory. Furthermore, we introduce the Hopf tensor network representation of Hopf cluster states by integrating the tensor representation of structure constants with the string diagrams of the Hopf algebra, which can be used to solve the Hopf cluster state model.
}

\keywords{Anyons, Quantum Groups, Topological States of Matter, Topological Field Theories, Symmetry-protected topological (SPT) order}

\maketitle

\section{Introduction}

Quantum computation, with the potential to revolutionize computation by solving certain tasks that are classically intractable, has garnered considerable attention over the past several decades. It remains one of the main challenges in the development of modern quantum technologies \cite{Nielsen2010,preskill1998}.
The standard model for quantum computation is the quantum circuit model \cite{deutsch1989quantum} which acts both as a framework for theoretical investigations and as a guide for experiments. 
The measurement-based, or one-way, model of quantum computation (MBQC) provides a crucial alternative \cite{Raussendorf2001,Raussendorf2003measurement,nielsen2006cluster,briegel2009measurement}. MBQC utilizes highly entangled states, commonly referred to as graph states (sometimes called cluster states by certain authors). In this work, we adopt the term `cluster graph' to denote a bipartite graph with specific additional properties (see Definition~\ref{def:ClusterGraph}). Quantum computing operations are accomplished by performing measurements on these states.

The definition of a qubit graph state is simple.
    For a graph $K=(V,E)$, a qubit graph state is uniquely determined by the underlying graph $K$ \cite{Raussendorf2001,Raussendorf2003measurement,nielsen2006cluster,briegel2009measurement}:
\begin{equation} \label{eq:QubiGraphState}
	|K,\Cbb[\Zbb_2]\rangle =( \prod_{k=1}^{|E|} CZ_{e_k}) ( \otimes_{i=1}^{|V|}  |+\rangle_{v_i}),
\end{equation}
where $CZ$ is controlled-$Z$ gate.
The stabilizers are $T_v=X_v \otimes(\otimes_{w\in N(v)}Z_w)$ with $N(v)$ denoting the set of vertices that are neighbours of $v$.
See Ref.~\cite{hein2006entanglement} for a thorough review of the qubit graph state.
There exist a variety of generalizations of the qubit graph state, for example, hypergraph states \cite{Qu2013encoding,rossi2013quantum}; qudit graph and hypergraph state \cite{Steinhoff2017qudit,Xiong2017qudit,cui2015generalized}; cluster graph state based on finite group \cite{brell2015generalized,fechisin2023noninvertible}; and graph and hypergraph state for continuous variables \cite{Walschaers2018tailoring,Moore2019quantum}.
In addition to its application in MBQC, these highly entangled states also find utility in quantum teleportation, quantum metrology, quantum error-correction, quantum secret sharing, and various other quantum information tasks \cite{hein2006entanglement,Qu2013encoding,rossi2013quantum,Steinhoff2017qudit,Xiong2017qudit,cui2015generalized,fechisin2023noninvertible,Walschaers2018tailoring,Moore2019quantum,Looi2008QECC,Markham2008graph,Keet2010quantum}, among others mentioned in the literature.

From a quantum matter perspective, the one dimensional cluster state is a typical example of symmetry-protected topological (SPT) order.
The Hamiltonian corresponds to Eq.~\eqref{eq:QubiGraphState} is 
\begin{equation}\label{eq:ClusterHamiltonianNormal}
    H_{\rm cluster}=-\sum_i Z_{i-1}X_{i}Z_{i+1}.
\end{equation}
This model has the $\Zbb_2\times \Zbb_2$ SPT order \cite{son2012topological}. Recently, it's also argued that the cluster state model also has a non-invertible global symmetry, described by the fusion category $\Rep(D_8)$ \cite{seifnashri2024cluster}.
As the qubit system is a quantum system valued in the Abelian group $\mathbb{Z}_2$, it is natural to consider the generalization to the finite group case. In this scenario, the Calderbank-Shor-Steane (CSS)-type formulation of the graph state for bipartite graphs must be utilized \cite{brell2015generalized}.
After applying Hadamard gates over even vertices, we obtain the CSS qubit cluster state:
\begin{equation} \label{eq:CSScluster}
\begin{aligned}
     |K,\Cbb[\Zbb_2]\rangle_{\rm CSS}=
     (\prod_{
     \substack{\langle u,v\rangle \\ \\ u\in V_o,v\in V_e}  } CX_{\langle u,v\rangle}) [(\otimes_{u\in V_o} |+\rangle_{u})\otimes (\otimes_{v\in V_e} |0\rangle_v)],
\end{aligned}
\end{equation}
where we have used $\langle u,v\rangle$ to represent adjacent vertices (an edge).
The stabilizer for this state becomes $X_{i-1}X_iX_{i+1}$ for odd $i$ and $Z_{i-1}Z_iZ_{i+1}$ for even $i$, which is of the CSS-type. The corresponding cluster state Hamiltonian is 
\begin{equation}\label{eq:ClusterHamiltonianCSS}
    H=-\sum_{i\in 2\Nbb} Z_{i-1}Z_iZ_{i+1}-\sum_{i\in 2\Nbb+1} X_{i-1}X_iX_{i+1}
\end{equation}

Based on this observation, Brell generalizes the cluster state by using the finite group valued qudit \cite{brell2015generalized} and generalized form of Pauli operators, akin to Kitaev's quantum double model \cite{Kitaev2003,albert2021spin}.
The non-invertible symmetry and SPT order for finite group cluster state are recently discussed in Ref.~\cite{fechisin2023noninvertible}.
Via using tensor network representation, they show that for finite group $G$ valued cluster state, there is a $G\times \Rep(G)$ non-invertible symmetry.

However, Hopf algebras provide a more general framework for studying non-generalized symmetries. Roughly speaking, generalized global symmetries can be classified as invertible and non-invertible symmetries. The invertible symmetries are characterized by groups, while non-invertible symmetries are usually characterized by categories or other algebraic structures, where the elements are generally non-invertible. For more details, see, e.g., Refs. \cite{cordova2022snowmass,brennan2023introduction,mcgreevy2023generalized,luo2023lecture,shao2024whats,SchaferNameki2024ICTP,Bhardwaj2024lecture,delcamp2024higher}. Typical examples of non-invertible symmetries include fusion category symmetries \cite{Frohlich2004kramers,fuchs2007topological,frohlich2010defect,bhardwaj2018finite,chang2019topological,thorngren2019fusion,thorngren2021fusion,komargodski2021symmetries,inamura2022lattice,inamura2023fermionization}.
Depending on the codimension of the topological defect on which the symmetry operator is defined, there can be $0$-form and higher-form symmetries \cite{gaiotto2015generalized,kapustin2017higher,gomes2023introduction,Bhardwaj2024lecture}. Constructing lattice models that exhibit non-invertible symmetries and higher-form symmetries has attracted significant attention in recent years \cite{inamura2022lattice,inamura2023fermionization,bhardwaj2024illustrating,bhardwaj2024lattice}.
Since cluster state models have non-invertible symmetry, building on this line of inquiry, the authors in Ref.~\cite{fechisin2023noninvertible} pose the open problem of constructing the Hopf algebra cluster state and exploring its SPT order. The main objective of this work is to present a general construction of Hopf algebra valued cluster states and develop their tensor network representation. Discussions regarding SPT order will be addressed elsewhere \cite{jia2024cluster}.

Similar to that for finite group case, we find the Hopf cluster state is closely related to the Hopf quantum double model \cite{Buerschaper2013a,chen2021ribbon,jia2023boundary,Jia2023weak}.
The Haar integral of the Hopf algebra plays a crucial role in the Hopf quantum double model, inspired by this, we use the Haar integral as a preferred input state. 
We emphasize the significance of coalgebra and antipode structures for constructing a Hopf algebra valued cluster state.
The generalized Pauli-Z operators are defined based on representations of the Hopf algebra, while the generalized Pauli-X operators are defined using regular actions of the Hopf algebra on itself. The edge entangler is defined by a generalized controlled-X operation, which heavily relies on the comultiplication of the Hopf algebra.
The following are our main contributions: 
(i) We introduce the Hopf algebra-valued qudit and outline the construction of generalized Pauli operators in this scenario. Building upon this foundation, we proceed to construct the Hopf cluster states.
(ii) We introduce a Hopf tensor network that combines Hopf algebra structure constants with string diagrams of the Hopf algebra. We provide a detailed discussion of the Hopf tensor network representation of the Hopf cluster state.
(iii) We delve into the non-invertible symmetry of the Hopf algebra-valued qudit system in detail and develop the stabilizer formalism for the Hopf cluster state. Additionally, we construct the Hamiltonian for the $1$-dimensional cluster graph.
(iv) We also provide a generalization of cluster states to arbitrary graph and hypergraph based on the pairing of (weak) Hopf algebras, this is more similar to the original definition of qubit graph state in Eq.~\eqref{eq:QubiGraphState}.

The paper is organized as follows. 
In Sec.~\ref{sec:HopfQudit}, we provide an overview of the fundamental definition of Hopf algebra, followed by a discussion on constructing a Hopf algebra valued qudit and presenting the corresponding generalized Pauli operators.
We also demonstrate that the Hopf qudit naturally embodies non-invertible symmetries.
Sec.~\ref{sec:HopfCSSCluster} is devoted to constructing the CSS-type Hopf cluster state.
Sec.~\ref{sec:ClusterLattice} explains how to construct a Hopf cluster state on arbitrary given spatial manifold such that the corresponding entangler circuit is of finite depth.
The Hopf tensor network is introduced in Sec.~\ref{sec:HopfTensor}.
In Sec.~\ref{sec:Hopf1d}, we take one dimensional as an example to illustrate our construction and present the explicit Hamiltonian in this case. We show that the Hopf cluster state model has non-invertible symmetries. In Sec.~\ref{sec:QD}, the equivalence between the 1d Hopf cluster state model and the quasi-1d Hopf quantum double model is established. In Sec.~\ref{sec:SymTFT}, we generalize the cluster state model to Hopf ladder model based on symmetry topological field theory (SymTFT).
Sec.~\ref{sec:GraphHypergraphHopf} is devoted to the construction of the general graph and hypergraph state based on a Hopf algebra.
In the last part, we provide some concluding remarks.

\section{Hopf algebra valued qudit}
\label{sec:HopfQudit}

In this section, we introduce the definition of a Hopf algebra and discuss its properties that will be utilized in our construction of Hopf cluster states. We delve into the concept of a Hopf algebra valued qudit, which is a generalization of the finite group valued qudit, and discuss its details. Furthermore, we provide a discussion of the concept of Hopf symmetry and its representation category symmetry for qudits. These concepts play a crucial role in the study of SPT phases for Hopf cluster states, which will be discussed elsewhere \cite{jia2024cluster}.

\subsection{Preliminaries of Hopf algebra theory}
\label{subsec:DefinitionHopf}
Let's begin by providing an abstract definition of a Hopf algebra before delving into its detailed explanation. For a more comprehensive discussion of Hopf algebras, refer to, e.g., Refs.~\cite{abe2004hopf,kassel2012quantum}. Our conventions and notations are the same as those used in our previous work on Hopf and weak Hopf quantum double models~\cite{jia2023boundary,Jia2023weak}. As we will see, our construction of Hopf cluster states is also closely related to the Hopf quantum double model, which will be illustrated later in Sec.~\ref{sec:QD}.

\begin{definition} \label{def:HopfAlg}
A Hopf algebra is a complex vector space $\cA$ equipped with several structure linear morphisms: multiplication $\mu: \cA\otimes \cA\to \cA$, unit $\eta: \mathbb{C}\to \cA$, comultiplication $\Delta: \cA\to \cA\otimes \cA$, counit $\varepsilon: \cA\to \mathbb{C}$ and antipode $S:\cA\to \cA$, for which some consistency conditions are satisfied:
	\begin{enumerate}
		\item 	$(\cA,\mu,\eta)$	is an algebra: $\mu \circ (\mu \otimes \id) =\mu \circ (\id \otimes \mu)$,  and $\mu \circ (\eta \otimes \id)=\id =\mu \circ (\id \otimes \eta)$.
		
		\item $(\cA,\Delta,\varepsilon)$ is a coalgebra:  $(\Delta \otimes \id)\circ \Delta = (\id \otimes \Delta )\circ \Delta$,  and 
		$ (\varepsilon \otimes \id) \circ \Delta = \id = (\id \otimes \varepsilon)\circ \Delta$. 
		
		\item $(\cA,\mu,\eta,\Delta,\varepsilon)$ is a bialgebra: $\Delta$ and $\varepsilon$ are algebra homomorphisms (equivalently $\mu$ and $\eta$ are coalgebra homomorphisms).
		
		\item The antipode $S$ satisfies: $ \mu \circ (S\otimes \id) \circ \Delta =\eta \circ \varepsilon = \mu \circ (\id \otimes S)\circ \Delta$.
	\end{enumerate} 
\end{definition}

To accommodate readers with a background in physics, here we provide a detailed discussion of these algebraic concepts.
An algebra structure $(\mathcal{A}, \mu, \eta)$ can also be represented using string diagrams \cite{turaev2016quantum,bakalov2001lectures}. These string diagrams are quite inspiring for our construction of Hopf tensor networks, and they make tensor calculations very intuitive (see Sec.~\ref{sec:HopfTensor}).
		\begin{equation}
			\mu=\begin{aligned}
			\begin{tikzpicture}
				 \draw[black, line width=1.0pt]  (-0.5, 0) .. controls (-0.4, 1) and (0.4, 1) .. (0.5, 0);
				 \draw[black, line width=1.0pt]  (0,0.75)--(0,1.15);
			%	 \draw [fill = black](0, 0.75) circle (2pt);
				\end{tikzpicture}
			\end{aligned},\quad 
		\eta= \begin{aligned}
			\begin{tikzpicture}
				%\draw[black, line width=1.5pt]  (-0.5, 0) .. controls (-0.4, 1) and (0.4, 1) .. (0.5, 0);
				\draw[black, line width=1.0pt]  (0,0.75)--(0,1.5);
				\draw [fill = white](0, 0.73) circle (2pt);
			\end{tikzpicture}
		\end{aligned}\,\,\, .
		\end{equation}
We denote the unit element of $\cA$ as $1_{\cA}=\eta(1)$.
The algebra axioms can be represented as
\begin{gather}	
 \begin{aligned}
			\begin{tikzpicture}
				\draw[black, line width=1.0pt]  (-0.5, 0) .. controls (-0.4, 0.8) and (0.4, 0.8) .. (0.5, -0.5);
				\draw[black, line width=1.0pt]  (-0.08,0.53)--(-0.08,0.93);
			%	\draw [fill = black](-0.10,0.71) circle (2pt);
				\draw[black, line width=1.0pt]  (-0.8, -0.5) .. controls (-0.7, 0.19) and (-0.3, 0.19) .. (-0.2, -0.5);
				%	\draw [fill = black](-0.5,0) circle (2pt);
			\end{tikzpicture}
		\end{aligned}
=	
	\begin{aligned}
	\begin{tikzpicture}
		\draw[black, line width=1.0pt]  (-0.5, -0.5) .. controls (-0.4, 0.8) and (0.4, 0.8) .. (0.5, 0);
		\draw[black, line width=1.0pt]  (0.08,0.53)--(0.08,0.93);
		%	\draw [fill = black](-0.10,0.71) circle (2pt);
		\draw[black, line width=1.0pt]  (0.2, -0.5) .. controls (0.3, 0.19) and (0.7, 0.19) .. (0.8, -0.5);
		%	\draw [fill = black](-0.5,0) circle (2pt);
	\end{tikzpicture}
\end{aligned},
	\quad 
			\id\,\,\,
	\begin{aligned}
		\begin{tikzpicture}
			\draw[black, line width=1.0pt] (0,-0.5)--(0,0.7);
		\end{tikzpicture}
	\end{aligned}
= \begin{aligned}
			\begin{tikzpicture}
				%\draw[black, line width=1.5pt]  (-0.5, 0) .. controls (-0.4, 1) and (0.4, 1) .. (0.5, 0);
				\draw[black, line width=1.0pt]   (0,0.3) arc (90:180:0.5);
				\draw [fill = white](-0.5, -0.2) circle (2pt);
				\draw[black, line width=1.0pt] (0,-0.5)--(0,0.7);
			\end{tikzpicture}
		\end{aligned}
	=
			 \begin{aligned}
		\begin{tikzpicture}
			%\draw[black, line width=1.5pt]  (-0.5, 0) .. controls (-0.4, 1) and (0.4, 1) .. (0.5, 0);
			\draw[black, line width=1.0pt]   (0.5,-0.2) arc (0:90:0.5);
			\draw [fill = white](0.5, -0.2) circle (2pt);
			\draw[black, line width=1.0pt] (0,-0.5)--(0,0.7);
		\end{tikzpicture}
	\end{aligned}\,\,\,.
\end{gather}
The group algebra $\Cbb[G]$ is a typical example of algebras, for which the multiplication of $x=\sum_{g\in G}x_g g$ and $y=\sum_{k\in G}y_k k$ is given by $xy=\sum_{l\in G}(\sum_{g\in G }x_g y_{g^{-1}l}) l$ which is induced by group multiplication.

The coalgebra is particularly crucial, as it outputs a multipartite (entangled) state for a single qudit input. The structure of entanglement is naturally encoded in the comultiplication maps.
Fix a basis $\{g_i\}$ of $\cA$, $\Delta:\cA\to \cA\otimes \cA$ can be represented as
\begin{equation}
    \Delta(x)=\sum_{i,j}X_{ij}g_i\otimes g_j,\quad  x\in \cA.
\end{equation}
This can be reformulated as
\begin{equation}
     \Delta(x)=\sum_{i} g_i\otimes h_i, \quad h_i=\sum_j X_{ij}g_j,
\end{equation}
which is similar to the Schmidt decomposition of an entangled state.
This type of representation of the coalgebra is known as Sweedler's notation. Through coassociativity, we can define $\Delta_n$ recursively via $\Delta_n=(\Delta\otimes \text{id}) \circ \Delta_{n-1}$ without ambiguity. Then, we have the following notation
\begin{equation}
\Delta_n(x)=\sum_{(x)} x^{(1)}\otimes \cdots \otimes x^{(n)}.
\end{equation}
See Refs.~\cite{abe2004hopf,kassel2012quantum} for details. Sweedler's notation, much like the Einstein summation convention, can simplify calculations involving many variables. We will also omit the sum symbol whenever necessary for avoiding the clustering of equations. It also has a very natural explanation via tensor networks, as will be discussed in Sec.~\ref{sec:HopfTensor}.
This string diagram for coalgebra is dual to that of algebra (by rotating $180^{\circ}$), 
\begin{equation}
			\Delta=\begin{aligned}
				\begin{tikzpicture}
					\draw[black, line width=1.0pt]  (-0.5, 1) .. controls (-0.4, 0) and (0.4, 0) .. (0.5, 1);
					\draw[black, line width=1.0pt]  (0,0.23)--(0,-0.23);
					%\draw [fill = black](0, 0.23) circle (2pt);
				\end{tikzpicture}
			\end{aligned},\quad 
			\varepsilon= \begin{aligned}
				\begin{tikzpicture}
					%\draw[black, line width=1.5pt]  (-0.5, 0) .. controls (-0.4, 1) and (0.4, 1) .. (0.5, 0);
					\draw[black, line width=1.0pt]  (0,0.75)--(0,1.5);
					\draw [fill = white](0, 1.5) circle (2pt);
				\end{tikzpicture}
			\end{aligned}\,\,\, .
		\end{equation}
The coalgebra axiom can be represented
  \begin{gather}
	\begin{aligned}
			\begin{tikzpicture}
				\draw[black, line width=1.0pt]  (-0.5, 0.8) .. controls (-0.4, 0) and (0.4, 0) .. (0.5, 1.3);
				\draw[black, line width=1.0pt]  (-0.1,-0.16)--(-0.1,0.24);
				\draw[black, line width=1.0pt]  (-0.8, 1.3) .. controls (-0.7, 0.62) and (-0.3, 0.62) .. (-0.2, 1.3);
				%\draw [fill = black](0, 0.23) circle (2pt);
			\end{tikzpicture}
		\end{aligned}
	=
		\begin{aligned}
		\begin{tikzpicture}
			\draw[black, line width=1.0pt]  (-0.5, 1.3) .. controls (-0.4, 0) and (0.4, 0) .. (0.5, 0.8);
			\draw[black, line width=1.0pt]  (0.1,-0.16)--(0.1,0.24);
			\draw[black, line width=1.0pt]  (0.2, 1.3) .. controls (0.3, 0.62) and (0.7, 0.62) .. (0.8, 1.3);
			%\draw [fill = black](0, 0.23) circle (2pt);
		\end{tikzpicture}
	\end{aligned},
	\quad 
		\id\,\,\,
			\begin{aligned}
			\begin{tikzpicture}
				\draw[black, line width=1.0pt] (0,-0.5)--(0,0.7);
			\end{tikzpicture}
		\end{aligned}
		=\begin{aligned}
			\begin{tikzpicture}
				%\draw[black, line width=1.5pt]  (-0.5, 0) .. controls (-0.4, 1) and (0.4, 1) .. (0.5, 0);
				\draw[black, line width=1.0pt]   (-0.5,0.5) arc (180:270:0.5);
				\draw [fill = white](-0.5, 0.5) circle (2pt);
				\draw[black, line width=1.0pt] (0,-0.5)--(0,0.7);
			\end{tikzpicture}
		\end{aligned}
		=
		\begin{aligned}
			\begin{tikzpicture}
				%\draw[black, line width=1.5pt]  (-0.5, 0) .. controls (-0.4, 1) and (0.4, 1) .. (0.5, 0);
				\draw[black, line width=1.0pt]   (0,0) arc (270:360:0.5);
				\draw [fill = white](0.5, 0.5) circle (2pt);
				\draw[black, line width=1.0pt] (0,-0.5)--(0,0.7);
			\end{tikzpicture}
		\end{aligned}.
	\end{gather}
Notice that expressed in elements, the second equation is of the form $\sum_{(x)}\varepsilon(x^{\cone})x^{\ctwo}=x= \sum_{(x)}\varepsilon(x^{\ctwo})x^{\cone}$.
The set $\Cbb[G]$ also has a coalgebra structure $\Delta(g)=g\otimes g$ and $\varepsilon(g)=1$ for all $g\in G$.

The string diagrams for bialgebra conditions (3 of Definition~\ref{def:HopfAlg}) can be represented as: (i) $\Delta(xy)=\Delta(x)\Delta(y)$: 
\begin{equation}
			\begin{aligned}
			\begin{tikzpicture}
				%\draw[black, line width=1.5pt]  (-0.5, 0) .. controls (-0.4, 1) and (0.4, 1) .. (0.5, 0);
				\draw[black, line width=1.0pt]   (-0.5,0)..   controls (-0.4,0.8) and (0.4,0.8)             ..(0.5,0);
				%\draw [fill = white](0.5, 0.5) circle (2pt);
				\draw[black, line width=1.0pt] (0,0.6)--(0,1);
					\draw[black, line width=1.0pt]   (-0.5,1.6)..   controls (-0.4,0.8) and (0.4,0.8)             ..(0.5,1.6);
			\end{tikzpicture}
		\end{aligned}
	=
		\begin{aligned}
		\begin{tikzpicture}
		\draw[black, line width=1.0pt] (-0.1,-1) arc (180:360:0.3);
		\draw[black, line width=1.0pt] (0.5,0) arc (0:180:0.3);
		 \draw[black, line width=1.0pt] (-0.1,0)--(-0.1,-1);
		 	\draw[black, line width=1.0pt] (1.6,0) arc (0:180:0.3);
		 	\draw[black, line width=1.0pt] (1,-1) arc (180:360:0.3);
		 \draw[black, line width=1.0pt] (1.6,0)--(1.6,-1);
		 	 \draw[black, line width=1.0pt] (0.2,-1.3)--(0.2,-1.6);
		 	  \draw[black, line width=1.0pt] (0.2,0.3)--(0.2,0.6);
		 	  \draw[black, line width=1.0pt] (1.3,0.3)--(1.3,0.6);
		 	  \draw[black, line width=1.0pt] (1.3,-1.3)--(1.3,-1.6);
			\braid[
			width=0.5cm,
			height=0.5cm,
			line width=1.0pt,
			style strands={1}{black},
			style strands={2}{black}] (Kevin)
			s_1^{-1} ;
		\end{tikzpicture}
	\end{aligned}\,\,\,.
	\end{equation}
and (ii) $\varepsilon(xy)=\varepsilon(x)\varepsilon(y)$:
\begin{equation}
			\begin{aligned}
			\begin{tikzpicture}
				 \draw[black, line width=1.0pt]  (-0.5, 0) .. controls (-0.4, 1) and (0.4, 1) .. (0.5, 0);
				 \draw[black, line width=1.0pt]  (0,0.75)--(0,1.15);
                \draw [fill = white](0, 1.15) circle (2pt);				\end{tikzpicture}
			\end{aligned}=
		\begin{aligned}
				\begin{tikzpicture}
					%\draw[black, line width=1.5pt]  (-0.5, 0) .. controls (-0.4, 1) and (0.4, 1) .. (0.5, 0);
					\draw[black, line width=1.0pt]  (0,0.75)--(0,1.5);
					\draw [fill = white](0, 1.5) circle (2pt);
				\end{tikzpicture}
			\end{aligned} \,\,
                \begin{aligned}
				\begin{tikzpicture}
					%\draw[black, line width=1.5pt]  (-0.5, 0) .. controls (-0.4, 1) and (0.4, 1) .. (0.5, 0);
					\draw[black, line width=1.0pt]  (0,0.75)--(0,1.5);
					\draw [fill = white](0, 1.5) circle (2pt);
				\end{tikzpicture}
			\end{aligned} \,\, .
		\end{equation}
The group algebra $\Cbb[G]$ is also a bialgebra.

The antipode map plays a similar role to taking inverses in a group, but in algebra, elements generally do not have inverses. Thus, we use a map $S: \mathcal{A} \rightarrow \mathcal{A}$ to characterize this operation.
The axiom of antipode can be expressed in string diagram as
\begin{equation}
    \begin{aligned}	
				\begin{tikzpicture}
					\draw[black, line width=1.0pt] (0.5,0) arc (0:180:0.3);
					%	\draw[black, line width=1.0pt] (1.6,0) arc (0:180:0.3);
					\draw[black, line width=1.0pt] (-0.1,-0.8) arc (180:360:0.3);
					\draw[black, line width=1.0pt] (-0.1,0)--(-0.1,-0.8);
					\draw[black, line width=1.0pt] (0.5,-0.8)--(0.5,-0.65);
					\draw[black, line width=1.0pt] (0.5,0)--(0.5,-0.3);
					%\draw[black, line width=1.0pt] (1,0)--(1,0.55);
					%\draw[black, line width=1.0pt] (1.6,0)--(1.6,-1.3);
					%	\draw[black, line width=1.0pt] (1.3,0.3)--(1.3,0.55);
					\draw[black, line width=1.0pt] (0.2,-1.1)--(0.2,-1.35);
					\draw[black, line width=1.0pt] (0.2,0.3)--(0.2,0.55);
					%\draw [fill = white]  (0.2,0.55) circle (2pt);
					%\draw [fill = white] (0.2,-1.35) circle (2pt);
					%\draw [fill = white] (1.3,0.55) circle (2pt);
					\draw (0.3,-0.3) rectangle (0.7,-0.65);
					\node (start) [at=(Kevin-1-s),yshift=-0.47cm] {$S$};
					%	\node[squarednode]    at  (1.25,0.4)                {2};
				\end{tikzpicture}
			\end{aligned}
			=
  \begin{aligned}
   \begin{tikzpicture}
					%\draw[black, line width=1.5pt]  (-0.5, 0) .. controls (-0.4, 1) and (0.4, 1) .. (0.5, 0);
					\draw[black, line width=1.0pt]  (0,0.75)--(0,2.25);
					\draw [fill = white](0, 1.5) circle (2pt);
				\end{tikzpicture}
    \end{aligned}
    =
    \begin{aligned}	
				\begin{tikzpicture}
					\draw[black, line width=1.0pt] (0.5,0) arc (0:180:0.3);
					\draw[black, line width=1.0pt] (-0.1,-0.8) arc (180:360:0.3);
					\draw[black, line width=1.0pt] (-0.1,-0.65)--(-0.1,-0.8);
					\draw[black, line width=1.0pt] (-0.1,0)--(-0.1,-0.3);
					\draw[black, line width=1.0pt] (0.5,-0.8)--(0.5,0);
					\draw[black, line width=1.0pt] (0.2,-1.1)--(0.2,-1.35);
					\draw[black, line width=1.0pt] (0.2,0.3)--(0.2,0.55);
					\draw (-0.3,-0.3) rectangle (0.1,-0.65);
					\node (start) [at=(Kevin-1-s),xshift=-0.58cm,yshift=-0.47cm] {$S$};
				\end{tikzpicture}
			\end{aligned}.
\end{equation}
Expressed in elements, we have
$\sum_{(x)}S(x^{(1)})x^{(2)}=\varepsilon(x)1_{\cA}=\sum_{(x)}x^{(1)}S(x^{(2)})$.
The antipode map also satisfies 
\begin{equation}
\begin{aligned}
    S(1_{\cA})=1_{\cA},
    S(xy)=S(y)S(x),
    \varepsilon(S(x))=\varepsilon(x),\\
    \sum_{(x)}S(x^{(2)})\otimes S(x^{(1)})=\sum_{(S(x)} S(x)^{(1)}\otimes S(x)^{(2)}.
\end{aligned}
\end{equation}
We will also assume the Hopf algebra to be both semisimple and cosemisimple, implying $S^2=\operatorname{id}_{\cA}$ by the Larson-Radford theorem \cite{larson1988semisimple}. From $S^2=\operatorname{id}_{\cA}$, we obtain
\begin{equation}
    \sum_{(x)}S(x^{(2)})x^{(1)}=\varepsilon(x),
\end{equation}
which plays a crucial role later in constructing entangler gates (See Sec.~\ref{sec:HopfCSSCluster}).

A $*$-operation on a Hopf algbra $\cA$ is a antilinear map $*:\cA\to \cA$ which satisfies
\begin{equation}
    \begin{aligned}
        (x^*)^*=x, (xy)^*=y^*x^*, 1_{\cA}^*=1_{\cA},\\
       \Delta(x^*)=(\Delta(x))^*.
    \end{aligned}
\end{equation}
And $\cA$ is called $C^*$ Hopf algebra if there is a $*$-embedding of $\cA$ into the algebra of operators for some Hilbert space.
Hereinafter, we will assume that our Hopf algebra is a $C^*$ Hopf algebra unless otherwise stated.
For  $C^*$ Hopf algebra, we have 
\begin{equation}
    S\circ *\circ S\circ *=\id.
\end{equation}
Using $S^2=\id$, we obtain
\begin{equation}\label{eq:Sstar}
    S(x)^*=S(x^*).
\end{equation}

Another notion we will use is the Haar integral $\lambda\in \cA$, which is defined as a normalized two-side integral.
A left (resp. right) integral of $\cA$ is an element $\ell$ (resp. $r$) satisfying $x\ell=\varepsilon(x) \ell$ (resp. $rx=r\varepsilon(x)$) for all $x\in \cA$. 
If $h$ is simultaneously left and right integral, it's called a two-side integral.
A left (resp. right) integral $\ell$ (resp. $r$) is called normalized if $\varepsilon(\ell)=1$ (resp. $\varepsilon(r)=1$). We see that  normalized integral is idempotent $\ell^2=\varepsilon(\ell) \ell=\ell$ (resp.  $r^2=r\varepsilon(r) =r$).
For a semisimple Hopf algebra $\cA$, there exist a unique Haar integral.
When $\cA=\Cbb[G]$, the Haar integral is of the form
\begin{equation}
    \lambda=\frac{1}{|G|}\sum_{g\in G} g.
\end{equation}
For general Hopf algebra, the Haar integral can be expressed as in Eq.~\eqref{eq:HaarLambda2}.

For a $C^*$ Hopf algebra $\cA$, the dual space $\bar{\cA}:=\operatorname{Hom}(\cA,\Cbb)$ consists of all bounded linear functions from $\cA$ to $\Cbb$ is also a $C^*$ Hopf algebra.
Consider the canonical pairing $\langle\bullet, \bullet \rangle: \bar{\cA}\times \cA \to \mathbb{C}$, we have
\begin{align}
&	\langle \bar{\mu}(\varphi\otimes \psi),x\rangle=\langle\varphi\otimes \psi, \Delta(x)\rangle,\label{eq:MultiplicationBarA}\\
&	\langle \bar{\eta} (1),x\rangle= \varepsilon(x),   \; \text{i.e.,}\; \bar{1}=\varepsilon,\\
&   \langle \bar{\Delta}(\varphi), x\otimes y \rangle=\langle \varphi, \mu(x\otimes y)\rangle,\label{eq:ComultiplicationBarA} \\
&    \bar{\varepsilon}(\varphi)=\langle \varphi, \eta(1)\rangle,\\
&    \langle \bar{S}(\varphi),x\rangle =\langle \varphi, S(x)\rangle.
\end{align}
The $*$-operation on $\bar{\cA}$ can be defined as
\begin{equation}
	\langle \varphi^*, x\rangle=\overline{ \langle \varphi, S(x)^*\rangle}. 
\end{equation}
We denote the Haar integral of $\bar{\cA}$ as $\Lambda$, which is also called a Haar measure of $\cA$. For a Hopf algebra $\cA$ endowed with a $*$-structure, we can introduce an inner product over $\cA$ as 
\begin{equation}\label{eq:innerProd}
    \langle x,y\rangle = \Lambda(x^*y).
\end{equation}
It can be proven that the above definition is a complex inner product.
Also notice that for $\cA=\Cbb[G]$, the Haar measure is $\Lambda=\delta_{1_G}(\cdot)$.

\subsection{Hopf algebra valued qudit and non-invertible Hopf symmetries}

With the above preparation, we now introduce the Hopf algebra valued qudit  (Hopf qudit), which is, by definition, a $C^*$ Hopf algebra $\cA$ equipped with the inner product in Eq.~\eqref{eq:innerProd}.
The standard basis of $\cA$ is a basis which includes $1_{\cA}$ as a special element:
\begin{equation}
   \{1_{\cA}=g_0,g_1,\cdots,g_{d-1}\}.
\end{equation}
Notice that any qudit can be regarded as a Hopf qudit. For example, qudit $\mathbb{C}^d$ can be regarded as a $d$-dimensional Hopf algebra $\mathcal{A}=\Cbb[\Zbb_d]$. The Hopf algebra structure for a given qudit is not unique, and for any $d$, such Hopf algebra structure exists.
Any element $h\in \cA$ can be regarded as a state 
\begin{equation}
   | h\rangle=\sum_{i=0}^{|\cA|-1} c_i |g_i\rangle.
\end{equation}

For a $C^*$ Hopf algebra $\cA$, a representation is an algebra homomorphism $\Gamma: \cA\to \operatorname{End}(V)$, viz.,
\begin{equation}
    \Gamma(g)\Gamma(h)=\Gamma(gh), \quad \Gamma(1_{\cA})=I_{d_{\Gamma}},
\end{equation}
where $d_{\Gamma}=\dim V$ is called the dimension of the representation.
In this work, we restrict our attention to finite-dimensional representations. The isomorphism classes of irreps of $\cA$ will be denoted as $\operatorname{Irr}(\cA)$.
If we use $\chi_{\Gamma}$ to denote the character corresponding to irreducible representation $\nu \in \Irr(\cA)$, the Haar integral of $\bar{\cA}$ can be expressed as \cite{larson1971characters}:
\begin{equation}\label{eq:HaarLambda1}
    \Lambda=\frac{1}{|\cA|}\sum_{\Gamma \in \Irr(\cA)} d_{\Gamma}\chi_{\Gamma}=\frac{1}{|\cA|}\sum_{\Gamma \in \Irr(\cA)} \chi_{\Gamma}(1_{\cA})\chi_{\Gamma}.
\end{equation}
From this, we see that the Haar integral of $\cA$ can be expressed as 
\begin{equation}\label{eq:HaarLambda2}
    \lambda=\frac{1}{|\cA|}\sum_{\Gamma \in \Irr(\bar{\cA})} d_{\Gamma}\chi_{\Gamma}=\frac{1}{|\cA|}\sum_{\Gamma \in \Irr(\bar{\cA})} {\varepsilon}(\chi_{\Gamma})\chi_{\Gamma},
\end{equation}
where we have used the fact $\chi_{\Gamma}\in \bar{\bar{\cA}}
\cong \cA$ is an element in $\cA$.

The trivial representation of $\cA$ is $\Cbb$ which is equipped with the action $g\triangleright z=\varepsilon(g)z,g\in \cA, z\in \Cbb$.
Recall that $\operatorname{End}(\Cbb)=\Cbb$, and $\varepsilon:\cA \to \Cbb$ satisfies $\varepsilon(gh)=\varepsilon(g)\varepsilon(h)$ and $\varepsilon(1_{\cA})=1$.
Unlike the group case, the tensor product of two representations $(\Gamma_1,V)$ and $(\Gamma_2,W)$ is a representation with action defined by
\begin{equation}\label{eq:RepTensor}
    h\triangleright (v\otimes w)=\sum_{(h)}\Gamma_1(h^{(1)}) v\otimes \Gamma_2(h^{(2)})w.
\end{equation}
The dual representation $\bar{\Gamma}$ of $(\Gamma,V)$ is defined as $\bar{V}:=\operatorname{Hom}(V,\mathbb{C})$, with the left $\cA$ action given by $(x\triangleright f)(y)=f(S(x)\triangleright y)$ for $f\in \bar{V}$. The category of finite dimensional representations of a Hopf algebra $\cA$ will be denoted as $\Rep(\cA)$.
The categories $\Rep(\cA)$ of finite dimensional representations of $\cA$ is a unitary fusion category~\cite{etingof2016tensor}. 
The fusion rule can be expressed as
\begin{equation}
      \Gamma\otimes \Phi=\oplus_{\Psi}N_{\Gamma,\Phi}^{\Psi}\Psi,
\end{equation}
where $N_{\Gamma,\Phi}^{\Psi}\in \Zbb_+$ is called the fusion symbol.
The Grothendieck ring $\operatorname{Gr}(\Rep(\cA))$ of $\Rep(\cA)$ is $\Zbb_+$-ring defined by 
\begin{equation}
    [\Gamma]\cdot [\Phi]=\sum_{\Psi}N_{\Gamma,\Phi}^{\Psi}[\Psi],
\end{equation}
where we have used the notation $[\Gamma]$ to stress the equivalence class of irreps.
The Grothendieck ring can also be expressed using the characters of irreps, we define
\begin{equation}\label{eq:MultiplicationChi}
    \chi_{\Gamma}\cdot \chi_{\Phi}:=\chi_{\Gamma\otimes \Phi}=\sum_{\Psi}N_{\Gamma,\Phi}^{\Psi} \chi_{\Psi}.
\end{equation}
Notice that $\chi_{\Gamma}$'s are functions over $\cA$, thus $\chi_{\Gamma}\in \bar{\cA}$. It satisfies
\begin{equation}
    \chi_{\Gamma}(ab)=\chi_{\Gamma}(ba),
\end{equation}
which implies that $\chi_{\Gamma}$ is cocommutative (from Eq.~\eqref{eq:ComultiplicationBarA}).
The multiplication in Eq.~\eqref{eq:MultiplicationChi} is also consistent with that of Eq.~\eqref{eq:MultiplicationBarA}:
\begin{equation}
    \langle \chi_{\Gamma}\cdot \chi_{\Phi},x\rangle=\sum_{(x)}\chi_{\Gamma}(x^{\cone}) \chi_{\Phi}(x^{\ctwo})=\sum_{\Psi}N_{\Gamma,\Phi}^{\Psi} \chi_{\Psi}(x).
\end{equation}
The above equations will be crucial for establishing the Hopf cluster state model and the Hopf quantum double model.

The standard basis plays a similar role as Pauli-Z basis $|0\rangle,|1\rangle$ for a qubit.
For $\Gamma\in \Rep(\cA)$, we introduce generalized Pauli-Z operators as follows:
\begin{equation}
Z_{\Gamma}|h\rangle=\sum_{(h)} |h^{(1)}\rangle \otimes \Gamma(h^{(2)}),
\end{equation}
which is well-defined since $Z_{\Gamma}=(\text{id}\otimes \Gamma)\circ \Delta$. Recall that $\Gamma: \mathcal{A} \to \operatorname{End}(V)$ has its domain in $\mathcal{A}$.
Similarly, we define
\begin{equation}
    Z^{\ddagger}_{\Gamma}|h\rangle=\sum_{(h)} \Gamma(S(h^{\cone}))\otimes |h^{\ctwo}\rangle.
\end{equation}
We can also introduce
\begin{align}
    \tilde{Z}_{\Gamma}|h\rangle= \sum_{(h)}\Gamma(h^{(1)})\otimes |h^{(2)}\rangle,\\
    \tilde{Z}_{\Gamma}^{\ddagger}|h\rangle= \sum_{(h)} |h^{(1)}\rangle\otimes \Gamma(S(h^{(2)})).
\end{align}
When $\mathcal{A} = \mathbb{C}[G]$, $Z_{\Gamma}$ and $Z^{\ddagger}_{\Gamma}$ take the same forms as Eqs (11) and (12) in Ref.~\cite{fechisin2023noninvertible}. And in this case, we have $Z_{\Gamma}=\tilde{Z}_{\Gamma}$ and $Z_{\Gamma}^{\ddagger}=\tilde{Z}_{\Gamma}^{\ddagger}$ (up to a reordering),  since $\Delta(g)=g\otimes g$ for all $g\in G$.
A similar notation is $Z_{\Gamma_{ij}}=(\text{id}\otimes \Gamma_{ij})\circ \Delta$:
\begin{equation}
    Z_{\Gamma_{ij}}|h\rangle =\sum_{(h)}\Gamma_{ij}(h^{\ctwo})|h^{\cone}\rangle.
\end{equation}
We can also introduce $Z_{\Gamma_{ij}}^{\ddagger}$,  $\tilde{Z}_{\Gamma_{ij}}$ and  $\tilde{Z}_{\Gamma_{ij}}^{\ddagger}$ in a similar way.
The comultiplication plays a crucial role, as we will see later, we cannot use a naive generalization of group case via $\sum_k\Gamma(g_k)\otimes |g_k\rangle \langle g_k|$.

The trace in the representation space ${V}$ will be denoted as $\operatorname{Tr}'$. Then, it's clear that
\begin{equation}
    (\Tr' Z_{\Gamma}) |1_{\cA}\rangle= d_{\Gamma}|1_{\cA}\rangle.
\end{equation}
We introduce 
\begin{equation}\label{eq:PauliZ}
    Z=\sum_{\Gamma\in \Irr(\cA)} \frac{d_{\Gamma}}{|\cA|} (\Tr' Z_{\Gamma}),
\end{equation}
and we have
\begin{equation}
     Z|1_{\cA}\rangle=|1_{\cA}\rangle.
\end{equation}
The unit element state plays the same role as that of Pauli-Z $+1$ eigenstate $|0\rangle$, $Z|0\rangle=|0\rangle$.

We introduce the following operators by tracing over representation spaces
\begin{equation}
    J_{\Gamma}=\Tr' Z_{\Gamma}, J^{\ddagger}_{\Gamma}=\Tr' Z^{\ddagger}_{\Gamma},  \tilde{J}_{\Gamma}=\Tr' \tilde{Z}_{\Gamma}, \tilde{J}^{\ddagger}_{\Gamma}=\Tr' \tilde{Z}^{\ddagger}_{\Gamma}. 
\end{equation}
Recall that in the quantum double model, we defined the left action of $\bar{\cA}$ on $\cA$ as
\begin{equation}
    T^{\varphi}_+|x\rangle=\sum_{(x)} \varphi(x^{\ctwo})|x^{\cone}\rangle, \quad  T^{\varphi}_-|x\rangle=\sum_{(x)} \bar{S}(\varphi)(x^{\cone})|x^{\ctwo}\rangle.
\end{equation}
The right action of $\bar{\cA}$ on $\cA$ are
\begin{equation}
    \tilde{T}^{\varphi}_+|x\rangle=\sum_{(x)} \bar{S}(\varphi)(x^{\ctwo})|x^{\cone}\rangle, \quad  \tilde{T}^{\varphi}_-|x\rangle=\sum_{(x)} \varphi(x^{\cone})|x^{\ctwo}\rangle.
\end{equation}
We have adopted the conventions we used in Ref.\cite[Eqs. (4.1)-(4.8)]{Jia2023weak} for quantum double model.
It's easy to verify that $J_{\Gamma}$ is the same as $T^{\chi_{\Gamma}}_+$. In fact, we have the following correspondence:
\begin{align}
   J_{\Gamma}= T^{\chi_{\Gamma}}_+,\quad 
    J_{\Gamma}^{\ddagger}= T^{\chi_{\Gamma}}_{-},\label{eq:JT1}\\
    \tilde{J}_{\Gamma}= \tilde{T}^{\chi_{\Gamma}}_{-},\quad
    \tilde{J}^{\ddagger}_{\Gamma}= \tilde{T}^{\chi_{\Gamma}}_{+}.\label{eq:JT2}
\end{align}
Eq.~\eqref{eq:HaarLambda1} and the fact that Haar measure $\Lambda$ in invariant under the action of antipode implies that
\begin{equation}\label{eq:Z1}
  T^{\Lambda}_+ = \sum_{\Gamma\in \Irr(\cA)} \frac{d_{\Gamma}}{|\cA|} (\Tr' Z_{\Gamma})=Z = \sum_{\Gamma\in \Irr(\cA)} \frac{d_{\Gamma}}{|\cA|} (\Tr' \tilde{Z}^{\ddagger}_{\Gamma})=\tilde{T}_+^{\Lambda},
\end{equation}
\begin{equation}\label{eq:Z2}
  \tilde{T}^{\Lambda}_- = \sum_{\Gamma\in \Irr(\cA)} \frac{d_{\Gamma}}{|\cA|} (\Tr' \tilde{Z}_{\Gamma})=\tilde{Z}= \sum_{\Gamma\in \Irr(\cA)} \frac{d_{\Gamma}}{|\cA|} (\Tr' {Z}^{\ddagger}_{\Gamma})= {T}^{\Lambda}_-.
\end{equation}
In fact, for Haar measure $\Lambda$, we have
\begin{equation}\label{eq:leftright}
    \sum_{(x)}x^{\cone} \Lambda(x^{(2)})= \Lambda(x)= \sum_{(x)}x^{\ctwo} \Lambda(x^{(1)}).
\end{equation}
This can be derived by taking $h=1$ for Eqs. (16) and (17) in Ref.~\cite{nikshych2002structure}.
Combining Eq.~\eqref{eq:leftright} and Eqs.~\eqref{eq:Z1} and \eqref{eq:Z2}, we obtain
\begin{equation}
    Z=\tilde{Z}, \quad Z|x\rangle=\Lambda(x) |1_{\cA}\rangle=|1_{\cA}\rangle \langle 1_{\cA}|x\rangle, \forall x\in \cA.
\end{equation}
Notice that $Z^2=T_+^{\Lambda^2}=T_+^{\Lambda}$. Since $\Lambda^2=\Lambda$, we see that $Z$ is idempotent.
From $Z^{\dagger}=(T_+^{\Lambda})^{\dagger}=T_+^{\Lambda^*}$ (this is not straightforward; see the proof of Theorem~\ref{thm:1dHopfLattice} for a discussion), and using $\Lambda^*=\Lambda$, we see that $Z$ is Hermitian.

These correspondence will play crucial roles in Sec.~\ref{sec:QD} for establishing the equivalence between 1d Hopf cluster state model and quasi-1d Hopf quantum double model.

\begin{table}
    \centering
    \begin{tabular}{|c|c|}
    \hline
Hopf qudit &  $\mathcal{H}_v=\cA$   \\
    \hline
    Standard basis  & $|g_0\rangle=|1_{\cA}\rangle,\cdots, |g_{|\cA|-1}\rangle$  \\     
        \hline
         Irrep basis & $ |\Gamma_{ij}\rangle=\sqrt{{d_{\Gamma}}{|\cA|}} \sum_{(\lambda)} \Gamma_{ij}(\lambda^{(1)})|\lambda^{(2)}\rangle$  \\
          \hline       
 Regular action  & $\begin{aligned}
& \XR_g|h\rangle =|gh\rangle,\XL_g|h\rangle =|hS(g)\rangle \\ 
& \tilde{\XR}_g|h\rangle =|S(g)h\rangle,\tilde{\XL}_g|h\rangle =|hg\rangle \end{aligned}$   \\
           \hline

Irrep action  &  $\begin{aligned}
   & Z_{\Gamma}|h\rangle=\sum_{(h)} |h^{(1)}\rangle \otimes \Gamma(h^{(2)})\\
   &  Z^{\ddagger}_{\Gamma}|h\rangle=\sum_{(h)} \Gamma(S(h^{\cone}))\otimes |h^{\ctwo}\rangle\\
   &  \tilde{Z}_{\Gamma} =\sum_{(h)} \Gamma(h^{\cone})\otimes |h^{\ctwo}\rangle\\
   &  \tilde{Z}^{\ddagger}_{\Gamma}  =\sum_{(h)} |h^{(1)}\rangle \otimes \Gamma(S(h^{(2)})) 
\end{aligned}$ \\
           \hline   
           
Symmetry action  &  $\begin{aligned}
   &  J_{\Gamma} = \Tr' Z_{\Gamma} , J^{\ddagger}_{\Gamma} = \Tr' Z^{\ddagger}_{\Gamma}   \\
   &   \tilde{J}_{\Gamma} = \Tr' Z_{\Gamma} , \tilde{J}^{\ddagger}_{\Gamma} = \Tr' Z^{\ddagger}_{\Gamma} 
\end{aligned}$ \\   
  \hline 
  
  Generalized Pauli operators &    $\begin{aligned}
& X=\XR_{\lambda}=\XL_{\lambda} =\tilde{\XR}_{\lambda}=\tilde{\XL}_{\lambda}=\tilde{X} \\
 & Z=\sum_{\Gamma\in \Irr(\cA)} \frac{d_{\Gamma}}{|\cA|} J_{\Gamma}=\sum_{\Gamma\in \Irr(\cA)} \frac{d_{\Gamma}}{|\cA|} \tilde{J}_{\Gamma}^{\ddagger}\\
  & \tilde{Z}=\sum_{\Gamma\in \Irr(\cA)} \frac{d_{\Gamma}}{|\cA|} \tilde{J}_{\Gamma}=\sum_{\Gamma\in \Irr(\cA)} \frac{d_{\Gamma}}{|\cA|} {J}_{\Gamma}^{\ddagger}\\
  &Z=\tilde{Z}
\end{aligned}$ \\
\hline
    \end{tabular}
    \caption{Summary of the quantum operations for Hopf qudit system.}
    \label{tab:HopfQudit}
\end{table}

We have constructed the generalized Pauli-Z operators and the standard basis, which play the role of qubit $Z$ operator and $|0\rangle, |1\rangle$ states. Let's now construct the generalized Pauli-X operator and the states corresponding to $|+\rangle,|-\rangle$ states.
Before that, we will first prove some results about the irreps of Hopf algebras. While these results may be familiar to experts in Hopf algebra, they do not appear elsewhere in the physical literature. Therefore, for comprehensiveness and convenience, we provide the proofs here.

\begin{proposition}[Schur's lemma for (weak) Hopf algebra]\label{prop:Schur}
Let $\Gamma:\cA\to \End(V)$ and $\Phi:\cA\to \End(W)$ be two irreducible representations of $\cA$, and let $f:V\to W$ be a linear map such that $\Phi(x)\circ f= f\circ \Gamma(x)$ for all $x\in \cA$. Then (i) If $\Gamma$ and $\Phi$ are not isomorphic, $f=0$; (ii) When $\Gamma=\Phi$, we have $f\propto \id_V$.
\end{proposition}

\begin{proof}
(i) Consider the nontrivial case that $f\neq 0$. We see that $\operatorname{Ker} f$ is stable under the action of $\cA$. Since for $v\in \operatorname{Ker} f$, $\Phi(x)\circ f (v)= f( \Gamma(x) v)=0$. Thus $\Gamma(x)v\in \operatorname{Ker} f$ for all $x\in \cA$. Since $(\Gamma,V)$ is irreducible, $\operatorname{Ker} f$ can only be $0$ or $V$.
A similar argument shows that $\operatorname{Im} f=W$. Thus $f$ is an isomorphism.

(ii) Since $f$ is a linear map over a complex field, thus must be at least one eigenvalue $p$ for $f$. Then $\operatorname{Ker} (f-p\id)\neq 0$. This implies that $f-p\id=1$ from assertion (i).
\end{proof}

For Haar integral $\lambda\in \cA$ and arbitrary $x\in \cA$, we have the following \cite[Lemma 2]{bridgeman2023invertible}
\begin{equation}\label{eq:lambdaSpecial}
  \sum_{(\lambda)}  xS(\lambda^{(1)})\otimes \lambda^{(2)}= S(\lambda^{(1)})\otimes \lambda^{(2)}x.
\end{equation}
Then for a linear map $f: V\to W$, we define
\begin{equation}\label{eq:Fave}
    F=\sum_{(\lambda)} \Phi(S(\lambda^{(1)})) \circ f \circ \Gamma(\lambda^{\ctwo}).
\end{equation}
It's clear from Eq.~\eqref{eq:lambdaSpecial} that
\begin{equation}
   \Phi(x) \circ F =F\circ \Gamma(x),\quad \forall x\in \cA.
\end{equation}

\begin{proposition}\label{prop:orthgonality}
 (1)  For any $f:V\to W$, and define $F$ as in Eq.~\eqref{eq:Fave}, we have: (i)  If $\Gamma$ and $\Phi$ are not isomorphic, $F=0$; (ii) when $\Gamma=\Phi$, we have
    \begin{equation}\label{eq:Ftrace}
        F=\frac{1}{d_{\Gamma}}\Tr' (f) \id_V.
    \end{equation}
    
(2) From the above result, we obtain the orthogonality relation for irreps of $\cA$:
\begin{equation}
    \sum_{(\lambda)} \Phi_{ij}(S(\lambda^{\cone})) \Gamma_{kl}(\lambda^{\ctwo})=\frac{\delta_{\Phi,\Gamma} \delta_{il}\delta_{jk} }{d_{\Gamma}}=  \sum_{(\lambda)} \Phi_{ij}(S(\lambda^{\ctwo})) \Gamma_{kl}(\lambda^{\cone}).
\end{equation}
where we have used the cocommutativity of $\lambda$.
\end{proposition}
\begin{proof}
   (1) This is a direct result of Proposition~\ref{prop:Schur}.  Notice that
    \begin{equation}
        \Tr' F= \sum_{(\lambda)} \Tr' [\Gamma(S(\lambda^{(1)})) \circ f \circ \Gamma(\lambda^{\ctwo})]=\Tr' [\Gamma(\varepsilon(\lambda)1_{\cA}) f]=\Tr' f,
    \end{equation}
which implies Eq.~\eqref{eq:Ftrace}.

(2) From (1), we have
\begin{equation}
    \begin{aligned}
         \sum_{(\lambda)} \sum_{jk} \Phi_{ij}(S(\lambda^{\cone}))f_{jk} \Gamma_{kl}(\lambda^{\ctwo})=  \frac{\delta_{\Phi,\Gamma}}{d_{\Gamma}} \delta_{il} \sum_{jk}
          f_{jk}\delta_{jk} .
    \end{aligned}
\end{equation}
This implies the required result.
\end{proof}

For irreps $\Gamma\in \Rep(\cA)$, the fusion basis states are defined using Haar integral $\lambda$ as follows:
\begin{equation}
    |\Gamma_{ij}\rangle=\sqrt{{d_{\Gamma}|\cA|}} \sum_{(\lambda)} \Gamma_{ij}(\lambda^{(1)})|\lambda^{(2)}\rangle, 
\end{equation}
where $i,j=1,\cdots, d_{\Gamma}$. From Peter-Weyl relation $H\cong \oplus_{\Gamma\in \operatorname{Irr}(\cA)} \bar{\Gamma}\otimes \Gamma$ we see that $|\cA|=\sum_{\Gamma\in \operatorname{Irr}(\cA)}d_{\Gamma}^2$.
Also notice that $\lambda$ is cocommutative $\sum_{(\lambda)}\lambda^{\cone}\otimes \lambda^{\ctwo}=\sum_{(\lambda)}\lambda^{\ctwo}\otimes \lambda^{\cone}$, we can equivalently define
\begin{equation}
     |\Gamma_{ij}\rangle=\sqrt{{d_{\Gamma}|\cA|} } \sum_{(\lambda)} \Gamma_{ij}(\lambda^{(2)})|\lambda^{(1)}\rangle, 
\end{equation}
Also note that the factor chosen here is different from that in Refs. \cite{buerschaper2013electric,chen2021ribbon,jia2023boundary,Jia2023weak}. This adjustment ensures that the basis is orthonormal under the inner product in Eq.~\eqref{eq:innerProd}. 

\begin{proposition}
    The fusion basis $|\Gamma_{ij}\rangle$ is an orthonormal basis for Hopf qudit $\cA$.
\end{proposition}

\begin{proof}
For clarity, we use $\eta,\lambda$ to denote the unique Haar integral of $\cA$,  we have
\begin{equation}
\begin{aligned}
        \langle \Phi_{ij}|\Gamma_{kl}\rangle=&\sum_{(\lambda),(\eta)} |\cA|\sqrt{d_{\Gamma}d_{\Phi} } \Phi^*_{ij}(\eta^{(1)}) \Gamma(\lambda^{(1)}) \Lambda((\eta^{\ctwo})^*\lambda^{\ctwo})\\
        =&\sum_{(\lambda),(\eta)} |\cA|\sqrt{d_{\Gamma}d_{\Phi} } \Phi^{\dagger}_{ji}(\eta^{(1)}) \Gamma(\lambda^{(1)}) \Lambda((\eta^{\ctwo})^*\lambda^{\ctwo}),\\
        =&\sum_{(\lambda),(\eta)} |\cA|\sqrt{d_{\Gamma}d_{\Phi} } \Phi_{ji}((\eta^{(1)})^*) \Gamma(\lambda^{(1)}) \Lambda((\eta^{\ctwo})^*\lambda^{\ctwo}),\\
        =&\sum_{(\lambda),(\eta)} |\cA|\sqrt{d_{\Gamma}d_{\Phi} } \Phi_{ji}((\eta^{(1)})) \Gamma(\lambda^{(1)}) \Lambda(\eta^{\ctwo}\lambda^{\ctwo}),
\end{aligned}
\end{equation}
where $\Lambda$ is Haar measure, we have used the property of unitary representation $\Gamma^{\dagger}(x)=\Gamma(x^*)$ and the properties for Haar integral $\eta=\eta^*$ and for $*$-operation $\Delta(\lambda^*)=\Delta(\lambda)^*$.
Using the property of Haar measure (see Ref.~\cite[Eq. (16)]{nikshych2002structure})
\begin{equation}
    \sum_{(x)}x^{\cone} \Lambda(y x^{\ctwo})= \sum_{(y)} S(y^{\cone}) \Lambda(y^{\ctwo}x),
\end{equation}
and cocommutativity of $\Lambda$, we obtain (ommiting the summation symbol for comultipication)
\begin{equation}
\begin{aligned}
     \eta^{\cone} \Lambda(\eta^{\ctwo}\lambda^{\ctwo}) =  &\eta^{\cone} \Lambda(\lambda^{\ctwo}\eta^{\ctwo}) = S(\lambda^{\ctwo}) \Lambda(\eta \lambda^{(3)})\\
     =&\varepsilon(\lambda^{\cthree})S(\lambda^{\ctwo}) \Lambda(\eta) =S(\lambda^{\ctwo}) \frac{1}{|\cA|}.
\end{aligned}
\end{equation}
Then we have
\begin{equation}
\begin{aligned}
        \langle \Phi_{ij}|\Gamma_{kl}\rangle=&\sum_{(\lambda),(\eta)} d_{\Gamma}d_{\Phi} \Phi_{ji}(\eta^{(1)}) \Gamma(\lambda^{(1)}) \Lambda((\eta^{\ctwo})\lambda^{\ctwo})\\
        =& \sum_{(\lambda)} |\cA|\sqrt{d_{\Gamma}d_{\Phi} } \Phi_{ji}(S(\lambda^{\ctwo})) \Gamma_{kl}(\lambda^{\cone}) \frac{1}{|\cA|}\\
        =&   \sqrt{d_{\Gamma}d_{\Phi} }  \frac{\delta_{\Phi,\Gamma} \delta_{ik}\delta_{jl} }{d_{\Gamma}}\\
        =& \delta_{\Phi,\Gamma} \delta_{ik}\delta_{jl},
\end{aligned}
\end{equation}
where we have used the orthogonality relation in Proposition~\ref{prop:orthgonality}.
\end{proof}

If we choose $\Gamma$ as trivial representation $\Gamma=\mathbb{1}$, then $d_{\Gamma}=1$ and $\Gamma(g)=\varepsilon(g)$, this implies that 
\begin{equation}
    |\mathbb{1}\rangle = {\sqrt{|\cA|}} |\lambda\rangle.
\end{equation}
We introduce the following generalized Pauli-X operators
\begin{align}
\overset{\rightarrow}{X}_{g}: |h\rangle \mapsto |gh\rangle,\quad 
\overset{\leftarrow}{X}_{g}: |h\rangle \mapsto |hS(g)\rangle. 
\end{align}
They are left and right regular actions of $\cA$ on itself. We can also introduce the right regular action $\tilde{\XR}_g$ and $\tilde{\XL}_g$ in a similar fashion:
\begin{equation}
    \tilde{\XR}_g:|h\rangle \to |S(g)h\rangle,\quad  \tilde{\XL}_g:|h\rangle \to |hg\rangle.
\end{equation}
Recall that in Refs.~\cite{Kitaev2003,Buerschaper2013a,chen2021ribbon,Jia2023weak}, these operators are denoted as $L_{\pm}^g,\tilde{L}_{\pm}^g$.
From the definition of Haar integral and $S(\lambda)=\lambda$, it's clear that 
\begin{align}
    \XR_g|\mathbb{1}\rangle= \varepsilon(g)|\mathbb{1}\rangle
    =\XL_g|\mathbb{1}\rangle, \quad \tilde{\XR}_g|\mathbb{1}\rangle= \varepsilon(g)|\mathbb{1}\rangle
    =\tilde{\XL}_g|\mathbb{1}\rangle,
\end{align}
If we set $X=\XR_{\lambda}=\XL_{\lambda}$ (notice that $\lambda=S(\lambda)$, from which it's easy to show $\XR_{\lambda}=\XL_{\lambda}$), we have
\begin{equation}
  X|\mathbb{1}\rangle=|\mathbb{1}\rangle.  
\end{equation}
Similarly, we can set $\tilde{X}=\tilde{\XR}_{\lambda}=\tilde{\XL}_{\lambda}$, and it's clear $X=\tilde{X}$.
From $\lambda^*=\lambda$ and $\XR_{\lambda}^{\dagger}=\XR_{\lambda^*}$, we see that $X$ is Hermitian. From $\lambda^2=\lambda$, we see that $X$ is a projector.
The trivial representations state plays the same role as that of Pauli-X $+1$ eigenstate $|+\rangle$, $X|+\rangle=|+\rangle$.

\subsection{Non-invertible symmetries: Hopf symmetry and fusion categorical symmetry}

For a quantum system governed by the Hamiltonian operator $H:\mathcal{H}\to \mathcal{H}$, it exhibits a (group) $G$-symmetry if the Hilbert space $\mathcal{H}$ functions as a $\Cbb[G]$-module, and the representations $U_g$ satisfy $[U_g,H]=0$ for all $g\in G$. From the definition of representation, we see that all $U_g$ must be invertible and $U_g^{-1}=U_{g^{-1}}$. In this context, `invertible symmetry' refers to the group symmetry. 
The representation is termed unitary when $U_g^{\dagger}=U_g^{-1}$. It's worth noting that all finite group symmetries are equivalent to unitary symmetries. However, for continuous groups, while they are invertible, they are not necessarily unitary in general.
If the ground state space of $H$ is non-degenerate, the unitarity of $U_g$ implies $U_g|\psi_{\rm GS}\rangle=|\psi_{\rm GS}\rangle$.
Therefore, for elements in group algebra $x\in \Cbb[G]$, $[U_x,H]=0$ and $U_x|\psi_{\rm GS}\rangle=(\sum_{g\in G} x_gU_g)|\psi_{\rm GS}\rangle=\varepsilon(x) |\psi_{\rm GS}\rangle$.

Non-invertible symmetry is defined as a symmetry characterized by an algebraic structure (e.g., ring, Hopf algebra, comodule algebra, category) where $U_g$ is not generally invertible \cite{cordova2022snowmass,brennan2023introduction,mcgreevy2023generalized,luo2023lecture,shao2024whats,SchaferNameki2024ICTP,Bhardwaj2024lecture}.
Hopf symmetry is a typical example of non-invertible symmetry \cite{bais2003hopf,jia2023boundary,Jia2023weak,jia2024weakTube}. A thorough discussion of Hopf symmetry for Hopf cluster states and its SPT phases will be provided elsewhere \cite{jia2024cluster}.
See Sec.~\ref{subsec:stabilizer} and Sec.~\ref{sec:Hopf1d} for brief discussions.
In this part, we will only provide a discussion from the perspective of a single Hopf qudit.
For a Hopf algebra $\mathcal{A}$, a quantum system that exhibits a $\mathcal{A}$ Hopf symmetry must form a representation $\mathcal{V}$ of $\mathcal{A}$, denoted by $\Gamma:\mathcal{A}\to \operatorname{End}(\mathcal{V})$.
It's notable that $\Phi(h)$ is typically non-invertible. Unlike group algebra, it's often impossible to identify a set of invertible basis elements for a general Hopf algebra.
However, we can still define the symmetry as a representation $U_x$ such that $[U_x,H]=0$ for all $x\in \cA$.
This non-invertibility implies that $U_x|\psi_{\rm GS}\rangle\neq |\psi_{\rm GS}\rangle$ generally.
A state $|\psi\rangle\in \mathcal{V}$ is called Hopf invariant if \cite{bais2003hopf,Jia2023weak}
\begin{equation}
    \Phi(h)|\psi\rangle=\varepsilon(h)|\psi\rangle, \forall h\in \cA.
\end{equation}

For Hopf qudit, if we define the Hamiltonian as 
\begin{equation}\label{eq:HopfQuditH}
    H_X=-X=-\XR_{\lambda}=-\XL_{\lambda}=-\tilde{\XR}_{\lambda}=-\tilde{\XL}_{\lambda}.
\end{equation}
It's clear that $[\XR_g,H]=0$ for all $g\in \cA$, which means that $H$ has a non-invertible $\cA$-symmetry. Similarly, we have $\cA$-symmetry $\XL_g$, and two $\cA^{\rm op}$-symmetries: $\tilde{\XR}_g$ and $\tilde{\XL}_g$ , where $\cA^{\rm op}$ means $g\cdot^{\rm op}h:=h\cdot g$.

\begin{proposition}
The Hopf qudit Hamiltonian $H_X$ has two Hopf $\cA$-symmetries $\XL_g,\XR_g$ if we treat $H_X=-\XR_{\lambda}=-\XL_{\lambda}$, and it has two $\cA^{op}$-symmetries $\tilde{\XL}_g,\tilde{\XR}_g$ if we treat $H_X=-\tilde{\XR}_{\lambda}=-\tilde{\XL}_{\lambda}$.
The trivial representation state $|\one\rangle$ is a ground state of Hamiltonian in Eq.~\eqref{eq:HopfQuditH}, and we have:

(i)    The trivial representation state $|\mathbb{1}\rangle$ is Hopf invariant under Hopf $\cA$-symmetry $\XR_h$, $\XL_h$:
    \begin{equation}
        \XR_h|\one\rangle=\varepsilon(h)|\one\rangle, \quad \XL_h|\one\rangle= \varepsilon(h)|\one\rangle.
    \end{equation}

 (ii)  The state $|\one\rangle$ is  also Hopf invariant under Hopf $\cA^{\rm op}$-symmetry $\tilde{\XR}_h$, $\tilde{\XL}_h$. It's clear that $\tilde{\XR}_g \tilde{\XR}_h=\tilde{\XR}_{g\cdot^{\rm op}h}$, and  $\tilde{\XL}_g \tilde{\XL}_h=\tilde{\XL}_{g\cdot^{\rm op}h}$. For state $|\one\rangle$, we have
        \begin{equation}
        \tilde{\XR}_h|\one\rangle=\varepsilon(h)|\one\rangle, \quad \tilde{\XL}_h|\one\rangle= \varepsilon(h)|\one\rangle.
    \end{equation}
\end{proposition}

\begin{proof}
    This is a direct result of the fact that, for Haar integral $\lambda$, $S(\lambda)=\lambda$ and  $x\lambda=\lambda x =\varepsilon(x)\lambda$.
\end{proof}

As we will see later, the Hopf cluster state has Hopf symmetries (Sec.~\ref{subsec:stabilizer} and Sec.~\ref{sec:Hopf1d}).
It's also worth mentioning that the ground state of the Hopf quantum double model also exhibits this type of Hopf symmetry \cite{jia2023boundary,Jia2023weak}.
Also note that if we choose $h \in \operatorname{Comm}(\cA)$ as an element in the center of $\cA$ and set $H = -\XR_h$, we see that $H$ will have Hopf $\cA$-symmetry given by $\XR_g$. This indicates that non-invertible Hopf symmetry for Hopf qudit system is a common phenomenon.

Another kind of non-invertible symmetry is the categorical symmetry, which is usually characterized by the unitary fusion category \cite{Frohlich2004kramers,fuchs2007topological,frohlich2010defect,bhardwaj2018finite,chang2019topological,thorngren2019fusion,thorngren2021fusion,komargodski2021symmetries,inamura2022lattice,inamura2023fermionization}.
The Hopf qudit possesses a categorical symmetry described by $\Rep(\cA)$, for which we have the fusion rule
\begin{equation}
    \Phi \otimes \Gamma =\oplus_{\Psi\in \Irr(\cA)} N_{\Phi\Gamma}^{\Psi} \Psi,
\end{equation}
where $N_{\Phi\Gamma}^{\Psi}\in \mathbb{Z}_{\geq 0}$ is called fusion multiplicity.
The symmetry operators $U_{\Gamma}$ are labeled by objects $\Gamma\in \Rep(\cA)$, their composition law is give by the fusion rule
\begin{equation}
    U_{\Gamma}U_{\Phi}=U_{\Gamma\otimes \Phi}.
\end{equation}
We can similarly introduce the $\Rep(\cA)^{\rm rev}$ symmetry, where $\Gamma\otimes^{\rm rev} \Phi:=\Phi\otimes \Gamma$.

\begin{proposition}\label{prop:RepSym}
If we define the Hamiltonian for Hopf qudit as
\begin{equation}
    H=-Z=-\tilde{Z}.
\end{equation}
The Hopf qudit have non-invertible $\Rep(\cA)$ and $\Rep^{\rm rev}(\cA)$ symmetries.

(i) 
If we treat the Hamiltonian for Hopf qudit as
\begin{equation}
    H_Z=-Z=-\sum_{\Gamma\in \Irr(\cA)} \frac{d_{\Gamma}}{|\cA|} (\Tr' Z_{\Gamma}),.
\end{equation}
The Hopf qudit system will have a non-invertible symmetry characterized by the category $\Rep(\cA)$.
The symmetry operator is given by $J_{\Gamma} = \Tr' Z_{\Gamma}$, and we have 
\begin{equation}
    J_{\Gamma} \circ J_{\Phi}=
    J_{\Gamma\otimes \Phi}
    =J_{\oplus_{\Psi\in \Irr(\cA)} N_{\Phi\Gamma}^{\Psi} \Psi}.
\end{equation} 
For ground state $|1_{\cA}\rangle$, 
\begin{equation}
    J_{\Gamma}|1_{\cA}\rangle=d_{\Gamma}|1_{\cA}\rangle.
\end{equation}
If we define $J_{\Gamma}^{\ddagger}=\Tr' Z_{\Gamma}^{\ddagger}$, we have 
\begin{equation}
    J_{\Gamma}^{\ddagger}J_{\Phi}^{\ddagger}=J^{\ddagger}_{\Gamma\otimes \Phi}=J^{\ddagger}_{\oplus_{\Psi\in \Irr(\cA)}N_{\Gamma \Phi}^{\Psi} \Psi}.
\end{equation}
This also gives a $\Rep(\cA)$ symmetry and we also have 
\begin{equation}
   J_{\Gamma}^{\ddagger}|1_{\cA}\rangle=d_{\Gamma}|1_{\cA}\rangle. 
\end{equation}

(ii) 
If we treat the Hamiltonian for Hopf qudit as
\begin{equation}
    H_{\tilde{Z}}=-\tilde{Z}=-\sum_{\Gamma\in \Irr(\cA)} \frac{d_{\Gamma}}{|\cA|} (\Tr' \tilde{Z}_{\Gamma}),.
\end{equation}
The Hopf qudit system will have a non-invertible symmetry characterized by the category $\Rep(\cA)^{\rm rev}$, where $\Gamma\otimes^{\rm rev} \Phi:=\Phi\otimes \Gamma$.
The symmetry is realized by $E_{\Gamma}=\Tr' \tilde{Z}_{\Gamma}$ and $E_{\Gamma}^{\ddagger}=\Tr' \tilde{Z}^{\ddagger}_{\Gamma}$.
We have
\begin{align}
    E_{\Gamma}E_{\Phi}=E_{\Gamma\otimes^{\rm rev}\Phi}=E_{\oplus_{\Psi\in \Irr(\cA)}N_{ \Phi \Gamma}^{\Psi} \Psi}\\
    E_{\Gamma}^{\ddagger}E_{\Phi}^{\ddagger}=E_{\Gamma\otimes^{\rm rev}\Phi}^{\ddagger}=E_{\oplus_{\Psi\in \Irr(\cA)}N_{ \Phi \Gamma}^{\Psi} \Psi}^{\ddagger},
\end{align}
and for state $|1_{\cA}\rangle$, 
\begin{equation}
E_{\Gamma}|1_{\cA}\rangle=d_{\Gamma}|1_{\cA}\rangle=E_{\Gamma}^{\ddagger}|1_{\cA}\rangle.
\end{equation}
\end{proposition}

\begin{proof}
To prove this, recall that $Z_{\Gamma}=(\text{id}\otimes \Gamma)\circ \Delta$, this means that
\begin{equation}
    J_{\Gamma} = (\id \otimes \Tr )\circ (\text{id}\otimes \Gamma)\circ \Delta,
\end{equation}
viz., $ J_{\Gamma}$ is map form Hopf qudit $\cA$ to itself. Then for $h\in \cA$, we have
\begin{equation}
    \begin{aligned}
         J_{\Gamma}\circ J_{\Phi}|h\rangle= &\sum_{(h)} J_{\Gamma} |h^{\cone} \rangle \otimes \Tr \Phi(h^{\ctwo})\\
         =& \sum_{(h)} |h^{\cone} \rangle \otimes \Tr \Gamma(h^{\ctwo}) \otimes \Tr \Phi(h^{\cthree})\\
         =& \sum_{(h)} |h^{\cone} \otimes \Tr [(\Gamma\otimes \Phi)(h^{\ctwo})]\\
         = &  \sum_{(h)} |h^{\cone} \otimes \Tr[ (\oplus_{\Psi\in \Irr(\cA)} N_{\Phi\Gamma}^{\Psi} \Psi)(h^{\ctwo})]\\
         =&J_{\oplus_{\Psi\in \Irr(\cA)} N_{\Phi\Gamma}^{\Psi} \Psi} |h\rangle.
    \end{aligned}
\end{equation}
Notice that we have used the definition of tensor product for two representations given in Eq.~\eqref{eq:RepTensor}. A more intuitive proof based Hopf tensor network will be given in Sec.~\ref{sec:HopfTensor}. All other equations can be proven analogously.

We also need to show that $[J_{\Gamma},H_Z]=0$. This is clear from the fact that $J_{\Gamma}=T^{\chi_{\Gamma}}_+$ and $H_Z=T^{\Lambda}_+$, with $\chi_{\Gamma}$ being the character of $\Gamma$ and $\Lambda$ the Haar measure of $\cA$. See Sec.~\ref{sec:QD} for a more detailed discussion.

All other claims can be proved in similar ways.
\end{proof}

Notice that if we use the characters $\chi_{\Gamma}$ to represent the categorical symmetry, we see that the $\Rep(\cA)$ symmetry can be regarded as a subalgebra of the dual Hopf algebra $\bar{\cA}$.
The Hopf qudit $\cA$ can be regarded as a left $\bar{\cA}$-module with the left action given by $T^{\varphi}_+$, viz., we have
\begin{equation}
    T^{\bar{1}}_+=\id_{\cA},\quad T^{\varphi \psi}_+=T^{\varphi}_+T_{+}^{\psi},
\end{equation}
where $\bar{1}=\varepsilon$ is the unit of $\bar{\cA}$. Using Eqs.~\eqref{eq:JT1} and \eqref{eq:JT2}, we see that the $\Rep(\cA)$ and $\Rep(\cA)^{\rm rev}$ symmetries discussed above can be embedded into $\bar{\cA}$.

\section{CSS-type Hopf Cluster states}
\label{sec:HopfCSSCluster}

With the preparation outlined above, in this section, we will construct the CSS-type Hopf cluster state, which we will simply refer to as the Hopf cluster state. Our construction has the property that it can recover the finite group cluster state \cite{brell2015generalized,fechisin2023noninvertible} when we set the Hopf algebra as the group algebra.

Indeed, our construction relies on a bipartite graph $K = (V, E)$, where each odd vertex has a local ordering of edges connected to it. We term this type of graph a cluster graph. Below is a formal definition (See Fig.~\ref{fig:cluster-graph} for a depiction):

\begin{definition}[Cluster graph]\label{def:ClusterGraph}
A graph $K(V,E)$ is called a cluster graph if (i) it is a bipartite graph; (ii) each edge is oriented; and (iii) there is global ordering of edge set, denoted as $e_1,\cdots,e_{|E|}$.
\end{definition}

\begin{figure}[t]
    \centering
    \includegraphics[width=5cm]{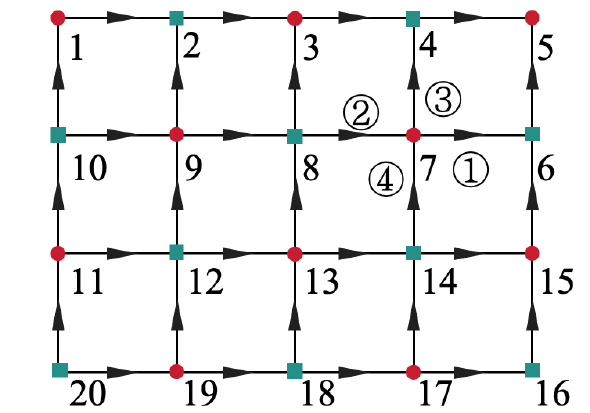}
    \caption{An illustration of cluster graph. The red vertices represent odd vertices, while the teal vertices represent even vertices. The vertices are labeled with the ordering from $1$ to $20$, and examples of local ordering for edges connected to odd vertex $v_7$ are shown as \textcircled{1}, $\cdots$, and \textcircled{4}.}
    \label{fig:cluster-graph}
\end{figure}

The condition (i) implies that we can partition the vertices into even and odd vertices, denoted as $V_e$ and $V_o$, respectively. Furthermore, there is no edge connecting odd vertices with odd vertices, and similarly for even vertices.
This means that the edge set can be divided into disjoint union $E=\sqcup_{v_o\in V_o} N_{E}(v_o)$, with $N_{E}(v_o)$ the set of edges connected to odd vertex $v_o$.

The global ordering of the edge set in condition (iii) is equivalent to the existence of local orderings for all vertices. Notice that for the group case, we only need to have a local ordering for the edge set $N_E(v_e)$ for all even vertices \cite{brell2015generalized}, since the group algebra is cocommutative.
For the Hopf qudit, the local ordering of $N_E(v_o)$  plays a crucial role. To ensure a finite-depth circuit, we can choose a special configuration of edge orientation for $N_E(v_e)$ at each even vertex $v_e$. This will be illustrated in Sec.~\ref{sec:ClusterLattice}.

\subsection{Vertex states and edge entangler operations}

For a cluster graph $K(V,E)$, we assign Hopf qudit $\mathcal{H}_v=\cA$ to each vertex $v\in V$ and the total space is $\otimes_{v\in V}\mathcal{H}_v$. This kind of Hopf qudit network also plays a crucial role in the Hopf quantum double model \cite{Kitaev2003,Buerschaper2013a,buerschaper2013electric,jia2023boundary,Jia2023weak}.
The initial state of the cluster state is a product of Pauli-X preferred states and Pauli-Z preferred states. This concept can be generalized to Hopf qudits as follows:

\begin{definition}
    We introduce two preferred states in the Hopf algebra $\cA$, one is unit element $|1_{\cA}\rangle$  call Z-preferred state;  the other is trivial representation state $|\mathbb{1}\rangle$ called X-preferred state.
\end{definition}

As we have proved, the Z and X preferred states satisfy the following properties:
\begin{enumerate}
        \item For all irreps $\Gamma\in \Rep(\cA)$, we have  $Z_{\Gamma^{ii}}|1_{\cA}\rangle = |1_{\cA}\rangle$,         $J_{\Gamma}|1_{\cA}\rangle=d_{\Gamma}|1_{\cA}\rangle$, $Z|1_{\cA}\rangle=|1_{\cA}\rangle$.
        \item $X$-preferred state $|\mathbb{1}\rangle$ satisfies $\XR_g|\one\rangle=\varepsilon(g)|\one\rangle=\XL_g|\one\rangle$ and $X|\one\rangle=|\one\rangle$.
    \end{enumerate}

    The first statement comes from the fact $\Delta(1_{\cA})=1_{\cA}\otimes 1_{\cA}$, while the second is a property of the Haar integral. For a cluster graph $K$, we assign $X$-preferred state $|\one\rangle$ to each odd vertex and $X$-preferred state  $|1_{\cA}\rangle$ to each even vertex: 
\begin{equation}
    |\Omega\rangle=(\otimes_{v_o}|\one\rangle_{v_o}) \otimes (\otimes_{v_e}|1_{\cA}\rangle_{v_e}),
\end{equation}
which will be called the preferred product state.

%\subsection{Edge entangler operations}

The edge entangler for a finite group cluster state manifests as controlled gates, inherently dependent on the chosen basis. At first glance, it might seem plausible to generalize the edge entangler solely using the algebraic structure of $\mathcal{A}$ and a selected basis, without resorting to the coalgebra structure and antipode morphism. However, such an approach is not feasible. The reasoning behind this limitation is as follows.
Choosing a basis $\{1_{\cA}=g_0,g_1,\cdots,g_{d-1}\}$ of the $d$ dimensional Hopf algebra $\cA$, using the regular $\cA$ action on $\cA$, we could define
\begin{equation}
    \overset{\rightarrow}{X}_{g_k}: |h\rangle \mapsto |g_kh\rangle,\quad  k=0,\cdots,d-1.
\end{equation}
However, unlike in the group case, the basis elements of a Hopf algebra are generally not invertible. This noninvertibility poses several challenges.
We cannot define $\overset{\leftarrow}{X}_{g_k}$ solely using the algebraic structure of $\mathcal{A}$, as there may not be an inverse $g_k^{-1}$. To address this, we must utilize the antipode map $S: \mathcal{A} \to \mathcal{A}$ and define
\begin{equation}
\overset{\leftarrow}{X}_{g_k}: |h\rangle \mapsto |hS(g_k)\rangle,\quad k=0,\cdots,d-1. \label{eq:XLhopf}
\end{equation}
An even more significant issue arises when defining the edge entangler operation as
\begin{equation*}
C \overset{\rightarrow}{X} =\sum_{k,l=0}^{d-1} |g_k\rangle \langle g_k|\otimes |g_kg_l\rangle \langle g_l|,
\end{equation*}
as this map is generally non-invertible. A similar issue exists for $C \overset{\leftarrow}{X}$.

Drawing from the aforementioned observations, we define the edge entangler by leveraging the comprehensive structure of a Hopf algebra. This includes not only the algebraic structure (as a regular action) but also the coalgebra structure and the antipode map.
For the left action, the edge entangler is defined as follows:
\begin{equation}
    C\overset{\rightarrow}{X}_{i,j}: |g\rangle_i|h\rangle_j\mapsto \sum_{(g)}|g^{(2)}\rangle_i |g^{(1)}h\rangle_j,
\end{equation}
which is invertible, and its inverse is given by:
\begin{equation}
     C\overset{\rightarrow}{X}_{i,j}^{-1}: |g\rangle_i|h\rangle_j\mapsto \sum_{(g)}|g^{(2)}\rangle_i |S(g^{(1)})h\rangle_j.
\end{equation}
Notice that $C\overset{\rightarrow}{X}:\cA\otimes \cA \to \cA\otimes \cA$ can be expressed as a composition of structure maps of Hopf algebra
\begin{equation}
  C\overset{\rightarrow}{X} = (\text{id} \otimes \mu)\circ (\tau \otimes \text{id})  \circ (\Delta\otimes \text{id}),  
\end{equation}
where $\tau(a\otimes b)=b\otimes a$ is the twist map. All other entangler operations can be expressed in a similar form as well.
We could also define edge entangler from the right regular action of Hopf algebra on itself as follows:
\begin{equation}
    C\overset{\leftarrow}{X}_{i,j}: |g\rangle_i|h\rangle_j\mapsto \sum_{(g)}|g^{(2)}\rangle_i |hS(g^{(1)})\rangle_j,
\end{equation}
and its inverse is given by:
\begin{equation}
     C\overset{\leftarrow}{X}_{i,j}^{-1}: |g\rangle_i|h\rangle_j\mapsto \sum_{(g)}|g^{(2)}\rangle_i |hg^{(1)}\rangle_j.
\end{equation}
Observe that by setting the Hopf algebra as the group algebra $\mathcal{A} = \mathbb{C}[G]$, the aforementioned definitions align seamlessly with finite group case \cite{brell2015generalized,fechisin2023noninvertible}.

To show that $C\overset{\rightarrow}{X}_{i,j}^{-1} \circ C\overset{\rightarrow}{X}_{i,j}= \id_{\cA}\otimes \id_{\cA}$, we need to use the fact $S^2=\id_{\cA}$ when $\cA$ is semisimple and cosemisimple \cite{larson1988semisimple}, this further implies $S(g^{(2)})g^{(1)}=\varepsilon(g)1_{\cA}$. With this and the Hopf algebra axiom $\sum_{(g)} g^{(2)} \varepsilon(g^{(1)})=g$ we have:
\begin{equation}
\begin{aligned}
    |g\rangle_i |h\rangle_j \overset{C\overset{\rightarrow}{X}_{i,j}}{\mapsto}  \sum_{(g)}|g^{(2)}\rangle_i |g^{(1)}h\rangle_j 
   & \overset{C\overset{\rightarrow}{X}^{-1}_{i,j}}{\mapsto} \sum_{(g)}|g^{(3)}\rangle_i |S(g^{(2)})g^{(1)}h\rangle_j\\
    =&\sum_{(g)}|g^{(2)}\rangle_i |\varepsilon(g^{(1)}) 1_{\cA}h\rangle_j = |g\rangle_i |h\rangle_j.
\end{aligned}
\end{equation}
To establish  $C\overset{\leftarrow}{X}_{i,j}^{-1}  \circ C\overset{\leftarrow}{X}_{i,j}= \id_{\cA}\otimes \id_{\cA}$, we only require the following Hopf algebra axioms: $S(g^{(1)})g^{(2)}=\varepsilon(g)1_{\cA}$ and $\sum_{(g)} g^{(2)} \varepsilon(g^{(1)})=g$.

\begin{remark}
Notice that the above construction does not hold for weak Hopf algebras \cite{BOHM1998weak}, as $\sum_{(g)}g^{(1)}S(g^{(2)})\neq \varepsilon(g)$ in this case. However, since we won't use the inevitability of local gates, our construction of Hopf cluster state also works for weak Hopf case.
A better way of generalization to the weak Hopf cluster state requires using the equivalence between the cluster state model and the quantum double model, which will be established in Sec.~\ref{sec:QD}.
\end{remark}

\begin{table}
    \centering
    \begin{tabular}{|c|c|}
    \hline
Total space &  $\mathcal{H}_{tot}=\otimes_{v\in V}\mathcal{H}_v$, $\mathcal{H}_v=\cA$   \\
    \hline
    Z-preferred state  & $|1_{\cA}\rangle$, $Z|1_{\cA}\rangle=|1_{\cA}\rangle$  \\     
        \hline
        X-preferred state & $ |\one\rangle=\sqrt{|\cA|} |\lambda\rangle$, $X|\one\rangle=|\one\rangle$  \\
          \hline   
Preferred product state & $ |\Omega\rangle =(\otimes_{v_o} |\lambda\rangle_{v_o})\otimes (\otimes_{v_e} |1_{\cA}\rangle_{v_e})$\\
\hline
Edge entangler  &  $\begin{aligned}
    C\overset{\rightarrow}{X}_{i,j}: |g\rangle_i|h\rangle_j\mapsto \sum_{(g)}|g^{(2)}\rangle_i |g^{(1)}h\rangle_j\\
     C\overset{\rightarrow}{X}_{i,j}^{-1}: |g\rangle_i|h\rangle_j\mapsto \sum_{(g)}|g^{(2)}\rangle_i |S(g^{(1)})h\rangle_j\\
       C\overset{\leftarrow}{X}_{i,j}: |g\rangle_i|h\rangle_j\mapsto \sum_{(g)}|g^{(2)}\rangle_i |hS(g^{(1)})\rangle_j\\
        C\overset{\leftarrow}{X}_{i,j}^{-1}: |g\rangle_i|h\rangle_j\mapsto \sum_{(g)}|g^{(2)}\rangle_i |hg^{(1)}\rangle_j
\end{aligned}$ \\
           \hline           
    \end{tabular}
    \caption{Summary of preferred states and edge entangler operations for Hopf cluster state.}
    \label{tab:HopfCluster}
\end{table}

\subsection{CSS-type Hopf cluster state}

Now, let's construct the CSS-type Hopf cluster state.
The input state of the cluster circuit is the preferred product state
\begin{equation}
    |\Omega\rangle =(\otimes_{v_o} |\one\rangle_{v_o})\otimes (\otimes_{v_e} |1_{\cA}\rangle_{v_e}).
\end{equation}
Recall that we have assumed a global ordering of the edge set. Denoting them with subscripts, we have $e_1, \cdots, e_{|E|}$.
Since each edge is directed and the cluster graph is bipartite, there are two possible configurations of the edges:
\begin{itemize}
\item If the edge is directed from an odd vertex to an even vertex, we set $U_e=C\XR$:
\begin{equation}
     \begin{aligned}
        \begin{tikzpicture}        
            \draw[line width=.6pt,black,-latex] (0,0)--(1,0);
            \draw[line width=.6pt,black,] (2,0)--(0.8,0);
            \draw[black,fill=red] (0,0) circle (0.15);  
             \draw[draw=black,fill=teal] (1.85,-0.15) rectangle ++(0.3,0.3);
            \node[ line width=0.6pt, dashed, draw opacity=0.5] (a) at (0.8,0.4){$C\XR$};
        \end{tikzpicture}
    \end{aligned}
\end{equation}
\item If the edge is directed from an even vertex to an odd vertex, we set $U_e=C\XL$:
\begin{equation}
     \begin{aligned}
        \begin{tikzpicture}        
            \draw[line width=.6pt,black,-latex] (2,0)--(1,0);
            \draw[line width=.6pt,black,] (1.2,0)--(0,0);
            \draw[black,fill=red] (0,0) circle (0.15);  
             \draw[draw=black,fill=teal] (1.85,-0.15) rectangle ++(0.3,0.3);
            \node[ line width=0.6pt, dashed, draw opacity=0.5] (a) at (0.8,0.4){$C\XL$};
        \end{tikzpicture}
    \end{aligned}
\end{equation}
\end{itemize}
In this manner, the CSS-type Hopf cluster state for a given cluster graph $K$ and Hopf algebra $\mathcal{A}$ is defined as
\begin{equation}\label{eq:HopfCluster}
|K,\mathcal{A}\rangle= U_{e_{|E|}}\cdots U_{e_1} |\Omega\rangle.
\end{equation}
We will also call $U_E= U_{e_{|E|}}\cdots U_{e_1}$ (where $E$ represents the edge set of graph $K$) the circuit operations and $ |\Omega\rangle$ the initial state.

\begin{proposition} \label{prop:Uedge}
For the edge entangler operations, we have the following observations:
\begin{enumerate}
       \item  Notice that if the distance between two edges $e,e'$ is greater or equal to two, viz., they are not connected via some vertex, then the edge entangler operations are commutative: $[U_e,U_{e'}]$.
       
       \item For each odd vertex $v_o$ with input state $h$, the local ordering $e_1,\cdots,e_n$ yields 
       \begin{equation}
       \sum_{(h)}   |h^{(n+1)}\rangle_{v_o} \otimes X^{\leftrightarrow}_{h^{(n)}}\otimes \cdots \otimes X^{\leftrightarrow}_{h^{(1)}},
       \end{equation}
       where $X^{\leftrightarrow}$ is chosen as $\XR$ or $\XL$  depending on the orientations of each edge.
       When $\cA$ is cocommutative ($\Delta=\tau\circ \Delta$), these edge operators will commute with each other, as is the case for the group algebra $\Cbb[G]$. In this scenario, there is no need for a local ordering of edges for each odd vertex.
       \item For each even vertex $v_e$, if $e$ is an input edge and $e'$ is an output edge, 
      \begin{equation*}
      \begin{aligned}
        \begin{tikzpicture}        
            \draw[line width=.6pt,black,-latex] (0,0)--(1,0);
            \draw[line width=.6pt,black,] (2,0)--(0.8,0);
            \draw[line width=.6pt,black,-latex] (2,0)--(3,0);
            \draw[line width=.6pt,black,] (2.8,0)--(4,0);
            \draw[black,fill=red] (0,0) circle (0.15);  
            \draw[black,fill=red] (4,0) circle (0.15);  
             \draw[draw=black,fill=teal] (1.85,-0.15) rectangle ++(0.3,0.3);
            \node[ line width=0.6pt, dashed, draw opacity=0.5] (a) at (0.8,0.4){$U_{e}$};
            \node[ line width=0.6pt, dashed, draw opacity=0.5] (a) at (3.2,0.4){$U_{e'}$};
        \end{tikzpicture}
    \end{aligned}
    \end{equation*}
       then $[U_{e},U_{e'}]$ vanishes.
For two input edges (or output edges), the edge entangler operators do not commute in general.
\end{enumerate}
\end{proposition}

\begin{proof}
   1. This is because for $d(e,e')\geq 2$, $U_e$ and $U_{e'}$ have support on different Hopf qudits.
   2. It's a direct consequence of the definition of edge entangler operations.
   3. Since $[\XR_v,\XL_w]=0$ for any $w,v\in V$, we have $[C\XL_e,C\XR_{e'}]=0$.
\end{proof}
The third statement of Proposition~\ref{prop:Uedge} is crucial: if the even vertex $v_o$ has more than two bonds ($|N_E(v_e)| > 2$), the commutativity cannot be reached in general. There must be a local ordering to make the definition of the Hopf cluster state well-defined.

\begin{proposition}
    For a general cluster graph, the circuit operator $U_E= U_{e_{|E|}}\cdots U_{e_1}$ in general cannot be decomposed into the composition of some commutative local gates.
\end{proposition}

\begin{proof}
   Notice that, to divide the edge set into some disjoint union of subsets $E_1\cup\cdots \cup E_m$ such that $[U_{E_{i}},U_{E_j}]=0$, we can only divide the edge set along the even vertex, since the edges connected to an odd vertex are not commutative in general. Based on the third statement of Proposition~\ref{prop:Uedge}, we must divide $N_E(v_e)$ into the input edge set $N_E^i(v_e)$ and the output edge set $N_E^o(v_e)$, and then cut the edge set of the graph along the even vertex in this way to satisfy the locality condition. However, when $|N_E^i(v_e)|\geq 2$, this will introduce more than two odd vertices, all the edges connected to these vertices must be put into the same set since they are not commutative. These odd vertices will introduce more even vertices that connect to them. Repeating the dividing along these new even vertices will introduce some other odd vertices, eventually, the size of the local set can reach the size of the total edge set $E$.
\end{proof}

Although we cannot have locality condition for the general graph, when the bonds of all even vertices are two and $|N_E^i(v_e)|=|N_E^o(v_e)|=1$, we can make the local stabilizers commute with each other.
The details will be illustrated in Sec.~\ref{sec:ClusterLattice}, which is crucial for discussing SPT phases of the Hopf cluster states \cite{jia2024cluster}.

\begin{remark}
 Notice that our construction of Hopf cluster state is based on controlled left regular actions. We can introduce the controlled right regular actions
    \begin{equation}
        C\tilde{\XL}_{i,j}|h\rangle\otimes |g\rangle=|h^{\ctwo}\rangle\otimes |gh^{\cone}\rangle,\quad
        C\tilde{\XR}_{i,j}|h\rangle\otimes |g\rangle=|h^{\ctwo}\rangle\otimes |S(h^{\cone})g\rangle.
    \end{equation}
By setting the same preferred initial state $|\Omega\rangle$ and edge entangler as the above operators, we can also construct a Hopf cluster state. All the properties we discussed in this section, and will discuss in the sequel, also hold for this construction.
\end{remark}

\subsection{Stabilizers and non-invertible symmetry}
\label{subsec:stabilizer}

The stabilizer group for the qubit state is a subgroup of the Pauli group for which the code state is invariant under the action of the stabilizers. Here we will adopt a more general definition of stabilizers, which are the operators that leave the state of our interest invariant.

For the Hopf cluster state, the stabilizers can be constructed in a similar way as that for the qubit cluster state. 
Recall that for Haar integral $\lambda=S(\lambda)$ which implies $\XR_{\lambda}=\XL_{\lambda}$. Set this operator as $X=\XR_{\lambda}=\XL_{\lambda}$ we have $X|\one\rangle = |\one\rangle$.
For each odd vertex $v_o$ and $X(v_o)$, the stabilizer is of the form
\begin{equation}
    T(v_o)=    U_{e_{|E|}}\cdots  U_{e_{1}} X(v_o)  U_{e_{1}}^{-1} \cdots  U_{e_{|E|}}^{-1}.
\end{equation}
It's clear that for Hopf cluster state $|K,\cA\rangle$, we have $T(v_o) |K,\cA\rangle=|K,\cA\rangle$.
Similarly, for the even vertex $v_e$, the operator $Z (v_e)$ in Eq.~\eqref{eq:PauliZ} leaves $|1_{\mathcal{A}} \rangle_{v_e}$ invariant.
The corresponding stabilizer is defined as
\begin{equation}
    Q(v_e)=    U_{e_{|E|}}\cdots  U_{e_{1}} Z(v_e)  U_{e_{1}}^{-1} \cdots  U_{e_{|E|}}^{-1}.
\end{equation}
It's easy to verify $Q(v_e) |K,\cA\rangle=|K,\cA\rangle$.

For odd vertex, we can also introduce
\begin{equation}\label{eq:THopfright}
      \overset{\rightarrow}{T}_g(v_o)=    U_{e_{|E|}}\cdots  U_{e_{1}} \XR_g(v_o)  U_{e_{1}}^{-1} \cdots  U_{e_{|E|}}^{-1}, g\in \cA.
\end{equation}
\begin{equation}\label{eq:THopfleft}
      \overset{\leftarrow}{T}_g(v_o)=    U_{e_{|E|}}\cdots  U_{e_{1}} \XL_g(v_o)  U_{e_{1}}^{-1} \cdots  U_{e_{|E|}}^{-1}.
\end{equation}
They form representations of Hopf algebra $\cA$ in the sense that 
\begin{equation}
    \overset{\rightarrow}{T}_{1_{\cA}}=I,\quad \overset{\rightarrow}{T}_g \overset{\rightarrow}{T}_h=\overset{\rightarrow}{T}_{gh},
\end{equation}
similarly for $ \overset{\leftarrow}{T}_g(v_o)$.
From $g\lambda=\varepsilon(g)\lambda=\lambda S(g)$, it's clear that 
\begin{equation}
    \overset{\rightarrow}{T}_g|K,\cA\rangle=\varepsilon(g)|K,\cA\rangle=\overset{\leftarrow}{T}_g|K,\cA\rangle.
\end{equation}
To summarize, we have:

\begin{proposition}
    The Hopf cluster state supports two Hopf $\cA$-symmetries on each odd vertex, which are defined by Eqs.~\eqref{eq:THopfright} and \eqref{eq:THopfleft}.   Similarly, it also supports two Hopf $\cA^{\rm op}$-symmetries, if we replace the left regular action with the right regular action.
\end{proposition}

For even vertex and $\Gamma\in \Rep(\cA)$, we introduce 
\begin{equation}\label{eq:RepQ}
   Q_{\Gamma}(v_e)=   U_{e_{|E|}}\cdots  U_{e_{1}} J_{\Gamma}(v_e)  U_{e_{1}}^{-1} \cdots  U_{e_{|E|}}^{-1}.
\end{equation}
Then it's easy to check that 
\begin{equation}\label{eq:RepQd}
     Q_{\Gamma}(v_e)|K,\cA\rangle=d_{\Gamma}|K,\cA\rangle.
\end{equation}
And 
\begin{equation}
     Q_{\Gamma}(v_e) Q_{\Phi}(v_e)= Q_{\oplus_{\Psi}N_{\Gamma\Phi}^{\Psi} \Psi}(v_e).
\end{equation}
From $J_{\Gamma}^{\ddagger}$, we define
\begin{equation}\label{eq:RepQd2}
       Q_{\Gamma}^{\ddagger}(v_e)=   U_{e_{|E|}}\cdots  U_{e_{1}} J_{\Gamma}^{\ddagger}(v_e)  U_{e_{1}}^{-1} \cdots  U_{e_{|E|}}^{-1}.
\end{equation}
Similar results hold. For $E_{\Gamma}$ and $E_{\Gamma}^{\ddagger}$, we can construct two $\Rep(\cA)^{\rm rev}$-symmetries.
To summarize, we have:
\begin{proposition}
    The Hopf cluster state supports two $\Rep(\cA)$-symmetries on each even vertex, which is defined by Eqs.~\eqref{eq:RepQ} and \eqref{eq:RepQd2}. Similarly, it also supports two $\Rep(\cA)^{\rm rev}$-symmetries on each even vertex, if we replace $J_{\Gamma}$ and  $J_{\Gamma}^{\ddagger}$ with $\tilde{J}_{\Gamma}$ and  $\tilde{J}_{\Gamma}^{\ddagger}$ in Eqs.~\eqref{eq:RepQ} and \eqref{eq:RepQd2}.
\end{proposition}

\section{Hopf cluster-lattice state on $d$-dimensional spatial manifold}
\label{sec:ClusterLattice}

As discussed in the last section, to ensure a finite-depth circuit, we impose commutativity for local operations, which implies certain local structures of even vertex for the graph.  
For our convenience, we introduce the following definition:

\begin{definition}[Cluster lattice for a $d$-dimensional manifold]
For an arbitrary $d$-dimensional spatial manifold, a cluster lattice $\mathbb{M}^d$ can be constructed as follows: We begin by cellulating the manifold, and then consider its $1$-dimensional skeleton as the lattice. The vertices of this lattice are referred to as odd vertices. Subsequently, we assign an orientation to each edge of the lattice and place even vertices along these edges. See Fig.~\ref{fig:cluster-lattice-graph} for an illustration.
\end{definition}

\begin{figure}
    \centering
    \includegraphics[width=8cm]{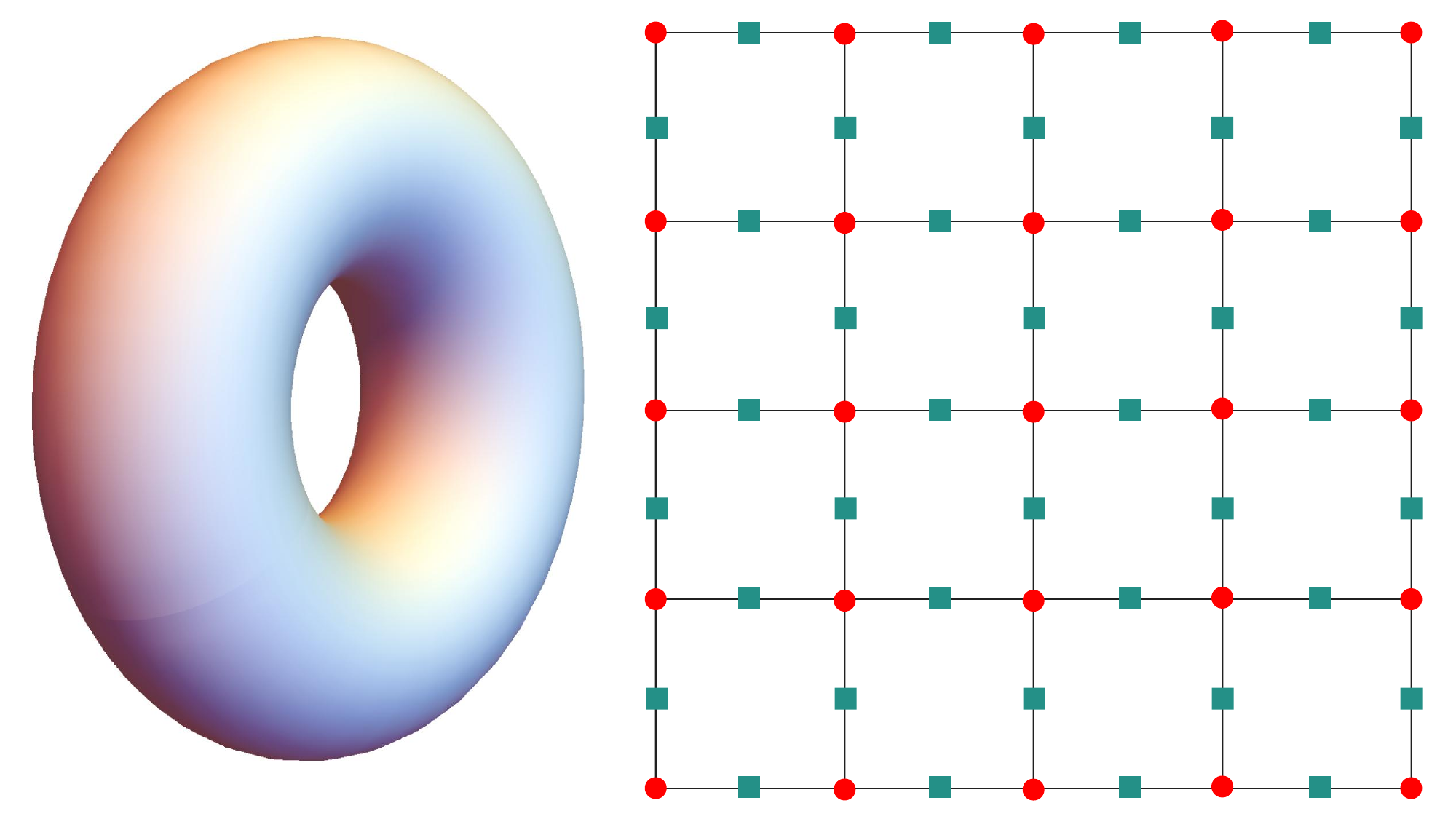}
    \caption{An example of cluster lattice on a torus $\mathbb{T}^2$. We first create a cellulation of $\mathbb{T}^2$ and regard the vertices as odd vertices (depicted as red vertices in the figure). An orientation is assigned to each edge. Then, for each edge, we add an even vertex (represented as teal vertices in the figure). These even vertices are all bivalent. The original edge now becomes two edges connected to the even vertex, one input edge, and one output edge.}
    \label{fig:cluster-lattice-graph}
\end{figure}

We will use the term `cluster lattice' to distinguish it from the `cluster graph'. Note that a cluster lattice is a special case of a cluster graph.
As a graph $\mathbb{M}^d=(V,E)$, a cluster lattice is always bipartite. The even vertices are all bivalent, with one input edge and one output edge, while the odd vertices can connect to an arbitrary number of even vertices depending on the lattice structure.
We still need to introduce a local ordering of edges connected to an odd vertex. To simplify the notation, we still adopt a global ordering of all edges of the lattice, naturally inducing a local ordering for each vertex. 
Also, notice that any given local ordering of $N_E(v_o)$ can be embedded into a global ordering of total edge set $E$.

Recall that there exists a local order of edges connecting to each odd vertex $v_o$, denoted as $e_{v_o,1}, \ldots, e_{v_o,m}$. We set $U_{e_{v_o,k}}=C\overset{\rightarrow}{X}$ if $v_o$ is the starting point of edge $e_{v_o,k}$; otherwise, we set $U_{e_{v_o,k}}= C\overset{\leftarrow}{X}$. We define the odd vertex operator as 
\begin{equation}
    W_{v_o}=\prod_{j=1}^m U_{e_{v_o,j}},
\end{equation}
where the order is crucial, since $[U_{e_{v_o,j}},U_{e_{v_o,l}}]\neq 0$ in general. 
However, for different odd vertices $v_o$ and $w_o$, the corresponding entangler operators commute 
\begin{equation}
    [W_{v_o},W_{w_o}]=0.
\end{equation}
For two odd vertices $v_o$ and $w_o$ that do not simultaneously connect to some even vertex, it's clear that $[W_{v_o},W_{w_o}]=0$. If they connect to some even vertex $u_e$ simultaneously via edges $e_i$ and $e_j$, from the construction of the cluster lattice, we observe that $u_e$ must not be the endpoint (or starting point) of two edges simultaneously. This implies that $[U_{e_1},U_{e_2}]=0$, as one arises from the $\XR$ and the other from the $\XL$. Consequently, $[W_{v_o},W_{w_o}]=0$. This guarantees that the circuits have a locality structure.

\begin{definition}[Hopf cluster-lattice state]
For a cluster lattice $\mathbb{M}^d$ and a Hopf algebra $\cA$, the Hopf cluster state is defined as
\begin{equation}\label{eq:Wcluster}
    |\mathbb{M}^d,\cA\rangle=(\prod_{v_o} W_{v_o} )|\Omega\rangle.
\end{equation}
Using the global ordering of the edge set, it can be expressed as
\begin{equation}
     |\mathbb{M}^d,\cA\rangle=(\prod_{j=1}^{|E|} U_{e_j}) |\Omega\rangle.
\end{equation}
It's easy to verify that the two definitions are equivalent.
\end{definition}

For a $1$-dimensional manifold, all vertices of a bipartite graph are bivalent. To ensure that the graph is a cluster lattice, we only need to impose some extra conditions on the orientation of each edge. This simplified the discussion of SPT phases in the 1d cluster state.
We believe our construction here for the Hopf cluster-lattice state provides a natural generalization of cluster states for discussing topological phases that have a Hopf symmetry. This provides a quite different construction from the anyonic chain \cite{Feiguin2007interacting} and the model based on Hopf comodule algebra \cite{inamura2022lattice}.
The detailed discussion from a topological phase perspective will be given in another work \cite{jia2024cluster}.

\section{Hopf tensor network representation}
\label{sec:HopfTensor}
In this section, we will establish the tensor network representation for the Hopf cluster state. 
Our construction of a tensor network is a combination of tensor network and string diagrams of Hopf algebra, which is a powerful tool when doing calculations. Notice that this should not be confused with the tensor network introduced in Refs.~\cite{Buerschaper2013a,jia2023boundary,Jia2023weak,girelli2021semidual} in solving Hopf and weak Hopf quantum double models.
In this part, we will always fix a basis ${1_{\mathcal{A}}=g_0,g_1,\cdots,g_{|\cA|-1}}$ of the Hopf algebra $\mathcal{A}$ and all the contractions are assumed to be made for this basis if not specified.

\subsection{Hopf algebra data as tensors}

We assume the network is read upwards, thus a general element $h\in \cA$ can be represented 
as 
\begin{equation}
h=\sum_k c_k g_k =  \begin{aligned}
        \begin{tikzpicture}
            \draw[black,fill=lightgray] (0.5,0.5) circle (0.2);          \draw[line width=.6pt,black] (0.5,0.7)--(0.5,1.3);
            \node[ line width=0.6pt, dashed, draw opacity=0.5] (a) at (0.5,0.5){$h$};
        \end{tikzpicture}
    \end{aligned},
\end{equation}
with $\{g_k\}_{k=0}^{|\cA|-1}$ the fixed basis and coefficients $c_k\in \Cbb$. 
recall that the inner product of $\cA$ is given by Eq.~\eqref{eq:innerProd}, $c_k=\langle g_k,h\rangle$.
The $*$-operation of $\cA$ is antilinear: $h^*=\sum_k c_k^* g^*_k$, where $c^*_k$ is complex conjugation $c_k$ and $g_k^*$ is involution of $g_k$. From $\langle g_k,g_l\rangle=\Lambda(g^*_kg_l)=\delta_{k,l}$, we see $\langle s,t\rangle =\sum_k a_k^*b_k$ with $a_k$ and $b_k$ the coefficients of $s$ and $t$ respectively.
The inner product can be diagrammatically represented as
\begin{equation}
\langle s,t\rangle =\Lambda(s^*t)=\sum_k a_k^* b_k =  
\begin{aligned}
        \begin{tikzpicture}
            \draw[black,fill=lightgray] (0.5,0.5) circle (0.2); 
            \draw[black,fill=lightgray] (0.5,1.5) circle (0.2);          
            \draw[line width=.6pt,black] (0.5,0.7)--(0.5,1.3);
            \node[ line width=0.6pt, dashed, draw opacity=0.5] (a) at (0.5,0.5){$s^*$};
            \node[ line width=0.6pt, dashed, draw opacity=0.5] (a) at (0.5,1.5){$t$};
        \end{tikzpicture}
    \end{aligned}.
\end{equation}
We will utilize three distinct types of tensor legs: (i) virtual Hopf legs, represented as black, denoting contracted edges spanning the basis $\cA$; (ii) physical Hopf legs, highlighted as blue, corresponding to physical degrees of freedom valued in $\cA$; and (iii) virtual non-Hopf legs, denoted by cyan, taking values in a Hilbert space V (possibly a representation space of $\cA$ or additional ancillary spaces). These distinctions facilitate the clear representation and manipulation of tensors in our computations.

\vspace{1em}
\emph{Algebra tensor.} ---
The algebraic structure is characterized by the structure constants $A_{ab}^c$, where
\begin{equation}
\mu(g_a\otimes g_b)=g_ag_b=\sum_c A_{ab}^c g_c.
\end{equation}
It can be represented as a tensor:
\begin{equation}
 A_{ab}^c=  \begin{aligned}
        \begin{tikzpicture}
           \draw[fill=lightgray]  (0,0) -- (1,0) -- (0.5,0.5) -- cycle;
           \draw[line width=.6pt,black] (0.2,0)--(0.2,-0.6);
           \draw[line width=.6pt,black] (0.8,0)--(0.8,-0.6);
           \draw[line width=.6pt,black] (0.5,0.5)--(0.5,1.1);
            \node[ line width=0.6pt, dashed, draw opacity=0.5] (a) at (0,-0.4){$a$};
            \node[ line width=0.6pt, dashed, draw opacity=0.5] (a) at (1,-0.4){$b$};
            \node[ line width=0.6pt, dashed, draw opacity=0.5] (a) at (0.7,0.9){$c$};
        \end{tikzpicture}
    \end{aligned}.
\end{equation}
We will use labeled lines to represent the given basis with the respective label, dotted lines to represent the identity element $g_0 = 1_{\mathcal{A}}$, and the lines connecting two tensors to represent the contraction of the corresponding indices. For more details about the tensor network, see Refs.~\cite{Orus2014tensornet,Cirac2021MPSreivew}.
The unit axiom of algebra can be represented diagrammatically as follows:
\begin{align}
  &  \begin{aligned}
    A_{0b}^c=A_{b0}^c=\delta_{b,c},
\end{aligned}\\
&   \begin{aligned}
        \begin{tikzpicture}
           \draw[fill=lightgray]  (0,0) -- (1,0) -- (0.5,0.5) -- cycle;
           \draw[dotted,line width=.6pt,black] (0.2,0)--(0.2,-0.6);
           \draw[line width=.6pt,black] (0.8,0)--(0.8,-0.6);
           \draw[line width=.6pt,black] (0.5,0.5)--(0.5,1.1);
            \node[ line width=0.6pt, dashed, draw opacity=0.5] (a) at (0,-0.4){$0$};
            \node[ line width=0.6pt, dashed, draw opacity=0.5] (a) at (1,-0.4){$b$};
            \node[ line width=0.6pt, dashed, draw opacity=0.5] (a) at (0.7,0.9){$c$};
        \end{tikzpicture}
    \end{aligned}=\delta_{b,c} \,\,
     \begin{aligned}
        \begin{tikzpicture}
         %  \draw[fill=lightgray]  (0.0) -- (1,0) -- (0.5,0.5) -- cycle;
         %  \draw[line width=.6pt,black] (0.2,0)--(0.2,-0.6);
           \draw[line width=.6pt,black] (0.8,0.6)--(0.8,-.6);
           %\draw[line width=.6pt,black] (0.5,0.5)--(0.5,1.1);
            %\node[ line width=0.6pt, dashed, draw opacity=0.5] (a) at (0,-0.4){$a$};
            %\node[ line width=0.6pt, dashed, draw opacity=0.5] (a) at (1,-0.4){$b$};
            \node[ line width=0.6pt, dashed, draw opacity=0.5] (a) at (1,0){$c$};
        \end{tikzpicture}
    \end{aligned}
    = \begin{aligned}
        \begin{tikzpicture}
           \draw[fill=lightgray]  (0,0) -- (1,0) -- (0.5,0.5) -- cycle;
           \draw[line width=.6pt,black] (0.2,0)--(0.2,-0.6);
           \draw[dotted,line width=.6pt,black] (0.8,0)--(0.8,-0.6);
           \draw[line width=.6pt,black] (0.5,0.5)--(0.5,1.1);
            \node[ line width=0.6pt, dashed, draw opacity=0.5] (a) at (0,-0.4){$b$};
            \node[ line width=0.6pt, dashed, draw opacity=0.5] (a) at (1,-0.4){$0$};
            \node[ line width=0.6pt, dashed, draw opacity=0.5] (a) at (0.7,0.9){$c$};
        \end{tikzpicture}
    \end{aligned}.
\end{align}
The associativity can be represented as follows:
\begin{align}
  &  \begin{aligned}
       \sum_{d} A_{ab}^d A_{dc}^e=\sum_{d} A_{ad}^e A_{bc}^d.
    \end{aligned}\\
   &\begin{aligned}
        \begin{tikzpicture}
           \draw[fill=lightgray]  (0,0) -- (1,0) -- (0.5,0.5) -- cycle;
           \draw[fill=lightgray]  (0.3,1.1) -- (1.5,1.1) -- (0.9,1.6) -- cycle;
           \draw[line width=.6pt,black] (0.2,0)--(0.2,-0.6);
           \draw[line width=.6pt,black] (0.8,0)--(0.8,-0.6);
           \draw[line width=.6pt,black] (0.5,0.5)--(0.5,1.1);
           \draw[line width=.6pt,black] (1.3,1.1)--(1.3,-.6);
          \draw[line width=.6pt,black] (0.9,1.6)--(0.9,2.1);
            \node[ line width=0.6pt, dashed, draw opacity=0.5] (a) at (0,-0.4){$a$};
            \node[ line width=0.6pt, dashed, draw opacity=0.5] (a) at (1.5,-0.4){$c$};
            \node[ line width=0.6pt, dashed, draw opacity=0.5] (a) at (1,-0.4){$b$};
            \node[ line width=0.6pt, dashed, draw opacity=0.5] (a) at (1.1,1.9){$e$};
        \end{tikzpicture}
    \end{aligned}=
    \begin{aligned}
        \begin{tikzpicture}
           \draw[fill=lightgray]  (0,0) -- (1,0) -- (0.5,0.5) -- cycle;
           \draw[fill=lightgray]  (-0.5,1.1) -- (.7,1.1) -- (0.1,1.6) -- cycle;
           \draw[line width=.6pt,black] (0.2,0)--(0.2,-0.6);
           \draw[line width=.6pt,black] (0.8,0)--(0.8,-0.6);
           \draw[line width=.6pt,black] (0.5,0.5)--(0.5,1.1);
           \draw[line width=.6pt,black] (-.3,1.1)--(-.3,-.6);
          \draw[line width=.6pt,black] (0.1,1.6)--(0.1,2.1);
            \node[ line width=0.6pt, dashed, draw opacity=0.5] (a) at (0,-0.4){$b$};
            \node[ line width=0.6pt, dashed, draw opacity=0.5] (a) at (1,-0.4){$c$};
            \node[ line width=0.6pt, dashed, draw opacity=0.5] (a) at (0.3,1.9){$e$};
            \node[ line width=0.6pt, dashed, draw opacity=0.5] (a) at (-0.5,-0.4){$a$};
        \end{tikzpicture}
    \end{aligned}.
\end{align}
Notice that associativity allows us to define the algebra tensor $A_{a_1\cdots a_n}^b$ with $n$ input legs and one output leg without ambiguity. There are $\frac{(2n-2)!}{n! (n-1)!}$ (the $(n-1)$-th Catalan number) equivalent ways to decompose $A_{a_1\cdots a_n}^b$ into the contraction of basic algebra tensors.

\vspace{1em}
\emph{Coalgebra tensor.} ---
The coproduct structure is characterized by the structure constant $C^{bc}_a$
\begin{equation}
    \Delta(g_a)=\sum_{b,c}C^{bc}_a g_b\otimes g_c.
\end{equation}
The $C^{bc}_a$ can be represented as
\begin{equation}
 C_{a}^{bc}=  \begin{aligned}
        \begin{tikzpicture}
           \draw[fill=lightgray]  (0,0) -- (1,0) -- (0.5,-0.5) -- cycle;
           \draw[line width=.6pt,black] (0.2,0)--(0.2,0.6);
           \draw[line width=.6pt,black] (0.8,0)--(0.8,0.6);
           \draw[line width=.6pt,black] (0.5,-0.5)--(0.5,-1.1);
            \node[ line width=0.6pt, dashed, draw opacity=0.5] (a) at (0,0.4){$b$};
            \node[ line width=0.6pt, dashed, draw opacity=0.5] (a) at (1,0.4){$c$};
            \node[ line width=0.6pt, dashed, draw opacity=0.5] (a) at (0.7,-0.9){$a$};
        \end{tikzpicture}
    \end{aligned}.
\end{equation}
The counit map can be represented as an edge-removing operation $\varepsilon(g_a)=\varepsilon_a$ and is denoted as
\begin{equation}
    \begin{aligned}
        \begin{tikzpicture}
            \draw[black,fill=lightgray] (0.5,1.5) circle (0.2);       
           \draw[line width=.6pt,black] (0.5,0.7)--(0.5,1.3);
            \node[ line width=0.6pt, dashed, draw opacity=0.5] (a) at (0.5,1.5){$\varepsilon$};
             \node[ line width=0.6pt, dashed, draw opacity=0.5] (a) at (0.7,1.1){$a$};
        \end{tikzpicture}
    \end{aligned} =\varepsilon_a.
\end{equation}
The counit axiom can thus be represented as 
\begin{align}
  &  \begin{aligned}
  \sum_b \varepsilon_b C_{a}^{bc}=\sum_c \varepsilon_d A_{a}^{cd}=\delta_{a,c},
\end{aligned}\\
&   \begin{aligned}
        \begin{tikzpicture}
           \draw[fill=lightgray]  (0,0) -- (1,0) -- (0.5,-0.5) -- cycle;
            \draw[black,fill=lightgray] (0.2,0.8) circle (0.2);      
           \draw[line width=.6pt,black] (0.2,0)--(0.2,0.6);
           \draw[line width=.6pt,black] (0.8,0)--(0.8,0.6);
           \draw[line width=.6pt,black] (0.5,-0.5)--(0.5,-1.1);
            \node[ line width=0.6pt, dashed, draw opacity=0.5] (a) at (0.2,0.8){$\varepsilon$};
            \node[ line width=0.6pt, dashed, draw opacity=0.5] (a) at (1,0.4){$c$};
            \node[ line width=0.6pt, dashed, draw opacity=0.5] (a) at (0.7,-0.9){$a$};
        \end{tikzpicture}
    \end{aligned}=\delta_{a,c} \,\,
     \begin{aligned}
        \begin{tikzpicture}
         %  \draw[fill=lightgray]  (0,0) -- (1,0) -- (0.5,0.5) -- cycle;
         %  \draw[line width=.6pt,black] (0.2,0)--(0.2,-0.6);
           \draw[line width=.6pt,black] (0.8,0.6)--(0.8,-.6);
           %\draw[line width=.6pt,black] (0.5,0.5)--(0.5,1.1);
            %\node[ line width=0.6pt, dashed, draw opacity=0.5] (a) at (0,-0.4){$a$};
            %\node[ line width=0.6pt, dashed, draw opacity=0.5] (a) at (1,-0.4){$b$};
            \node[ line width=0.6pt, dashed, draw opacity=0.5] (a) at (1,0){$c$};
        \end{tikzpicture}
    \end{aligned}
    =  \begin{aligned}
        \begin{tikzpicture}
           \draw[fill=lightgray]  (0,0) -- (1,0) -- (0.5,-0.5) -- cycle;
            \draw[black,fill=lightgray] (0.8,0.8) circle (0.2);      
           \draw[line width=.6pt,black] (0.2,0)--(0.2,0.6);
           \draw[line width=.6pt,black] (0.8,0)--(0.8,0.6);
           \draw[line width=.6pt,black] (0.5,-0.5)--(0.5,-1.1);
            \node[ line width=0.6pt, dashed, draw opacity=0.5] (a) at (0.8,0.8){$\varepsilon$};
            \node[ line width=0.6pt, dashed, draw opacity=0.5] (a) at (0,0.4){$c$};
            \node[ line width=0.6pt, dashed, draw opacity=0.5] (a) at (0.7,-0.9){$a$};
        \end{tikzpicture}
    \end{aligned}.
\end{align}
The coassociativity can be represented as
\begin{align}
  &  \begin{aligned}
       \sum_{c} C^{bc}_a C^{de}_c=\sum_{c} C_{c}^{bd} C_{a}^{ce},
    \end{aligned}\\
   &\begin{aligned}
        \begin{tikzpicture}
           \draw[fill=lightgray]  (0,0) -- (1,0) -- (0.5,-0.5) -- cycle;
           \draw[fill=lightgray]  (0.3,1.1) -- (1.3,1.1) -- (0.8,0.6) -- cycle;
           \draw[line width=.6pt,black] (0.2,0)--(0.2,1.7);
           \draw[line width=.6pt,black] (0.8,0)--(0.8,0.6);
           \draw[line width=.6pt,black] (0.5,-0.5)--(0.5,-1.1);
           \draw[line width=.6pt,black] (1.1,1.1)--(1.1,1.7);
          \draw[line width=.6pt,black] (0.5,1.1)--(0.5,1.7);
            \node[ line width=0.6pt, dashed, draw opacity=0.5] (a) at (0.7,-0.7){$a$};
            \node[ line width=0.6pt, dashed, draw opacity=0.5] (a) at (0,1.5){$b$};
            \node[ line width=0.6pt, dashed, draw opacity=0.5] (a) at (0.7,1.5){$d$};
            \node[ line width=0.6pt, dashed, draw opacity=0.5] (a) at (1.3,1.5){$e$};
        \end{tikzpicture}
    \end{aligned}=
    \begin{aligned}
        \begin{tikzpicture}
           \draw[fill=lightgray]  (0,0) -- (1,0) -- (0.5,-0.5) -- cycle;
           \draw[fill=lightgray]  (-0.3,1.1) -- (.7,1.1) -- (0.2,0.6) -- cycle;
           \draw[line width=.6pt,black] (0.2,0)--(0.2,0.6);
           \draw[line width=.6pt,black] (0.8,0)--(0.8,1.7);
           \draw[line width=.6pt,black] (0.5,-0.5)--(0.5,-1.1);
           \draw[line width=.6pt,black] (-.1,1.1)--(-.1,1.7);
          \draw[line width=.6pt,black] (0.5,1.1)--(0.5,1.7);
            \node[ line width=0.6pt, dashed, draw opacity=0.5] (a) at (0.7,-0.7){$a$};
            \node[ line width=0.6pt, dashed, draw opacity=0.5] (a) at (-.4,1.5){$b$};
            \node[ line width=0.6pt, dashed, draw opacity=0.5] (a) at (0.2,1.5){$d$};
            \node[ line width=0.6pt, dashed, draw opacity=0.5] (a) at (1.,1.5){$e$};
        \end{tikzpicture}
    \end{aligned}.
\end{align}
The coassociativity allows us to define the coalgebra tensor $C^{a_1\cdots a_n}_b$ with $n$ input legs and one output leg without ambiguity. There are also $\frac{(2n-2)!}{n! (n-1)!}$ (the $(n-1)$-th Catalan number) equivalent ways to decompose $C^{a_1\cdots a_n}_b$ into the contraction of basic coalgebra tensors.

\vspace{1em}
\emph{Bialgebra axioms.} --- The definition of bialgebra requires that $\varepsilon$ and $\Delta$  must be algebra homomorphisms.
For counit tensor, we have $\varepsilon(g_ag_b)=\varepsilon(g_a)\varepsilon(g_b)$ for all $a,b$,  which is equivalent to say $\sum_c A_{ab}^c\varepsilon_c=\varepsilon_a\varepsilon_b$.
\begin{equation}
\begin{aligned}
        \begin{tikzpicture}
           \draw[black,fill=lightgray] (0.5,1.3) circle (0.2);    
           \draw[fill=lightgray]  (0,0) -- (1,0) -- (0.5,0.5) -- cycle;
           \draw[line width=.6pt,black] (0.2,0)--(0.2,-0.6);
           \draw[line width=.6pt,black] (0.8,0)--(0.8,-0.6);
           \draw[line width=.6pt,black] (0.5,0.5)--(0.5,1.1);
            \node[ line width=0.6pt, dashed, draw opacity=0.5] (a) at (0,-0.4){$a$};
            \node[ line width=0.6pt, dashed, draw opacity=0.5] (a) at (1,-0.4){$b$};
            \node[ line width=0.6pt, dashed, draw opacity=0.5] (a) at (0.7,0.9){$c$};
              \node[ line width=0.6pt, dashed, draw opacity=0.5] (a) at (0.5,1.3){$\varepsilon$};
        \end{tikzpicture}
    \end{aligned}
    =
    \begin{aligned}
        \begin{tikzpicture}
           \draw[black,fill=lightgray] (0.2,0.2) circle (0.2);    
            \draw[black,fill=lightgray] (0.8,0.2) circle (0.2);
           \draw[line width=.6pt,black] (0.2,0)--(0.2,-0.6);
           \draw[line width=.6pt,black] (0.8,0)--(0.8,-0.6);
         %  \draw[line width=.6pt,black] (0.5,0.5)--(0.5,1.1);
            \node[ line width=0.6pt, dashed, draw opacity=0.5] (a) at (0,-0.4){$a$};
            \node[ line width=0.6pt, dashed, draw opacity=0.5] (a) at (1,-0.4){$b$};
          \node[ line width=0.6pt, dashed, draw opacity=0.5] (a) at (0.2,0.2){$\varepsilon$};
        \node[ line width=0.6pt, dashed, draw opacity=0.5] (a) at (0.8,0.2){$\varepsilon$};
        \end{tikzpicture}
    \end{aligned}.
\end{equation}
For the coalgebra tensor, we have
$\Delta(g_ag_b)=\Delta(g_a)\Delta(g_b)$, which is equivalent to $\sum_c A_{ab}^cC_c^{de}=\sum_{i,j,k,l}C_a^{ij}C_b^{kl}A_{ik}^dA_{jl}^e$.
\begin{equation}
 \begin{aligned}
        \begin{tikzpicture}
           \draw[fill=lightgray]  (0,0) -- (1,0) -- (0.5,0.5) -- cycle;
           \draw[fill=lightgray]  (0,1.6) -- (1,1.6) -- (0.5,1.1) -- cycle;
           \draw[line width=.6pt,black] (0.2,1.6)--(0.2,2.2);
           \draw[line width=.6pt,black] (0.8,1.6)--(0.8,2.2);
           \draw[line width=.6pt,black] (0.2,0)--(0.2,-0.6);
           \draw[line width=.6pt,black] (0.8,0)--(0.8,-0.6);
           \draw[line width=.6pt,black] (0.5,0.5)--(0.5,1.1);
            \node[ line width=0.6pt, dashed, draw opacity=0.5] (a) at (0,-0.4){$a$};
            \node[ line width=0.6pt, dashed, draw opacity=0.5] (a) at (1,-0.4){$b$};
            \node[ line width=0.6pt, dashed, draw opacity=0.5] (a) at (0.7,0.9){$c$};
            \node[ line width=0.6pt, dashed, draw opacity=0.5] (a) at (0,2.0){$d$};
            \node[ line width=0.6pt, dashed, draw opacity=0.5] (a) at (1,2.0){$e$};
        \end{tikzpicture}
    \end{aligned}=
    \begin{aligned}
      \begin{tikzpicture}
           \draw[fill=lightgray]  (0,0) -- (1,0) -- (0.5,-0.5) -- cycle;
           \draw[line width=.6pt,black] (0.2,0)--(0.2,0.6);
           \draw[line width=.6pt,black] (0.8,0)--(1.4,0.6);
           \draw[line width=.6pt,black] (0.5,-0.5)--(0.5,-1.1);
            \node[ line width=0.6pt, dashed, draw opacity=0.5] (a) at (0,0.2){$i$};
            \node[ line width=0.6pt, dashed, draw opacity=0.5] (a) at (0.6,0.2){$j$};
            \node[ line width=0.6pt, dashed, draw opacity=0.5] (a) at (0.7,-0.9){$a$};
          \draw[fill=lightgray]  (1.2,0) -- (2.2,0) -- (1.7,-0.5) -- cycle;
           \draw[line width=.6pt,black] (1.4,0)--(1.15,0.25);
           \draw[line width=.6pt,black] (0.8,0.6)--(1.05,0.35);

           \draw[line width=.6pt,black] (2,0)--(2,0.6);
           \draw[line width=.6pt,black] (1.7,-0.5)--(1.7,-1.1);
            \node[ line width=0.6pt, dashed, draw opacity=0.5] (a) at (2.2,0.2){$l$};
            \node[ line width=0.6pt, dashed, draw opacity=0.5] (a) at (1.5,0.2){$k$};
            \node[ line width=0.6pt, dashed, draw opacity=0.5] (a) at (1.9,-0.9){$b$};
           \draw[fill=lightgray]  (0,0.6) -- (1,0.6) -- (0.5,1.1) -- cycle;
               \draw[fill=lightgray]  (1.2,0.6) -- (2.2,0.6) -- (1.7,1.1) -- cycle;
        \draw[line width=.6pt,black] (0.5,1.1)--(.5,1.7);
        \draw[line width=.6pt,black] (1.7,1.1)--(1.7,1.7);
       \node[ line width=0.6pt, dashed, draw opacity=0.5] (a) at (1.9,1.4){$e$};
               \node[ line width=0.6pt, dashed, draw opacity=0.5] (a) at (0.7,1.4){$d$};
        \end{tikzpicture}
    \end{aligned}.
\end{equation}
These basic rules of tensor contraction can be used to derive more complicated rules, such as those for $\Delta(g_{a}g_bg_c\cdots)$.

\vspace{1em}
\emph{Antipode tensor.} --- For Hopf algebra, antipode can also be represented as a tensor $S(g_a)=\sum_b S_a^bg_b$.
It satisfies $\sum_{b,c,d} C_a^{bc}S_c^dA_{bd}^e=\varepsilon_a \delta_{e,0}=\sum_{b,c,d} C_a^{bc}S_b^dA_{dc}^e$.

\begin{equation}
\begin{aligned}
        \begin{tikzpicture}
           \draw[fill=lightgray]  (0,0) -- (1,0) -- (0.5,-0.5) -- cycle;
           \draw[line width=.6pt,black] (0.2,0)--(0.2,0.6);
           \draw[line width=.6pt,black] (0.8,0)--(0.8,1.6);
           \draw[line width=.6pt,black] (0.5,-0.5)--(0.5,-1.1);
           \draw[line width=.6pt,black,fill=lightgray] (0,0.6) rectangle ++(0.4,0.4);
           \draw[line width=.6pt,black] (0.2,1)--(0.2,1.6);
            \node[ line width=0.6pt, dashed, draw opacity=0.5] (a) at (0.2,0.8){$S$};
            \draw[fill=lightgray]  (0,1.6) -- (1,1.6) -- (0.5,2.1) -- cycle;
            \draw[line width=.6pt,black] (0.5,2.1)--(0.5,2.6);
            \node[ line width=0.6pt, dashed, draw opacity=0.5] (a) at (0.7,2.3){$e$};
            \node[ line width=0.6pt, dashed, draw opacity=0.5] (a) at (0.7,-0.9){$a$};
        \end{tikzpicture}
    \end{aligned}=\delta_{e,0}
    \begin{aligned}
        \begin{tikzpicture}
            \draw[black,fill=lightgray] (0.5,1.5) circle (0.2);       
           \draw[line width=.6pt,black] (0.5,0.7)--(0.5,1.3);
            \node[ line width=0.6pt, dashed, draw opacity=0.5] (a) at (0.5,1.5){$\varepsilon$};
             \node[ line width=0.6pt, dashed, draw opacity=0.5] (a) at (0.7,1.1){$a$};
        \end{tikzpicture}
    \end{aligned}
    =\begin{aligned}
        \begin{tikzpicture}
           \draw[fill=lightgray]  (0,0) -- (1,0) -- (0.5,-0.5) -- cycle;
           \draw[line width=.6pt,black] (0.2,0)--(0.2,1.6);
           \draw[line width=.6pt,black] (0.8,0)--(0.8,0.6);
           \draw[line width=.6pt,black] (0.5,-0.5)--(0.5,-1.1);
           \draw[line width=.6pt,black,fill=lightgray] (0.6,0.6) rectangle ++(0.4,0.4);
           \draw[line width=.6pt,black] (0.8,1)--(0.8,1.6);
            \node[ line width=0.6pt, dashed, draw opacity=0.5] (a) at (0.8,0.8){$S$};
            \draw[fill=lightgray]  (0,1.6) -- (1,1.6) -- (0.5,2.1) -- cycle;
            \draw[line width=.6pt,black] (0.5,2.1)--(0.5,2.6);
            \node[ line width=0.6pt, dashed, draw opacity=0.5] (a) at (0.7,2.3){$e$};
            \node[ line width=0.6pt, dashed, draw opacity=0.5] (a) at (0.7,-0.9){$a$};
        \end{tikzpicture}
    \end{aligned}.
\end{equation}

Using these structure tensors, given a Hopf quantum circuit with input as Hopf qudits and gates as Hopf operations, we can translate the corresponding string diagram into a tensor network, which we call a Hopf tensor network.

\subsection{Hopf tensor network representation of Hopf cluster state}

For the Hopf cluster state, there exists a Hopf tensor network representation. As we will see, our construction of the Hopf tensor network is a combination of tensor network and string diagrams, and it can naturally recover the matrix product state (MPS) and projected entangled pair states (PEPS) representations for the case $\mathcal{A}=\mathbb{C}[G]$ \cite{brell2015generalized,fechisin2023noninvertible}. Many properties therein have a natural explanation using the Hopf algebra structure.

Hereinafter we will put Haar integral $|\lambda\rangle$ instead of $|\one\rangle \propto |\lambda\rangle$ on each odd vertex of a cluster graph for the simplicity of the discussion.
For a given Hopf cluster state $|K,\cA\rangle$, we can draw a corresponding string diagram using the structure constant of $\cA$, see Fig.~\ref{fig:cluster-TN2d} and Fig.~\ref{fig:cluster-TN} for an illustration, and this string diagram can be translated into a Hopf tensor network if we replace the corresponding morphisms with structure tensors.

For an odd vertex connecting $n-1$ edges, the corresponding local tensor is an $n$-leg tensor with one physical leg and $n-1$ virtual legs, it can be represented by
\begin{equation}
   (S\otimes \id \otimes \cdots \otimes \id)\circ \Delta_n(\lambda)=\begin{aligned}
        \begin{tikzpicture}
           \draw[fill=lightgray]  (-0.5,0) -- (1.5,0) -- (0.5,-0.5) -- cycle;
           \draw[line width=.6pt,black] (-0.3,0)--(-0.3,1);
           \draw[line width=.6pt,black] (0,0)--(0,1);
           \draw[line width=.6pt,blue] (1.3,0)--(1.3,1);
           \draw[line width=.6pt,black] (0.5,-0.5)--(0.5,-1.1);
           \draw[black,fill=lightgray] (0.5,-.9) circle (0.2);   
           \draw[line width=.6pt,fill=lightgray] (-.5,0.3) rectangle ++(0.4,.4);
           \node[ line width=0.6pt, dashed, draw opacity=0.5] (a) at (-0.7,1.3){$S(\lambda^{(1)})$};
           \node[ line width=0.6pt, dashed, draw opacity=0.5] (a) at (0.3,1.3){$\lambda^{(2)}$};
           \node[ line width=0.6pt, dashed, draw opacity=0.5] (a) at (0.6,0.4){$\cdots$};
            \node[ line width=0.6pt, dashed, draw opacity=0.5] (a) at (1.3,1.3){{\color{blue} $\lambda^{(n)}$ } };
            \node[ line width=0.6pt, dashed, draw opacity=0.5] (a) at (0.5,-0.9){$\lambda$};
            \node[ line width=0.6pt, dashed, draw opacity=0.5] (a) at (0.5,-0.2){$C$}; 
            \node[ line width=0.6pt, dashed, draw opacity=0.5] (a) at (-0.3,.5){$S$};
        \end{tikzpicture}
    \end{aligned}
    =
    \begin{aligned}
        \begin{tikzpicture}  
        \draw[-latex,line width=.6pt,blue] (0,0) -- (-0.5,0.7);
        \draw[-latex,line width=.6pt,black] (-0.6,-0.6)-- (-0.2,-0.2);
        \draw[-latex,line width=.6pt,black] (-0.5,-0.5)-- (0,0);
        \draw[-latex,line width=.6pt,black] (0,0.1) -- (.9,0.1);
        %\draw[line width=.6pt,black] (0,0) -- (-0.8,0);
        \draw[black,fill=red] (0,0) circle (0.2);   
      %  \draw[fill=white]  (-0.25,0.5) rectangle ++(.5,.5);
       % \node[ line width=0.6pt, dashed, draw opacity=0.5] (a) at (0,0.75){\small $\XL_h$};
        \node[ line width=0.6pt, dashed, draw opacity=0.5] (a) at (0,0){\small {\color{white}$\lambda$}};
        \node[ line width=0.6pt, dashed, draw opacity=0.5] (a) at (-1.2,-0.7){\small {\color{black} $S(\lambda^{\cone})$ } };
        \node[ line width=0.6pt, dashed, draw opacity=0.5] (a) at (1.3,0){\small {\color{black} $\lambda^{\ctwo}$ } };
        \node[ line width=0.6pt, dashed, draw opacity=0.5] (a) at (-0.5,0.9){\small {\color{blue} $\lambda^{(n)}$ } };
        \node[ line width=0.6pt, dashed, draw opacity=0.5] (a) at (0.3,0.6){\small {\color{black} $\ddots$ } };
        \end{tikzpicture}
    \end{aligned}.
\end{equation}
The middle diagram represents the complete version of the odd vertex tensor, while the right one is a simplified version. In the middle diagram, we assume that all legs are directed upwards. In the right diagram, we need to emphasize the direction of the legs to determine if there is an antipode map.
The local ordering of the edges connecting to the odd vertex determines how we place the comultiplication components on each leg of the tensor.
If the gate operator for the edge connecting to even vertex is $C\XL$, there will be an antipode tensor; otherwise, there is no $S$.
We've employed blue legs to denote the physical degrees of freedom and black lines to represent the virtual indices in $\cA$. The blue leg always take the last comultiplication component $\lambda^{(n)}$.

We observe that the odd vertex tensor essentially serves as a comultiplication tensor.
It's also clear that, to contract the odd vertex tensor with a nearby even vertex tensor, the local ordering of edges is crucial, we can not change the order of edges arbitrarily.
However, it is worth mentioning that for the Haar integral, we have $\tau\circ \Delta(\lambda)=\Delta(\lambda)$, this means 
\begin{equation}
    \sum_{(\lambda)} \lambda^{\cone}\otimes \lambda^{\ctwo}=  \sum_{(\lambda)} \lambda^{\ctwo}\otimes \lambda^{\cone},
\end{equation}
or equivalently, $\sum_a \lambda^aC_{a}^{b_1 b_2}=\sum_a \lambda^aC_{a}^{b_2 b_1}$.
This further implies
\begin{equation}
\begin{aligned}
     \sum_{(\lambda)} \lambda^{\cone}\otimes \lambda^{\ctwo}\otimes \lambda^{\cthree}
     = & \sum_{(\lambda)} \lambda^{\cthree}\otimes \lambda^{\cone}\otimes \lambda^{\ctwo}
     =&\sum_{(\lambda)} \lambda^{\ctwo}\otimes \lambda^{\cthree}\otimes \lambda^{\cone},
\end{aligned}
\end{equation}
and so on in a similar fashion. Thus, the odd vertex tensor exhibits a cyclic property. This implies that we can permute the indices according to the above rule before performing the contraction.

\begin{figure}[t]
    \centering
    \includegraphics[width=9cm]{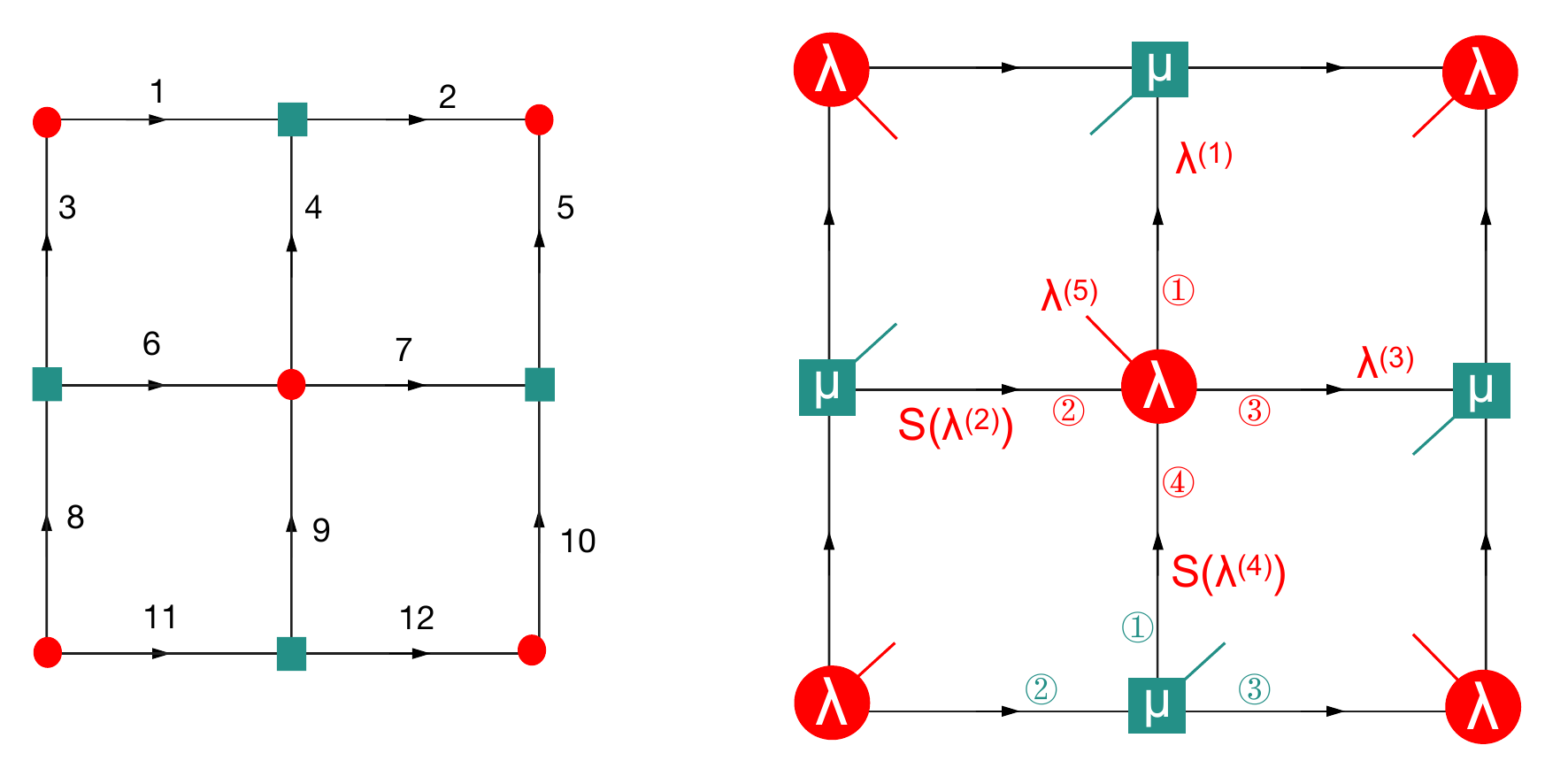}
    \caption{Illustration of the Hopf tensor network representation of the Hopf cluster state. On the left, there is a cluster graph with a globally ordered edge set, $e_1, \cdots, e_{12}$. For each vertex, the neighboring edges have a local ordering induced by global ordering of edge set. For an odd vertex, this local ordering determines how we place the comultiplication component on each edge and the edge direction determine whether we take antipode; for an even vertex, the local ordering determines the order for taking multiplication.
    On the right, we present the Hopf tensor network for the corresponding Hopf cluster state, where $\lambda$ represents Haar integral and $\mu$ represnets the multiplication map.
}
    \label{fig:cluster-TN2d}
\end{figure}

For the even vertex $v_e$, the tensor is essentially a multiplication tensor,
the leg corresponding to input state $1_{\cA}=g_0$ can be removed under the multiplication. The even vertex tensor can be represented as
\begin{equation}
  \mu(a_1 \cdots 1_{\cA} \cdots a_n )=
  \begin{aligned}
        \begin{tikzpicture}
           \draw[fill=lightgray]  (-0.5,0) -- (1.5,0) -- (0.5,0.5) -- cycle;
           \draw[line width=.6pt,black] (-0.3,0)--(-0.3,-1);
           \draw[line width=.6pt,black] (1.3,0)--(1.3,-1);
           \draw[line width=.6pt,black,blue] (0.5,0.5)--(0.5,0.8);
           \draw[line width=.6pt,black,dotted] (0.5,0)--(0.5,-1);
            \node[ line width=0.6pt, dashed, draw opacity=0.5] (a) at (0,-0.4){$\cdots$};
            \node[ line width=0.6pt, dashed, draw opacity=0.5] (a) at (1,-0.4){$\cdots$};
            \node[ line width=0.6pt, dashed, draw opacity=0.5] (a) at (0.4,-1.3){$1_{\cA}$};
            \node[ line width=0.6pt, dashed, draw opacity=0.5] (a) at (-0.4,-1.3){$a_1$};
             \node[ line width=0.6pt, dashed, draw opacity=0.5] (a) at (0.5,0.2){$A$};
            \node[ line width=0.6pt, dashed, draw opacity=0.5] (a) at (1.4,-1.3){$a_n$};
            \node[ line width=0.6pt, dashed, draw opacity=0.5] (a) at (0.5,1){\small {\color{blue} $a_1\cdots a_n$ } };
        \end{tikzpicture}
    \end{aligned}
    =
       \begin{aligned}
        \begin{tikzpicture}  
        \draw[line width=.6pt,blue] (0,0) -- (0,0.8);
        \draw[line width=.6pt,black] (0,0) -- (0.8,0);
        \draw[line width=.6pt,black] (0,0) -- (-0.8,0);
       \draw[dotted,line width=.6pt,black] (0,0) -- (0,-0.8);
       % \draw[black,fill=red] (0,0) circle (0.2);   
        \draw[fill=teal]  (-0.2,-0.2) rectangle ++(.4,.4);
       % \node[ line width=0.6pt, dashed, draw opacity=0.5] (a) at (0,0.75){\small $\XL_h$};
        \node[ line width=0.6pt, dashed, draw opacity=0.5] (a) at (0,0){\small {\color{white}$\mu$}};
        \node[ line width=0.6pt, dashed, draw opacity=0.5] (a) at (-0.9,0){\small {\color{black} $a_1$ } };
        \node[ line width=0.6pt, dashed, draw opacity=0.5] (a) at (1.2,0){\small {\color{black} $a_n$ } };
         \node[ line width=0.6pt, dashed, draw opacity=0.5] (a) at (0,1){\small {\color{blue} $a_1\cdots a_n$ } };
            \node[ line width=0.6pt, dashed, draw opacity=0.5] (a) at (-0.4,-.4){\small {\color{black} $\ddots$ } };
        \node[ line width=0.6pt, dashed, draw opacity=0.5] (a) at (0.6,-.4){\small {\color{black} $\iddots$ } };
        \end{tikzpicture}
    \end{aligned}
\end{equation}
Notice that for even vertex tensor, we do not need to stress that the directions of legs, since this part of information is encoded in the odd vertex tensor, all virtual legs of even vertex tensor are contracted with an odd vertex tensor.
We can alternatively encode the information of the antipode map in the even vertex: If the edge is an input edge, no antipode tensor is needed; otherwise, we contract the corresponding input line with an antipode tensor. This simply reflects how we handle the legs connecting even and odd vertices (See Fig.~\ref{fig:cluster-TN}).
However, the local ordering of edges connecting to the even vertex is still crucial here. If there are $n$ odd vertices connected to this even vertex, the corresponding even vertex tensor is given by $A_{a_1,\cdots,a_n'}^b S_{a_n}^{a'_n}$. It's clear that the local ordering for $N_E(v_e)$ is also necessary, as multiplication is not generally commutative.

The operations on Hopf qudits can also be represented using Hopf tensor networks.
The generalized Pauli-X operators can be represented as
\begin{equation}
    \XR_g=\begin{aligned}
        \begin{tikzpicture}
           \draw[fill=lightgray]  (0,0) -- (1,0) -- (0.5,0.5) -- cycle;
           \draw[line width=.6pt,black] (0.2,0)--(0.2,-0.6);
           \draw[line width=.6pt,black] (0.8,0)--(0.8,-1.6);
           \draw[line width=.6pt,black] (0.5,0.5)--(0.5,1.1);
           \draw[black,fill=lightgray] (0.2,-.6) circle (0.2);   
            \node[ line width=0.6pt, draw opacity=0.5] (a) at (0.2,-0.6){$g$};
        \end{tikzpicture}
    \end{aligned},\quad 
        \XL_g=\begin{aligned}
        \begin{tikzpicture}
           \draw[fill=lightgray]  (0,0) -- (1,0) -- (0.5,0.5) -- cycle;
           \draw[line width=.6pt,black] (0.2,0)--(0.2,-1.6);
           \draw[line width=.6pt,black] (0.8,0)--(0.8,-1.6);
           \draw[line width=.6pt,black] (0.5,0.5)--(0.5,1.1);
           \draw[black,fill=lightgray] (0.6,-.4) rectangle ++(0.4,-0.4);   
                   \draw[black,fill=lightgray] (0.8,-1.4) circle (0.2);  
            \node[ line width=0.6pt, draw opacity=0.5] (a) at (0.8,-0.6){$S$};
            \node[ line width=0.6pt, draw opacity=0.5] (a) at (0.8,-1.4){$g$};
        \end{tikzpicture}
    \end{aligned}\,\,.
\end{equation}
The generalized Pauli-Z operators can be represented as
\begin{equation}
    Z_{\Gamma}=\begin{aligned}
        \begin{tikzpicture}
           \draw[fill=lightgray]  (0,0) -- (1,0) -- (0.5,-0.5) -- cycle;
           \draw[line width=.6pt,black] (0.2,0)--(0.2,1);
           \draw[line width=.6pt,black] (0.8,0)--(0.8,0.6);
           \draw[line width=.6pt,black] (0.5,-0.5)--(0.5,-1.1);
          \draw[black,fill=lightgray] (0.6,.4) rectangle ++(0.4,0.4);   
           \draw[line width=.6pt,black,cyan] (1,0.6)--(1.4,0.6);
            \draw[line width=.6pt,black,cyan] (0.6,0.6)--(0.23,0.6);
            \draw[line width=.6pt,black,cyan] (0,0.6)--(0.17,0.6);
            \node[ line width=0.6pt, dashed, draw opacity=0.5] (a) at (0.8,0.6){$\Gamma$};
        \end{tikzpicture}
    \end{aligned}, \quad 
       Z_{\Gamma}^{\ddagger}= \begin{aligned}
        \begin{tikzpicture}
           \draw[fill=lightgray]  (0,0) -- (1,0) -- (0.5,-0.5) -- cycle;
           \draw[line width=.6pt,black] (0.2,0)--(0.2,1);
           \draw[line width=.6pt,black] (0.8,0)--(0.8,1);
           \draw[line width=.6pt,black] (0.5,-0.5)--(0.5,-1.1);
          \draw[black,fill=lightgray] (0,.4) rectangle ++(0.4,0.4);   
           \draw[line width=.6pt,black,cyan] (0.4,1.2)--(.8,1.2);
            \draw[line width=.6pt,black,cyan] (0,1.2)--(-.4,1.2);
           \draw[black,fill=lightgray] (0,1) rectangle ++(0.4,0.4); 
            \node[ line width=0.6pt, dashed, draw opacity=0.5] (a) at (0.2,1.2){$\Gamma$};
            \node[ line width=0.6pt, dashed, draw opacity=0.5] (a) at (0.2,0.6){$S$};
        \end{tikzpicture}
    \end{aligned}\,\,,
\end{equation}
where we use a cyan line to represent the virtual indices for the representation matrix $\Gamma\in \Rep(\cA)$. All other Hopf qudit operators can be represented in a similar way.

%\vspace{1em}
%\emph{Proof of Proposition~\ref{prop:RepSym} based on Hopf tensor network representation.} --- 
For $\Gamma,\Phi\in \Rep(\cA)$, $ Z_{\Gamma}Z_{\Phi}|g_k\rangle $ can be represented as  
\begin{equation}
   \begin{aligned}
        \begin{tikzpicture}
           \draw[fill=lightgray]  (0,0) -- (1,0) -- (0.5,-0.5) -- cycle;
           \draw[line width=.6pt,black] (0.2,0)--(0.2,1);
           \draw[line width=.6pt,black] (0.8,0)--(0.8,0.6);
           \draw[line width=.6pt,black] (0.5,-0.5)--(0.5,-1.1);
          \draw[black,fill=lightgray] (0.6,.4) rectangle ++(0.4,0.4);   
           \draw[line width=.6pt,black,cyan] (1,0.6)--(1.4,0.6);
            \draw[line width=.6pt,black,cyan] (0.6,0.6)--(0.23,0.6);
            \draw[line width=.6pt,black,cyan] (0,0.6)--(0.17,0.6);
         \draw[fill=lightgray]  (-0.3,1.5) -- (0.7,1.5) -- (0.2,1) -- cycle;
    \draw[line width=.6pt,black] (-0.1,1.5)--(-0.1,2.3);
           \draw[line width=.6pt,black] (0.5,1.5)--(0.5,1.9);
           \draw[black,fill=lightgray] (0.3,1.7) rectangle ++(0.4,0.4); 
          \draw[line width=.6pt,black,cyan] (0.7,1.9)--(1.1,1.9);
          \draw[line width=.6pt,black,cyan] (0.3,1.9)--(0,1.9);
            \node[ line width=0.6pt, dashed, draw opacity=0.5] (a) at (0.8,0.6){$\Phi$};
            \node[ line width=0.6pt, dashed, draw opacity=0.5] (a) at (0.5,1.9){$\Gamma$};
         \node[ line width=0.6pt, dashed, draw opacity=0.5] (a) at (0.7,-1){$k$};
        \end{tikzpicture}
    \end{aligned}
    = \begin{aligned}
        \begin{tikzpicture}
           \draw[fill=lightgray]  (0,0) -- (1,0) -- (0.5,-0.5) -- cycle;
           \draw[line width=.6pt,black] (0.2,0)--(0.2,0.4);
           \draw[line width=.6pt,black] (0.8,0)--(0.8,0.6);
           \draw[line width=.6pt,black] (0.5,-0.5)--(0.5,-1.1);
          \draw[black,fill=lightgray] (0.6,.4) rectangle ++(0.4,0.4);   
           \draw[line width=.6pt,black,cyan] (1,0.6)--(1.4,0.6);
            \draw[line width=.6pt,black,cyan] (0.6,0.6)--(0.5,0.6);
           % \draw[line width=.6pt,black,cyan] (0.5,0.6)--(0.5,0.8);
         \draw[fill=lightgray]  (-0.3,-1.1) -- (0.7,-1.1) -- (0.2,-1.6) -- cycle;
         \draw[line width=.6pt,black] (-0.1,-1.1)--(-0.1,1);
           \draw[black,fill=lightgray] (0,.2) rectangle ++(0.4,0.4); 
          \draw[line width=.6pt,black,cyan] (0.4,0.45)--(0.5,0.45);
          %\draw[line width=.6pt,black,cyan] (0.5,.45)--(0.5,0.3);
        \draw[line width=.6pt,black,cyan] (0,0.45)--(-0.05,0.45);
                \draw[line width=.6pt,black,cyan] (-0.15,0.45)--(-0.35,0.45);
             \draw[line width=.6pt,black] (0.2,-1.6)--(0.2,-2.1);
            \node[ line width=0.6pt, dashed, draw opacity=0.5] (a) at (0.8,0.6){$\Phi$};
            \node[ line width=0.6pt, dashed, draw opacity=0.5] (a) at (0.2,0.4){$\Gamma$};
         \node[ line width=0.6pt, dashed, draw opacity=0.5] (a) at (0.4,-1.8){$k$};
        \end{tikzpicture}
    \end{aligned}
    =\begin{aligned}
        \begin{tikzpicture}
          % \draw[fill=lightgray]  (0,0) -- (1,0) -- (0.5,-0.5) -- cycle;
         %  \draw[line width=.6pt,black] (0.2,0)--(0.2,0.4);
        %   \draw[line width=.6pt,black] (0.8,0)--(0.8,0.6);
           \draw[line width=.6pt,black] (0.5,-0.5)--(0.5,-1.1); 
           \draw[line width=.6pt,black,cyan] (0.2,-0.25)--(0,-0.25);
            \draw[line width=.6pt,black,cyan] (2,-0.25)--(2.4,-0.25);
           % \draw[line width=.6pt,black,cyan] (0.5,0.6)--(0.5,0.8);
         \draw[fill=lightgray]  (-0.3,-1.1) -- (0.7,-1.1) -- (0.2,-1.6) -- cycle;
         \draw[line width=.6pt,black] (-0.1,-1.1)--(-0.1,0.5);
           \draw[black,fill=lightgray] (0.2,-.5) rectangle ++(1.8,0.6); 
        %  \draw[line width=.6pt,black,cyan] (0.4,0.45)--(0.5,0.45);
          %\draw[line width=.6pt,black,cyan] (0.5,.45)--(0.5,0.3);
        %\draw[line width=.6pt,black,cyan] (0,0.45)--(-0.05,0.45);
            %    \draw[line width=.6pt,black,cyan] (-0.15,0.45)--(-0.35,0.45);
             \draw[line width=.6pt,black] (0.2,-1.6)--(0.2,-2.1);
         \node[ line width=0.6pt, dashed, draw opacity=0.5] (a) at (0.4,-1.8){$k$};
              \node[ line width=0.6pt, dashed, draw opacity=0.5] (a) at (1,-0.2){$\oplus_{\Psi}N_{\Gamma\Phi}^{\Psi}\Psi$};
        \end{tikzpicture}
    \end{aligned}.
\end{equation}
This implies $Z_{\Gamma}Z_{\Phi}=Z_{\Gamma \otimes \Phi}$. In a similar way, we can show $\tilde{Z}_{\Gamma}\tilde{Z}_{\Phi}=\tilde{Z}_{\Phi\otimes\Gamma}$.
By employing the Hopf tensor network, we have a more intuitive approach to handling quantum states and operations for Hopf qudits

Notice that our proposal for the Hopf tensor network representation differs from the matrix product operator approach discussed in Refs.~\cite{molnar2022matrix, GarreRubio2023classifyingphases}. However, the method of matrix product operator algebras can also be applied to 1d Hopf qudit chains.
The tensor networks based on Hopf algebras discussed in Refs.~\cite{Buerschaper2013a, jia2023boundary, Jia2023weak, girelli2021semidual} are actually special cases of the Hopf tensor network described here (if we regard the function $\varphi \in \bar{\cA}$ as a special tensor $f(g_k)=f_k$). A more general form will be discussed in Sec.~\ref{sec:GraphHypergraphHopf}.
Since the Hopf tensor network in Refs.~\cite{Buerschaper2013a, girelli2021semidual, jia2023boundary, Jia2023weak} can be used to solve the Hopf quantum double model, it can also be employed to solve the Hopf cluster state model based on the correspondence between the Hopf cluster state and the Hopf quantum double model (See Sec.~\ref{sec:QD}).

\begin{figure}[t]
    \centering
    \includegraphics[width=7cm]{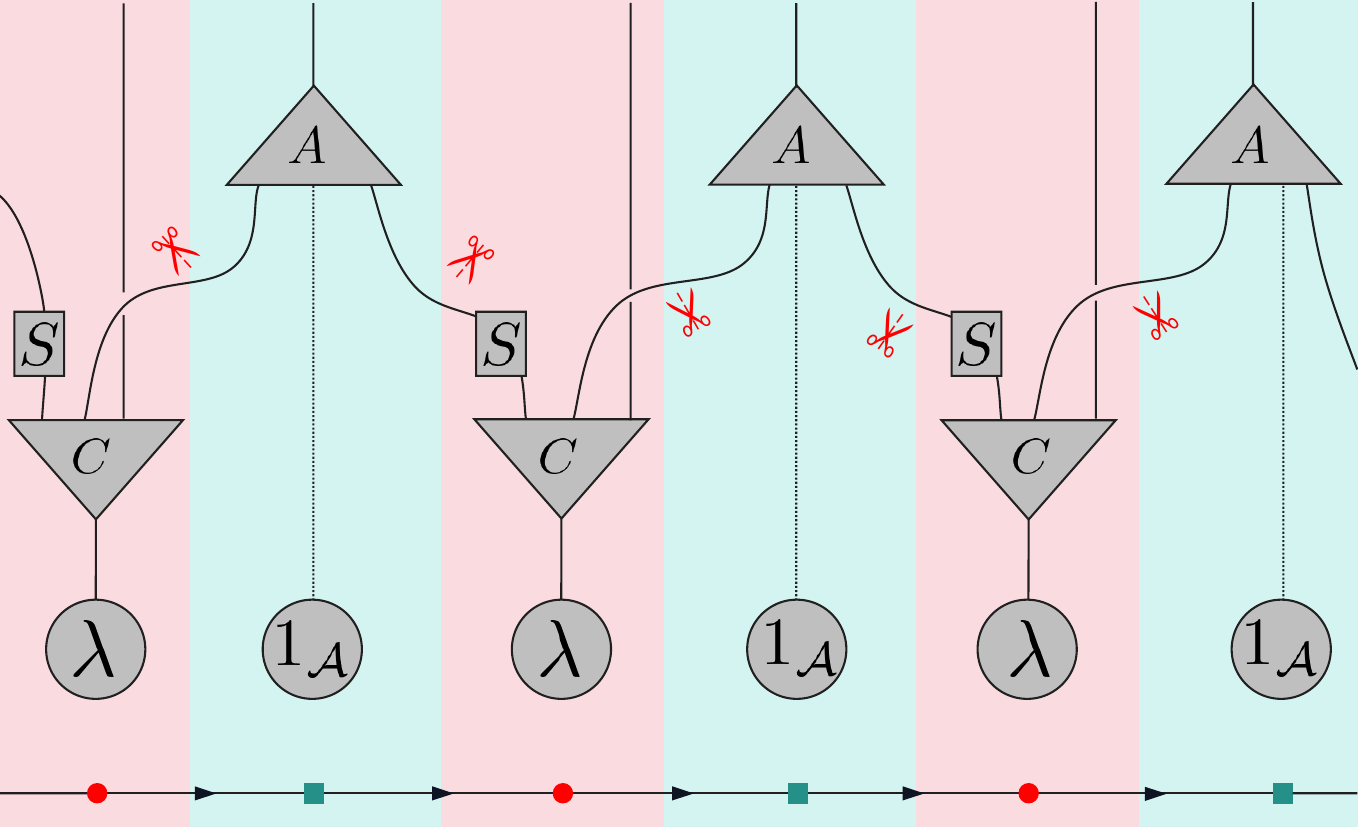}
    \caption{The Hopf tensor network for an one dimensional Hopf cluster state.}
    \label{fig:cluster-TN}
\end{figure}

\subsection{Example: 1d cluster lattice case}
\label{sec:OneDimHopf}

Let us take (1+1)D Hopf cluster lattice state as an example. The basic building block of the Hopf tensor network is odd vertex tensor and even vertex tensor, see Fig.~\ref{fig:cluster-TN} for a depiction.

Consider a cluster lattice as in Fig.~\ref{fig:cluster-TN}. For odd vertex, the tensor is given by
\begin{equation}
       \begin{aligned}
        \begin{tikzpicture}  
        \draw[line width=.6pt,blue] (0,0) -- (0,0.8);
        \draw[line width=.6pt,black] (0,0) -- (0.8,0);
        \draw[line width=.6pt,black] (0,0) -- (-0.8,0);
        \draw[black,fill=red] (0,0) circle (0.2);   
      %  \draw[fill=white]  (-0.25,0.5) rectangle ++(.5,.5);
       % \node[ line width=0.6pt, dashed, draw opacity=0.5] (a) at (0,0.75){\small $\XL_h$};
        \node[ line width=0.6pt, dashed, draw opacity=0.5] (a) at (0,0){\small {\color{white}$\lambda$}};
        \node[ line width=0.6pt, dashed, draw opacity=0.5] (a) at (-1.4,0){\small {\color{black} $S(\lambda^{\cone})$ } };
        \node[ line width=0.6pt, dashed, draw opacity=0.5] (a) at (1.2,0){\small {\color{black} $\lambda^{\ctwo}$ } };
         \node[ line width=0.6pt, dashed, draw opacity=0.5] (a) at (0,1.2){\small {\color{black} $\lambda^{\cthree}$ } };
        \end{tikzpicture}
    \end{aligned}
    :=
    \begin{aligned}
    \begin{tikzpicture}
           \draw[fill=lightgray]  (-0.2,0) -- (1.2,0) -- (0.5,-0.5) -- cycle;
           \draw[line width=.6pt,black] (0,0)--(0,0.6) -- (-0.4,0.8);
           \draw[line width=.6pt,black] (0.5,0)--(0.5,0.2) -- (0.96,0.4) ;
           \draw[line width=.6pt,black] (1.04,0.44)-- ++ (.4,.2) ;
           \draw[line width=.6pt,blue] (1,0)--(1,0.8);
           \draw[line width=.6pt,black] (0.5,-0.5)--(0.5,-1.1);
           \draw[black,fill=lightgray] (0.5,-.9) circle (0.2);   
           \draw[fill=lightgray]  (-0.2,0.1) rectangle ++(.4,.4);
           \node[ line width=0.6pt, dashed, draw opacity=0.5] (a) at (-0.4,1.2){$S(\lambda^{(1)})$};
           \node[ line width=0.6pt, dashed, draw opacity=0.5] (a) at (1,1.2){$\lambda^{(3)}$};
            \node[ line width=0.6pt, dashed, draw opacity=0.5] (a) at (1.8,0.8){$\lambda^{(2)}$};
            \node[ line width=0.6pt, dashed, draw opacity=0.5] (a) at (0.5,-0.9){$\lambda$};
            \node[ line width=0.6pt, dashed, draw opacity=0.5] (a) at (0.5,-0.2){$C$}; 
            \node[ line width=0.6pt, dashed, draw opacity=0.5] (a) at (0,0.3){$S$}; 
        \end{tikzpicture}
    \end{aligned}.
\end{equation}
This is given by first taking the comultiplication of the Haar integral $\Delta_2(\lambda)$, then acting the antipode $S$ to the first component $\lambda^{\cone}$. This antipode reflects the fact that for the corresponding edge, we apply the edge gate $C\XL$.

The even vertex tensor is simply a multiplication tensor with two input legs, notice that since we assign $1_{\cA}$ to the even vertex, it won't contribute to the multiplication, thus we omit this dotted leg. Diagrammatically, we have 
\begin{equation}
       \begin{aligned}
        \begin{tikzpicture}  
        \draw[line width=.6pt,blue] (0,0) -- (0,0.8);
        \draw[line width=.6pt,black] (0,0) -- (0.8,0);
        \draw[line width=.6pt,black] (0,0) -- (-0.8,0);
       % \draw[black,fill=red] (0,0) circle (0.2);   
        \draw[fill=teal]  (-0.2,-0.2) rectangle ++(.4,.4);
       % \node[ line width=0.6pt, dashed, draw opacity=0.5] (a) at (0,0.75){\small $\XL_h$};
        \node[ line width=0.6pt, dashed, draw opacity=0.5] (a) at (0,0){\small {\color{white}$\mu$}};
        \node[ line width=0.6pt, dashed, draw opacity=0.5] (a) at (-0.9,0){\small {\color{black} $x$ } };
        \node[ line width=0.6pt, dashed, draw opacity=0.5] (a) at (1,0){\small {\color{black} $y$ } };
         \node[ line width=0.6pt, dashed, draw opacity=0.5] (a) at (0,1){\small {\color{black} $x\cdot y$ } };
        \end{tikzpicture}
    \end{aligned}
    :=
    \begin{aligned}
     \begin{tikzpicture}
           \draw[fill=lightgray]  (0,0) -- (1,0) -- (0.5,0.5) -- cycle;
           \draw[line width=.6pt,black] (0.2,0)--(0.2,-0.4);
           \draw[line width=.6pt,black] (0.8,0)--(0.8,-.4);
           \draw[line width=.6pt,blue] (0.5,0.5)--(0.5,.9);
           %\draw[black,fill=lightgray] (0.2,-.6) circle (0.2);   
            \node[ line width=0.6pt, draw opacity=0.5] (a) at (0.2,-0.6){$x$};
             \node[ line width=0.6pt, draw opacity=0.5] (a) at (0.8,-0.6){$y$};
            \node[ line width=0.6pt, draw opacity=0.5] (a) at (0.5,1.1){$x\cdot y$};
             \node[ line width=0.6pt, draw opacity=0.5] (a) at (0.46,0.2){$A$};
        \end{tikzpicture}
    \end{aligned},
\end{equation}
where $\mu$ represents the multiplication map.

By gluing the odd and even vertex tensors, we obtain a Hopf tensor network representation of the 1d Hopf cluster state. For example, consider the open chain, we have
\begin{equation}
       \begin{aligned}
        \begin{tikzpicture}  
        \draw[line width=.6pt,black] (-0.8,0) -- (6.8,0);
        \draw[line width=.6pt,black] (0,0) -- (0.8,0);
        \draw[line width=.6pt,black] (0,0) -- (-0.8,0);
           \draw[line width=.6pt,blue] (0,0) -- (0,0.8);
        \draw[line width=.6pt,blue] (6,0) -- (6,0.8);
         \draw[line width=.6pt,blue] (1,0) -- (1,1.2);
      \draw[line width=.6pt,blue] (2,0) -- (2,0.8);
                 \draw[line width=.6pt,blue] (3,0) -- (3,1.2);
       \draw[line width=.6pt,blue] (4,0) -- (4,0.8);
          \draw[line width=.6pt,blue] (5,0) -- (5,1.2);
                     \draw[black,fill=red] (0,0) circle (0.2); 
        \draw[fill=teal]  (0.8,-0.2) rectangle ++(.4,.4);
                \draw[fill=teal]  (2.8,-0.2) rectangle ++(.4,.4);
                        \draw[fill=teal]  (4.8,-0.2) rectangle ++(.4,.4);
        \draw[black,fill=red] (2,0) circle (0.2);
                   \draw[black,fill=red] (4,0) circle (0.2); 
                        \draw[black,fill=red] (6,0) circle (0.2);
        \node[ line width=0.6pt, dashed, draw opacity=0.5] (a) at (0,0){\small {\color{white}$\lambda$}};
            \node[ line width=0.6pt, dashed, draw opacity=0.5] (a) at (2,0){\small {\color{white}$\lambda$}};
                \node[ line width=0.6pt, dashed, draw opacity=0.5] (a) at (4,0){\small {\color{white}$\lambda$}};
                    \node[ line width=0.6pt, dashed, draw opacity=0.5] (a) at (6,0){\small {\color{white}$\lambda$}};
        \node[ line width=0.6pt, dashed, draw opacity=0.5] (a) at (-1.4,0){\small {\color{black} $S(\lambda^{\cone})$ } };
                \node[ line width=0.6pt, dashed, draw opacity=0.5] (a) at (7.4,0){\small {\color{black} $\lambda^{(2''')}$ } };
        \node[ line width=0.6pt, dashed, draw opacity=0.5] (a) at (1.05,0){\small {\color{white} $\mu$ } };
                \node[ line width=0.6pt, dashed, draw opacity=0.5] (a) at (3.05,0){\small {\color{white} $\mu$ } };
                        \node[ line width=0.6pt, dashed, draw opacity=0.5] (a) at (5.05,0){\small {\color{white} $\mu$ } };
         \node[ line width=0.6pt, dashed, draw opacity=0.5] (a) at (0,1){\small {\color{black} $\lambda^{(3)}$ } };
        \node[ line width=0.6pt, dashed, draw opacity=0.5] (a) at (2,1){\small {\color{black} $\lambda^{(3')}$ } };
         \node[ line width=0.6pt, dashed, draw opacity=0.5] (a) at (4,1){\small {\color{black} $\lambda^{(3'')}$ } };
         \node[ line width=0.6pt, dashed, draw opacity=0.5] (a) at (6,1){\small {\color{black} $\lambda^{(3''')}$ } };
        \node[ line width=0.6pt, dashed, draw opacity=0.5] (a) at (0.7,1.5){\small {\color{black} $\lambda^{(2)}S(\lambda^{(1')})$ } };
        \node[ line width=0.6pt, dashed, draw opacity=0.5] (a) at (3,1.5){\small {\color{black} $\lambda^{(2')}S(\lambda^{(1'')})$ } };
        \node[ line width=0.6pt, dashed, draw opacity=0.5] (a) at (5.3,1.5){\small {\color{black} $\lambda^{(2'')}S(\lambda^{(1''')})$ } };
        \end{tikzpicture}
    \end{aligned}
\end{equation}
where we use the prime symbol to distinguish the component obtained from the comultiplication of Haar integral on different odd vertices.

Let us now derive the local relations for the action of Hopf Pauli-X and Pauli-Z on odd and even vertex tensors. These results will be crucial tools for studying the symmetry of the 1d Hopf cluster chain.
Notice that for Haar integral of $\cA$, we have
\begin{equation}
    \sum_{(\lambda)}  S(\lambda^{\cone})\otimes \lambda^{(2)} x=  \sum_{(\lambda)}  xS(\lambda^{\cone})\otimes \lambda^{\ctwo}. \label{eq:idempLambda}
\end{equation}
This further implies that 
\begin{equation}
\begin{aligned}
    \sum_{(\lambda)} S(\lambda^{\cone})\otimes \lambda^{\ctwo}\otimes \lambda^{\cthree} x 
    = \sum_{(\lambda)} 
 x^{\ctwo} S(\lambda^{\cone})\otimes \lambda^{\ctwo} S(x^{\cone} )\otimes \lambda^{\cthree}.
\end{aligned}
\end{equation} 
Substituting $x=S(h)$, we obtain 
\begin{equation}
    \begin{aligned}
   \sum_{(\lambda)} S(\lambda^{\cone})\otimes \lambda^{\ctwo}\otimes \lambda^{\cthree} S(h) 
    = \sum_{(\lambda)} 
 S(h^{\cone}) S(\lambda^{\cone})\otimes \lambda^{\ctwo} h^{\ctwo} \otimes \lambda^{\cthree}.
\end{aligned}
\end{equation}
This can be represented using the Hopf tensor network as
\begin{equation}
        \begin{aligned}
        \begin{tikzpicture}  
        \draw[line width=.6pt,black] (0,0) -- (0,1.3);
        \draw[line width=.6pt,black] (0,0) -- (0.8,0);
        \draw[line width=.6pt,black] (0,0) -- (-0.8,0);
        \draw[black,fill=red] (0,0) circle (0.2);   
        \draw[fill=white]  (-0.3,0.4) rectangle ++(.6,.6);
        \node[ line width=0.6pt, dashed, draw opacity=0.5] (a) at (0,0.73){\small $\XL_x$};
        \node[ line width=0.6pt, dashed, draw opacity=0.5] (a) at (0,0){\small {\color{white}$\lambda$}};
        \end{tikzpicture}
    \end{aligned}
    =
     \begin{aligned}
        \begin{tikzpicture}  
        \draw[line width=.6pt,black] (0,0) -- (0,0.8);
        \draw[line width=.6pt,black] (0,0) -- (1.8,0);
        \draw[line width=.6pt,black] (0,0) -- (-1.8,0);
        \draw[black,fill=red] (0,0) circle (0.2);   
        \draw[fill=white]  (-1.4,-0.35) rectangle ++(1.,.7);
        \draw[fill=white]  (.4,-0.35) rectangle ++(1.,.7);
        \node[ line width=0.6pt, dashed, draw opacity=0.5] (a) at (-0.9,0){\small $\tilde{\XR}_{x^{\cone}}$};
        \node[ line width=0.6pt, dashed, draw opacity=0.5] (a) at (0.9,0){\small $\tilde{\XL}_{x^{\ctwo}}$};
        \node[ line width=0.6pt, dashed, draw opacity=0.5] (a) at (0,0){\small {\color{white}$\lambda$}};
        \end{tikzpicture}
    \end{aligned}. \label{eq:TNX1}
\end{equation}
Since $S(\lambda) = \lambda$, substituting it into Eq.~\eqref{eq:idempLambda} and applying the twist map $\tau: w \otimes v \mapsto v \otimes w$, we have
\begin{equation}
  \sum_{(\lambda)}  S(\lambda^{\cone})x\otimes \lambda^{\ctwo}=\sum_{(\lambda)} S(\lambda^{\cone})\otimes x \lambda^{\ctwo}.
\end{equation}
This further implies that
\begin{equation}
 \sum_{(\lambda)}   S(\lambda^{\cone})x^{\ctwo} \otimes S(x^{\cone}) \lambda^{\ctwo} \otimes \lambda^{\cthree} =\sum_{(\lambda)} S(\lambda^{\cone})\otimes \lambda^{\ctwo} \otimes x \lambda^{\cthree}.
\end{equation}
This can be represented using the Hopf tensor network as
\begin{equation}
        \begin{aligned}
        \begin{tikzpicture}  
        \draw[line width=.6pt,black] (0,0) -- (0,1.3);
        \draw[line width=.6pt,black] (0,0) -- (0.8,0);
        \draw[line width=.6pt,black] (0,0) -- (-0.8,0);
        \draw[black,fill=red] (0,0) circle (0.2);   
        \draw[fill=white]  (-0.3,0.4) rectangle ++(.6,.6);
        \node[ line width=0.6pt, dashed, draw opacity=0.5] (a) at (0,0.73){\small $\XR_x$};
        \node[ line width=0.6pt, dashed, draw opacity=0.5] (a) at (0,0){\small {\color{white}$\lambda$}};
        \end{tikzpicture}
    \end{aligned}
    =
     \begin{aligned}
        \begin{tikzpicture}  
        \draw[line width=.6pt,black] (0,0) -- (0,0.8);
        \draw[line width=.6pt,black] (0,0) -- (1.8,0);
        \draw[line width=.6pt,black] (0,0) -- (-1.8,0);
        \draw[black,fill=red] (0,0) circle (0.2);   
        \draw[fill=white]  (-1.4,-0.35) rectangle ++(1.,.7);
        \draw[fill=white]  (.4,-0.35) rectangle ++(1.,.7);
        \node[ line width=0.6pt, dashed, draw opacity=0.5] (a) at (-0.9,0){\small $\tilde{\XL}_{x^{\ctwo}}$};
        \node[ line width=0.6pt, dashed, draw opacity=0.5] (a) at (0.9,0){\small $\tilde{\XR}_{x^{\cone}}$};
        \node[ line width=0.6pt, dashed, draw opacity=0.5] (a) at (0,0){\small {\color{white}$\lambda$}};
        \end{tikzpicture}
    \end{aligned}. \label{eq:TNX2}
\end{equation}
Using the facts $\tilde{\XR}_x = {\XR}_{S(x)}$ and $\tilde{\XL}_x = \XL_{S^{-1}(x)}$, we can obtain expressions for the action of $\tilde{\XL}$ and $\tilde{\XR}$ on odd tensors similar to Eqs.~\eqref{eq:TNX1} and \eqref{eq:TNX2}.

For even vertex tensors, the Hopf Pauli-X actions have simpler forms, since the odd tensor are multiplication map
\begin{equation}
        \begin{aligned}
        \begin{tikzpicture}  
        \draw[line width=.6pt,black] (0,0) -- (0,1.3);
        \draw[line width=.6pt,black] (0,0) -- (0.8,0);
        \draw[line width=.6pt,black] (0,0) -- (-0.8,0);
        %\draw[black,fill=red] (0,0) circle (0.2);   
        \draw[fill=teal]  (-0.2,-0.2) rectangle ++(.4,.4);
        \draw[fill=white]  (-0.3,0.4) rectangle ++(.6,.6);
        \node[ line width=0.6pt, dashed, draw opacity=0.5] (a) at (0,0.73){\small $\XR_x$};
        \node[ line width=0.6pt, dashed, draw opacity=0.5] (a) at (0,0){\small {\color{white}$\mu$}};
        \end{tikzpicture}
    \end{aligned}
    =
     \begin{aligned}
        \begin{tikzpicture}  
        \draw[line width=.6pt,black] (0,0) -- (0,0.8);
        \draw[line width=.6pt,black] (0,0) -- (0.8,0);
        \draw[line width=.6pt,black] (0,0) -- (-1.8,0);
        \draw[fill=teal]  (-0.2,-0.2) rectangle ++(.4,.4);
        \draw[fill=white]  (-1.2,-0.3) rectangle ++(.6,.6);
      %  \draw[fill=white]  (.4,-0.35) rectangle ++(1.,.7);
        \node[ line width=0.6pt, dashed, draw opacity=0.5] (a) at (-0.9,0){\small ${\XR}_{x}$};
       % \node[ line width=0.6pt, dashed, draw opacity=0.5] (a) at (0.9,0){\small $\tilde{\XR}_{x^{\cone}}$};
        \node[ line width=0.6pt, dashed, draw opacity=0.5] (a) at (0,0){\small {\color{white}$\mu$}};
        \end{tikzpicture}
    \end{aligned}, \label{eq:TNX3}
\end{equation}
\begin{equation}
        \begin{aligned}
        \begin{tikzpicture}  
        \draw[line width=.6pt,black] (0,0) -- (0,1.3);
        \draw[line width=.6pt,black] (0,0) -- (0.8,0);
        \draw[line width=.6pt,black] (0,0) -- (-0.8,0);
        %\draw[black,fill=red] (0,0) circle (0.2);   
        \draw[fill=teal]  (-0.2,-0.2) rectangle ++(.4,.4);
        \draw[fill=white]  (-0.3,0.4) rectangle ++(.6,.6);
        \node[ line width=0.6pt, dashed, draw opacity=0.5] (a) at (0,0.73){\small $\XL_x$};
        \node[ line width=0.6pt, dashed, draw opacity=0.5] (a) at (0,0){\small {\color{white}$\mu$}};
        \end{tikzpicture}
    \end{aligned}
    =
     \begin{aligned}
        \begin{tikzpicture}  
        \draw[line width=.6pt,black] (0,0) -- (0,0.8);
        \draw[line width=.6pt,black] (0,0) -- (1.8,0);
        \draw[line width=.6pt,black] (0,0) -- (-.8,0);
        \draw[fill=teal]  (-0.2,-0.2) rectangle ++(.4,.4);
      %  \draw[fill=white]  (-1.4,-0.35) rectangle ++(.8,.7);
        \draw[fill=white]  (.6,-0.3) rectangle ++(.6,.6);
      %  \node[ line width=0.6pt, dashed, draw opacity=0.5] (a) at (-0.9,0){\small ${\XL}_{x}$};
        \node[ line width=0.6pt, dashed, draw opacity=0.5] (a) at (0.9,0){\small ${\XL}_{x}$};
        \node[ line width=0.6pt, dashed, draw opacity=0.5] (a) at (0,0){\small {\color{white}$\mu$}};
        \end{tikzpicture}
    \end{aligned}. \label{eq:TNX4}
\end{equation}
For $\tilde{\XR}_x = {\XR}_{S(x)}$ and $\tilde{\XL}_x = \XL_{S^{-1}(x)}$, similar expressions exist.

The action of Hopf Pauli-Z on odd and even vertex tensors can also be derived using the property of Hopf algebra. For odd vertex tensor, we have
\begin{equation}
        \begin{aligned}
        \begin{tikzpicture}  
        \draw[line width=.6pt,black] (0,0) -- (0,1.3);
        \draw[line width=.6pt,black] (0,0) -- (0.8,0);
        \draw[line width=.6pt,black] (0,0) -- (-0.8,0);
        \draw[black,fill=red] (0,0) circle (0.2);   
        \draw[line width=.6pt,cyan] (-.6,0.7) -- (0.6,0.7);
        \draw[fill=white]  (-0.3,0.4) rectangle ++(.6,.6);
        \node[ line width=0.6pt, dashed, draw opacity=0.5] (a) at (0,0.73){\small $Z_{\Gamma}$};
        \node[ line width=0.6pt, dashed, draw opacity=0.5] (a) at (0,0){\small {\color{white}$\lambda$}};
        \end{tikzpicture}
    \end{aligned}
    =
     \begin{aligned}
        \begin{tikzpicture}  
        \draw[line width=.6pt,black] (0,0) -- (0,0.8);
        \draw[line width=.6pt,black] (0,0) -- (0.8,0);
        \draw[line width=.6pt,black] (0,0) -- (-1.8,0);
        \draw[black,fill=red] (0,0) circle (0.2);   
        \draw[line width=.6pt,cyan] (-1.5,0.1) -- (-0.3,0.1);
        \draw[fill=white]  (-1.2,-0.3) rectangle ++(.6,.6);
      %  \draw[fill=white]  (.4,-0.35) rectangle ++(1.,.7);
        \node[ line width=0.6pt, dashed, draw opacity=0.5] (a) at (-0.9,0){\small $\tilde{Z}_{\Gamma}^{\ddagger}$};
       % \node[ line width=0.6pt, dashed, draw opacity=0.5] (a) at (0.9,0){\small $\tilde{\XR}_{x^{\cone}}$};
        \node[ line width=0.6pt, dashed, draw opacity=0.5] (a) at (0,0){\small {\color{white}$\lambda$}};
        \end{tikzpicture}
    \end{aligned}, \label{eq:TNZ3}
\end{equation}
where we have used $\sum_{(\lambda)} \lambda^{\cone}\otimes \lambda^{\ctwo}\otimes \lambda^{\cthree}\otimes \lambda^{(4)}=\sum_{(\lambda)} \lambda^{\ctwo}\otimes \lambda^{\cthree}\otimes \lambda^{(4)}\otimes \lambda^{(1)}$ and then apply $S\otimes \id \otimes \id \otimes \Gamma$, and we use cyan leg to represent the virtual indices that live in the representation space of $\cA$.
Similarly, we have
\begin{equation}
        \begin{aligned}
        \begin{tikzpicture}  
        \draw[line width=.6pt,black] (0,0) -- (0,1.3);
        \draw[line width=.6pt,black] (0,0) -- (0.8,0);
        \draw[line width=.6pt,black] (0,0) -- (-0.8,0);
        \draw[black,fill=red] (0,0) circle (0.2);   
    \draw[line width=.6pt,cyan] (-.6,0.7) -- (0.6,0.7);
        \draw[fill=white]  (-0.3,0.4) rectangle ++(.6,.6);
        \node[ line width=0.6pt, dashed, draw opacity=0.5] (a) at (0,0.73){\small $Z_{\Gamma}^{\ddagger}$};
        \node[ line width=0.6pt, dashed, draw opacity=0.5] (a) at (0,0){\small {\color{white}$\lambda$}};
        \end{tikzpicture}
    \end{aligned}
     =
  \begin{aligned}
        \begin{tikzpicture}  
        \draw[line width=.6pt,black] (0,0) -- (0,0.8);
        \draw[line width=.6pt,black] (0,0) -- (1.8,0);
        \draw[line width=.6pt,black] (0,0) -- (-.8,0);
       \draw[black,fill=red] (0,0) circle (0.2);   
        \draw[line width=.6pt,cyan] (1.5,0.1) -- (0.3,0.1);
        \draw[fill=white]  (.6,-0.3) rectangle ++(.6,.6);
      %  \node[ line width=0.6pt, dashed, draw opacity=0.5] (a) at (-0.9,0){\small ${\XL}_{x}$};
        \node[ line width=0.6pt, dashed, draw opacity=0.5] (a) at (0.9,0){\small $\tilde{Z}_{\Gamma}^{\ddagger}$};
        \node[ line width=0.6pt, dashed, draw opacity=0.5] (a) at (0,0){\small {\color{white}$\lambda$}};
        \end{tikzpicture}
    \end{aligned}, \label{eq:TNZ4}
\end{equation}

\begin{equation}
        \begin{aligned}
        \begin{tikzpicture}  
        \draw[line width=.6pt,black] (0,0) -- (0,0.8);
        \draw[line width=.6pt,black] (0,0) -- (0.8,0);
        \draw[line width=.6pt,black] (0,0) -- (-1.8,0);
        \draw[black,fill=red] (0,0) circle (0.2);   
        \draw[line width=.6pt,cyan] (-1.5,0.1) -- (-0.3,0.1);
        \draw[fill=white]  (-1.2,-0.3) rectangle ++(.6,.6);
      %  \draw[fill=white]  (.4,-0.35) rectangle ++(1.,.7);
        \node[ line width=0.6pt, dashed, draw opacity=0.5] (a) at (-0.9,0){\small ${Z}_{\Gamma}^{\ddagger}$};
       % \node[ line width=0.6pt, dashed, draw opacity=0.5] (a) at (0.9,0){\small $\tilde{\XR}_{x^{\cone}}$};
        \node[ line width=0.6pt, dashed, draw opacity=0.5] (a) at (0,0){\small {\color{white}$\lambda$}};
        \end{tikzpicture}
    \end{aligned}
     =
  \begin{aligned}
        \begin{tikzpicture}  
        \draw[line width=.6pt,black] (0,0) -- (0,0.8);
        \draw[line width=.6pt,black] (0,0) -- (1.8,0);
        \draw[line width=.6pt,black] (0,0) -- (-.8,0);
       \draw[black,fill=red] (0,0) circle (0.2);   
        \draw[line width=.6pt,cyan] (1.5,0.1) -- (0.3,0.1);
        \draw[fill=white]  (.6,-0.3) rectangle ++(.6,.6);
      %  \node[ line width=0.6pt, dashed, draw opacity=0.5] (a) at (-0.9,0){\small ${\XL}_{x}$};
        \node[ line width=0.6pt, dashed, draw opacity=0.5] (a) at (0.9,0){\small $\tilde{Z}_{\Gamma}$};
        \node[ line width=0.6pt, dashed, draw opacity=0.5] (a) at (0,0){\small {\color{white}$\lambda$}};
        \end{tikzpicture}
    \end{aligned}, \label{eq:TNZZ1}
\end{equation}

\begin{equation}
        \begin{aligned}
        \begin{tikzpicture}  
        \draw[line width=.6pt,black] (0,0) -- (0,0.8);
        \draw[line width=.6pt,black] (0,0) -- (0.8,0);
        \draw[line width=.6pt,black] (0,0) -- (-1.8,0);
        \draw[black,fill=red] (0,0) circle (0.2);   
        \draw[line width=.6pt,cyan] (-1.5,0.1) -- (-0.3,0.1);
        \draw[fill=white]  (-1.2,-0.3) rectangle ++(.6,.6);
      %  \draw[fill=white]  (.4,-0.35) rectangle ++(1.,.7);
        \node[ line width=0.6pt, dashed, draw opacity=0.5] (a) at (-0.9,0){\small $\tilde{Z}_{\Gamma}$};
       % \node[ line width=0.6pt, dashed, draw opacity=0.5] (a) at (0.9,0){\small $\tilde{\XR}_{x^{\cone}}$};
        \node[ line width=0.6pt, dashed, draw opacity=0.5] (a) at (0,0){\small {\color{white}$\lambda$}};
        \end{tikzpicture}
    \end{aligned}
     =
  \begin{aligned}
        \begin{tikzpicture}  
        \draw[line width=.6pt,black] (0,0) -- (0,0.8);
        \draw[line width=.6pt,black] (0,0) -- (1.8,0);
        \draw[line width=.6pt,black] (0,0) -- (-.8,0);
       \draw[black,fill=red] (0,0) circle (0.2);   
        \draw[line width=.6pt,cyan] (1.5,0.1) -- (0.3,0.1);
        \draw[fill=white]  (.6,-0.3) rectangle ++(.6,.6);
      %  \node[ line width=0.6pt, dashed, draw opacity=0.5] (a) at (-0.9,0){\small ${\XL}_{x}$};
        \node[ line width=0.6pt, dashed, draw opacity=0.5] (a) at (0.9,0){\small ${Z}_{\Gamma}^{\ddagger}$};
        \node[ line width=0.6pt, dashed, draw opacity=0.5] (a) at (0,0){\small {\color{white}$\lambda$}};
        \end{tikzpicture}
    \end{aligned}. \label{eq:TNZZ2}
\end{equation}
For even vertex tensor, we have
\begin{equation}
        \begin{aligned}
        \begin{tikzpicture}  
        \draw[line width=.6pt,black] (0,0) -- (0,1.3);
        \draw[line width=.6pt,black] (0,0) -- (0.8,0);
        \draw[line width=.6pt,black] (0,0) -- (-0.8,0);
            \draw[line width=.6pt,cyan] (0.8,0.7) -- (-0.8,0.7);
                \draw[fill=teal]  (-0.2,-0.2) rectangle ++(.4,.4);
        \draw[fill=white]  (-0.3,0.4) rectangle ++(.6,.6);
        \node[ line width=0.6pt, dashed, draw opacity=0.5] (a) at (0,0.73){\small $Z_{\Gamma}$};
        \node[ line width=0.6pt, dashed, draw opacity=0.5] (a) at (0,0){\small {\color{white}$\mu$}};
        \end{tikzpicture}
    \end{aligned}
    =
     \begin{aligned}
        \begin{tikzpicture}  
        \draw[line width=.6pt,black] (0,0) -- (0,1.2);
        \draw[line width=.6pt,black] (0,0) -- (1.4,0);
        \draw[line width=.6pt,black] (0,0) -- (-1.4,0);
                 \draw[line width=.6pt,cyan] (-1.4,0.7) -- (-0.05,0.7);
                      \draw[line width=.6pt,cyan] (1.4,0.7) -- (0.05,0.7);
        \draw[fill=teal]  (-0.2,-0.2) rectangle ++(.4,.4);
        \draw[fill=white]  (-1.,-0.2) rectangle ++(.6,1);
        \draw[fill=white]  (.4,-0.2) rectangle ++(.6,1);
        \node[ line width=0.6pt, dashed, draw opacity=0.5] (a) at (-0.75,0.3){\small $Z_{\Gamma}$};
        \node[ line width=0.6pt, dashed, draw opacity=0.5] (a) at (0.7,0.3){\small $Z_{\Gamma}$};
        \node[ line width=0.6pt, dashed, draw opacity=0.5] (a) at (0,0){\small {\color{white}$\mu$}};
        \end{tikzpicture}
    \end{aligned}, \label{eq:TNZ1}
\end{equation}
\begin{equation}
        \begin{aligned}
        \begin{tikzpicture}  
        \draw[line width=.6pt,black] (0,0) -- (0,1.3);
        \draw[line width=.6pt,black] (0,0) -- (0.8,0);
        \draw[line width=.6pt,black] (0,0) -- (-0.8,0);
            \draw[line width=.6pt,cyan] (0.8,0.7) -- (-0.8,0.7);
                \draw[fill=teal]  (-0.2,-0.2) rectangle ++(.4,.4);
        \draw[fill=white]  (-0.3,0.4) rectangle ++(.6,.6);
        \node[ line width=0.6pt, dashed, draw opacity=0.5] (a) at (0,0.73){\small $Z_{\Gamma}^{\ddagger}$};
        \node[ line width=0.6pt, dashed, draw opacity=0.5] (a) at (0,0){\small {\color{white}$\mu$}};
        \end{tikzpicture}
    \end{aligned}
    =
     \begin{aligned}
        \begin{tikzpicture}  
        \draw[line width=.6pt,black] (0,0) -- (0,1.2);
        \draw[line width=.6pt,black] (0,0) -- (1.4,0);
        \draw[line width=.6pt,black] (0,0) -- (-1.4,0);
        \draw[line width=.6pt,cyan] (-1.05,0.35) -- (-1.2,0.35) -- (-1.2,0.1) -- (-1,0.1);
                \draw[line width=.6pt,cyan] (-.35,0.35) -- (-0.05,0.35) ;
                               \draw[line width=.6pt,cyan] (.05,0.35) -- (1.2,0.35) -- (1.2,0.1) -- (1,0.1) ;
                        \draw[line width=.6pt,cyan] (0.4,-0.35) -- (-1.4,-0.35) ;
                      \draw[line width=.6pt,cyan] (1.4,0.6) -- (0.05,0.6);
  \draw[line width=.6pt,cyan] (-.4,0.6) -- (-0.05,0.6);
        \draw[fill=teal]  (-0.2,-0.2) rectangle ++(.4,.4);
        \draw[fill=white]  (-1.,-0.2) rectangle ++(.6,1);
        \draw[fill=white]  (.4,0.2) rectangle ++(.6,-0.8);
        \node[ line width=0.6pt, dashed, draw opacity=0.5] (a) at (-0.75,0.3){\small $Z_{\Gamma}^{\ddagger}$};
        \node[ line width=0.6pt, dashed, draw opacity=0.5] (a) at (0.7,-0.3){\small $Z_{\Gamma}^{\ddagger}$};
        \node[ line width=0.6pt, dashed, draw opacity=0.5] (a) at (0,0){\small {\color{white}$\mu$}};
        \end{tikzpicture}
    \end{aligned}, \label{eq:TNZ2}
\end{equation}
Using $\tilde{Z}_{\Gamma} = Z_{\Gamma(S^{-1}(\cdot))}$ and $\tilde{Z}_{\Gamma}^{\ddagger} = Z_{\Gamma(S(\cdot))}$, we obtain similar expressions for the actions of $\tilde{Z}_{\Gamma}$ and $\tilde{Z}_{\Gamma}^{\ddagger}$ on odd and even vertex tensors.

\section{One dimensional Hopf cluster state model}
\label{sec:Hopf1d}

In this section, we delve into the one-dimensional Hopf cluster state, a lattice theory that exhibits fusion category symmetry in (1+1)D. Distinct from anyonic chains, the Hopf cluster state model possesses a tensor product structure for its total Hilbert space, rendering it advantageous for the discussion of SPT phases in the context of quantum matter.
As we will see in Sec.~\ref{sec:QD}, the (1+1)D Hopf cluster state model can be regarded as a boundary theory of the (2+1)D Hopf quantum double model~\cite{Buerschaper2013a,jia2023boundary,Jia2023weak}.
This also sheds new light on the connection between the SPT phase and the topological phase: the categorical symmetry of $d$ dimensional SPT phase is related to the $d+1$ dimension topological order \cite{Kong2020algebraic}. These topics will be discussed in detail in Ref.~\cite{jia2024cluster}.

\subsection{Hopf cluster state lattice Hamiltonian}
Let us now consider a 1d cluster lattice $\mathbb{M}^1$, which is depicted in Fig.~\ref{fig:cluster-TN}. We will label the vertices as $1,\cdots,2L$.
The preferred state is $|\Omega\rangle=|\lambda\rangle_1|1_{\cA}\rangle_2 |\lambda\rangle_3\cdots |1_{\cA}\rangle_{2L}$.
For periodic boundary conditions, we assume that all edges of the cluster graph are directed from left to right. For an odd vertex $v_o$, the local ordering of the edge set $N_E(v_o)$ is assumed to be $\{e_1 = e_L, e_2 = e_R\}$, where $e_L$ and $e_R$ are the edges connecting to $v_o$ from the left and right sides, respectively.
The entangler operator is then given by
\begin{equation}\label{eq:ClusterU}
    U_E = \prod_{i: \text{odd}} C\XR_{i,i+1} C\XL_{i,i-1},
\end{equation}
which means $W_{v_i} = C\XR_{i,i+1} C\XL_{i,i-1}$ in Eq.~\eqref{eq:Wcluster}.
Notice that different choices of edge orientations and the local ordering of the edge set $N_E(v_o)$ result in different models, but their physical properties are the same. This will be discussed later in Sec.~\ref{sec:QD}.

For odd vertex $i$, we define X-type local stabilizer operator as ($\lambda$ is Haar integral)
\begin{equation}
    \Av_i=\sum_{(\lambda)} \XL_{\lambda^{\cone}}(i-1) \otimes \XR_{\lambda^{\cthree}}(i)\otimes \XR_{\lambda^{\ctwo}}(i+1). 
\end{equation}
We can similarly define $\Av^h$ for $h\in \cA$.
For even vertex $j$, we define the Z-type operator
\begin{equation}
    \Bf^{\Gamma}_j=\Tr' [Z^{\ddagger}_{\Gamma}(j-1)\otimes Z_{\Gamma}(j) \otimes Z_{\Gamma}(j+1)]
\end{equation}
and
\begin{equation}
    \Bf_j=\sum_{\Gamma\in \Irr(\cA)} \frac{d_{\Gamma}}{|\cA|} \Bf^{\Gamma}_j.
\end{equation}
As will be explained in Sec.~\ref{sec:QD}, $\Bf^{\Gamma}_j$ is defined based on the character function $\chi_{\Gamma}:\cA \to \mathbb{C}$ which is an element in the dual Hopf algebra $\bar{\cA}$. This can also be generalized to a general element $\psi\in \bar{\cA}$. We can define $\Bf^{\psi}$ using the comultiplication structure of $\bar{\cA}$.

\begin{remark}

Notice that $\Av_i$'s and $\Bf_j$'s bear a resemblance to the vertex and face operators $A_v$ and $B_f$ in the quantum double model \cite{Kitaev2003,Buerschaper2013a,chen2021ribbon,jia2022electricmagnetic,jia2023boundary,Jia2023weak}.
We can prove their equivalence, see Sec.~\ref{sec:QD}.
Refer to Fig.~\ref{fig:QDcluster}, folding the 1d cluster lattice into a zigzag configuration yields a 2d (quasi-1d) lattice. In the quantum double lattice construction, Hopf qudits are positioned on the edges. To ensure a properly defined quantum double lattice, an ancillary vertex (seen as the bottom vertex in Fig.~\ref{fig:QDcluster}) is introduced.
We'll term this lattice the virtual lattice, with its associated vertices and faces named virtual vertices and virtual faces, respectively.
Given that the directions of each virtual edge are chosen as downwards and rightwards, and the edges surrounding the faces and vertices are assumed to be clockwise, we assert that the corresponding vertex operator $A_v$ is equivalent to the odd-site operator $\Av_i$ defined earlier. Similarly, the face operator $B_f$ corresponds to the odd-site operator $\Bf_j$ defined previously. See Sec.~\ref{sec:QD} for a detailed discussion. 
\end{remark}

\begin{figure}
    \centering
    \includegraphics[width=12cm]{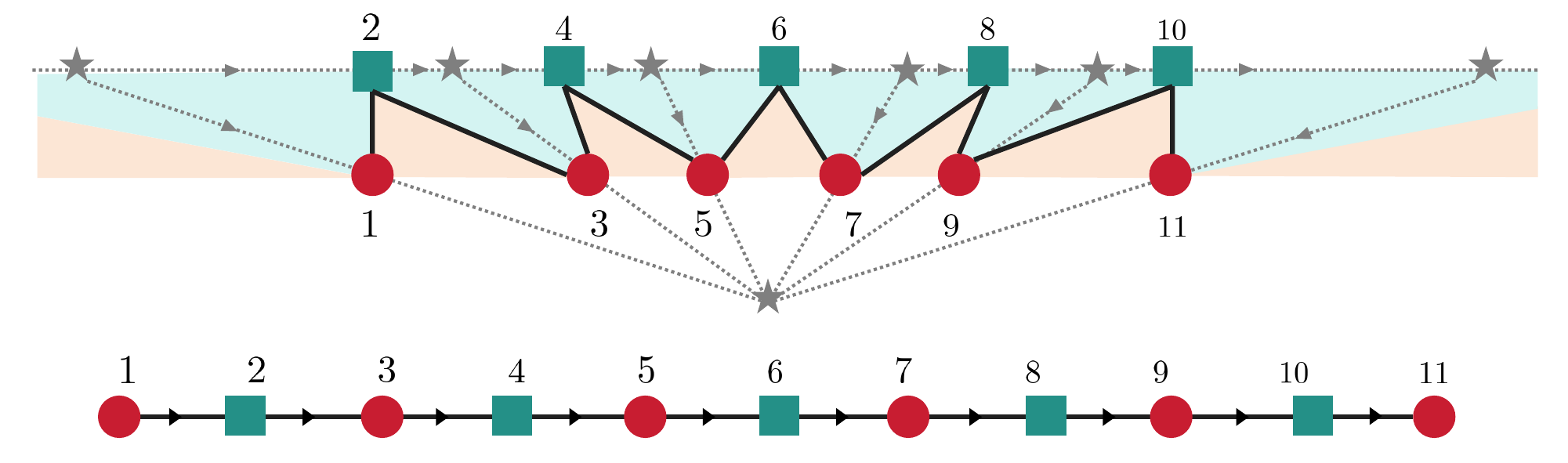}
    \caption{Illustration depicting the transformation from local stabilizers for the Hopf cluster state into the face and vertex operators of the Hopf quantum double model. By folding the 1d cluster lattice into a zigzag configuration, a 2d lattice is formed, represented by gray dotted lines in the figure, with vertices depicted as stars.}
    \label{fig:QDcluster}
\end{figure}

\begin{theorem}\label{thm:1dHopfLattice}
    The Hopf cluster state lattice Hamiltonian 
    \begin{equation}\label{eq:Hamiltonian1d}
        H_{\rm cluster}=-\sum_{i:\rm odd} \Av_i-\sum_{j:\rm even} \Bf_j
    \end{equation}
    is a local commutative projector (LCP) Hamiltonian, viz.,  $\Av_i,\Bf_j$ are projectors and $[\Av_i,\Bf_j]=0$ for all $i\in 2\mathbb{N}+1$ and even $j\in 2\mathbb{N}$, $[\Av_k,\Av_l]=0$ and $[\Bf_s,\Bf_t]=0$ for all $k,l\in 2\mathbb{N}+1$ and $s,t\in 2\mathbb{N}$. 
The Hopf cluster state $|\mathbb{M}^1,\cA\rangle$ is the stabilizer state of $\Av_i$ and $\Bf_j$
\begin{equation}
\Av_i|\mathbb{M}^1,\cA\rangle=\Bf_j|\mathbb{M}^1,\cA\rangle=|\mathbb{M}^1,\cA\rangle
\end{equation}
for all odd $i$ and even $j$.    
\end{theorem}
\begin{proof}
    We first prove that $\Av_i$'s are projectors. Recall that the inner product of Hopf qudit $\cA$ is given by Eq.~\eqref{eq:innerProd}, $\Lambda(x^*gy)=\Lambda((g^*x)^*y)$ implies that 
\begin{equation}
    \XR_g^{\dagger}=\XR_{g^*}.
\end{equation}
Using Eq.~\eqref{eq:Sstar} and cocomutativity of $\Lambda$, we have $\Lambda(x^*yS(g))=\Lambda(S(g)x^*y)=\Lambda((xS(g)^*)^*y)=\Lambda((xS(g^*))^*y)$, which implies
\begin{equation}
     \XL_g^{\dagger}=\XL_{g^*}.
\end{equation}
By definition of $\Av_i^{\lambda}$, we see
\begin{equation}
    (\Av_i^{\lambda})^{\dagger}=\Av_i^{\lambda^*}.
\end{equation}
Since $\lambda^*=\lambda$, we have $(\Av_i^{\lambda})^{\dagger}=\Av_i^{\lambda}$, viz., $\Av_i$'s are Hermitian.
It's clear that $\Av^{x}_i\Av^y_i=\Av^{xy}_i$. From $\lambda^2=\lambda$, we have $(\Av_i^{\lambda})^2=\Av_i^{\lambda}$, meaning that $\Av_i$'s are projectors.

Directly proving $\Bf^{\dagger}=\Bf$ and $\Bf^{2}=\Bf$ is challenging. However, leveraging Proposition~\ref{prop:AB-QD} in Appendix, we obtain $\Bf=B_f^{\Lambda}$. In references \cite{jia2023boundary,Jia2023weak}, it's shown that $B_f^{\Lambda}$ are Hermitian projectors. Hence, $\Bf_j$ are projectors for all $j\in 2\Nbb$.

Using Proposition~\ref{prop:AB-QD} in the next section, the commutativity of $[\Av_k,\Av_l]=[\Bf_s,\Bf_t]=[\Av_i,\Bf_j]=0$ can be proved in a similar way as that in Refs.~\cite{jia2023boundary,Jia2023weak}.
Notice that the character $\chi_{\Gamma}$ is cocommutative (i.e. $\chi_{\Gamma}(xy)=\chi_{\Gamma}(yx)$), which implies a stronger relation
\begin{equation}
    [\Av_i, \Bf_{j}^{\Gamma}]=0.
\end{equation}
Here we will provide an alternative proof of above equation based on the Hopf tensor network representation of the Hopf cluster state.
For a 1d lattice with $2L$ vertices, the total space is $\cA^{\otimes 2L}$. For product state $|\Omega(x_1,\cdots,x_{2L})\rangle=|x_i\rangle\otimes \cdots\otimes |x_{2L}\rangle$, which applying the cluster circuit $U_E$ in Eq.~\eqref{eq:ClusterU}, we obtain a state $|\Psi(x_1,\cdots,x_{2L})\rangle$. Since $U_E$ is invertible, its necessary to show that $\Av_i\Bf_{j}^{\Gamma}|\Psi(x_1,\cdots,x_{2L})\rangle= \Bf_j^{\Gamma} \Av_i|\Psi(x_1,\cdots,x_{2L})\rangle$ for all such $|\Psi(x_1,\cdots,x_{2L})\rangle$.
The only non-trivial case is when $j=i\pm 1$, let's assume $j=i+1$ without loss of generality.
We have the following Hopf tensor network representation of $|\Psi(x_1,\cdots,x_{2L})\rangle$:
\begin{equation}\label{eq:TENSORPROOF}
    \begin{aligned}
        \begin{tikzpicture}  
        \draw[line width=.6pt,black] (0.3,0) -- (0.3,0.5);
        \draw[line width=.6pt,black]  (0.3,0.63) -- (0.3,2);
        \draw[line width=.6pt,black] (0,0.3)-- (0,0.5) -- (2,1);
        \draw[line width=.6pt,black] (-0.3,0)-- (-0.3,0.5) -- (-0.5,1);
        \draw[line width=.6pt,black] (4.3,0) -- (4.3,0.5);
        \draw[line width=.6pt,black]  (4.3,0.63) -- (4.3,2);
        \draw[line width=.6pt,black] (4,0.3)-- (4,0.5) -- (6,1);
        \draw[line width=.6pt,black] (3.7,0)--(3.7,0.5) --(2,1);
        \draw[line width=.6pt,black] (2,0)--(2,2);
        \draw[line width=.6pt,black] (7.7,0)--(7.7,0.5) --(6,1);
        \draw[line width=.6pt,black] (6,0)--(6,2);
        \draw[line width=.6pt,black] (8.3,0)--(8.3,0.75);
        \draw[line width=.6pt,black] (8.3,2)--(8.3,0.85);

        \draw[line width=.6pt,black] (8,0)-- (8,0.5) -- (8.5,1);
        \draw[black,fill=lightgray] (0,0) circle (0.35);   
        \draw [fill=lightgray] (2,0) circle (0.35);
        \draw [fill=lightgray] (4,0) circle (0.35);
        \draw [fill=lightgray] (6,0) circle (0.35);
        \draw [fill=lightgray] (8,0) circle (0.35);
        \draw[fill=lightgray]  (1.75,0.9) -- (2.25,0.9) -- (2,1.2) -- cycle;
        \draw[fill=lightgray]  (5.75,0.9) -- (6.25,0.9) -- (6,1.2) -- cycle;
        \draw[fill=lightgray]  (3.1,0.9) rectangle ++(-.4,-0.4) ;
        \draw[fill=lightgray]  (7.1,0.9) rectangle ++(-.4,-0.4) ;
        \node[ line width=0.6pt, dashed, draw opacity=0.5] (a) at (2.9,0.7){\small $S$};
        \node[ line width=0.6pt, dashed, draw opacity=0.5] (a) at (6.9,0.7){\small $S$};
        \node[ line width=0.6pt, dashed, draw opacity=0.5] (a) at (0,0){\small $x_i$};
        \node[ line width=0.6pt, dashed, draw opacity=0.5] (a) at (2,0){\small $x_{i+1}$};
        \node[ line width=0.6pt, dashed, draw opacity=0.5] (a) at (4,0){\small $x_{i+2}$};
        \node[ line width=0.6pt, dashed, draw opacity=0.5] (a) at (6,0){\small $x_{i+3}$};
        \node[ line width=0.6pt, dashed, draw opacity=0.5] (a) at (8,0){\small $x_{i+4}$};
        \node[ line width=0.6pt, dashed, draw opacity=0.5] (a) at (-0.8,1){$x_i^{\cone}$};
        \node[ line width=0.6pt, dashed, draw opacity=0.5] (a) at (-0.1,1.8){$x_i^{\cthree}$};
        \node[ line width=0.6pt, dashed, draw opacity=0.5] (a) at (2,2.4){$x_i^{\ctwo}x_{i+1}S(x_{i+2}^{\cone})$};
        \node[ line width=0.6pt, dashed, draw opacity=0.5] (a) at (4,2.4){$x_{i+2}^{\cthree}$};
         \node[ line width=0.6pt, dashed, draw opacity=0.5] (a) at (6,2.4){$x_{i+2}^{\ctwo}x_{i+3}S(x_{i+4}^{\cone})$};
        \node[ line width=0.6pt, dashed, draw opacity=0.5] (a) at (8,2.4){$x_{i+4}^{\cthree}$};
        \node[ line width=0.6pt, dashed, draw opacity=0.5] (a) at (9,1){$x_{i+4}^{\ctwo}$};
        \end{tikzpicture}
    \end{aligned}
\end{equation}
When applying $\Av_{i+2}$, from $\lambda x_{i+2}=\varepsilon(x_{i+2})\lambda$, we have
\begin{equation}
   \Av_{i+2} |\Psi(\cdots,x_{i},x_{i+1},x_{i+2},x_{i+3},x_{i+4},\cdots)\rangle= \varepsilon(x_{i+2}) |\Psi(\cdots,x_{i},x_{i+1},\lambda,x_{i+3},x_{i+4},\cdots)\rangle\label{eq:Avact}
\end{equation}
Applying $\Bf_{i+1}^{\Gamma}$ on  $|\Psi(\cdots,x_{i},x_{i+1},x_{i+2},x_{i+3},x_{i+4},\cdots)\rangle$ gives
\begin{align}
   & \Tr' \Gamma[S(x_{i}^{(4)})x_{i}^{\cthree} x_{i+1}^{\ctwo} S(x_{i+2}^{\cone}) x_{i+2}^{(5)}]
    |x_{i}^{\cone},x_i^{(5)}, x_{i}^{(2)}x_{i+1}^{\cone} S(x_{i+2}^{\ctwo}),x_{i+2}^{(4)},x_{i+2}^{(3)}x_{i+3}S(x_{1+4}^{\cone}),x_{i+4}^{\cthree},x_{i+4}^{\ctwo}\rangle\nonumber\\
  &  = \Tr' \Gamma[x_{i+1}^{\ctwo} S(x_{i+2}^{\cone}) x_{i+2}^{(5)}]  |x_{i}^{\cone},x_i^{(3)}, x_{i}^{(2)}x_{i+1}^{\cone} S(x_{i+2}^{\ctwo}),x_{i+2}^{(4)},x_{i+2}^{(3)}x_{i+3}S(x_{1+4}^{\cone}),x_{i+4}^{\cthree},x_{i+4}^{\ctwo}\rangle \label{eq:Bfact}
\end{align}
where we omit the sum symbol for comultiplication to avoid cluttering the equations.
Combining the fact that 
\begin{equation}\label{eq:lambdaPerm}
\begin{aligned}
   & \sum_{(\lambda)}\lambda^{(1)}\otimes \lambda^{(2)}\otimes \lambda^{(3)}\otimes \lambda^{(4)}\otimes \lambda^{(5)}
    = &\sum_{(\lambda)}\lambda^{(2)}\otimes \lambda^{(3)}\otimes \lambda^{(4)}\otimes \lambda^{(5)}\otimes \lambda^{(1)}
\end{aligned} 
\end{equation}
and Eqs.~\eqref{eq:Avact}  and \eqref{eq:Bfact} we obtain
\begin{align}
&\Bf_{i+1}^{\Gamma}\Av_{1+2}|\Psi(\cdots,x_{i},x_{i+1},x_{i+2},x_{i+3},x_{i+4},\cdots)\rangle\nonumber\\
=& \varepsilon(x_{i+2}) \Bf_{i+1}^{\Gamma}|\Psi(\cdots,x_{i},x_{i+1},\lambda,x_{i+3},x_{i+4},\cdots)\rangle\nonumber\\
=&\Tr' \Gamma[x_{i+1}^{\ctwo}] \varepsilon(x_{i+2}) |\Psi(\cdots,x_{i},x_{i+1}^{(1)},x_{i+2},x_{i+3},x_{i+4},\cdots)\rangle
\end{align}
In a similar way, we can prove that, applying $\Av_{1+2}\Bf_{i+1}^{\Gamma}$ on $|\Psi(\cdots,x_{i},x_{i+1},x_{i+2},x_{i+3},x_{i+4},\cdots)\rangle$ gives the same result.

Notice that Eqs.~\eqref{eq:Avact}  and \eqref{eq:Bfact} also implies that $|\mathbb{M}^1,\cA\rangle$ is invariant under $\Av_i$ and $\Bf_j$. Since $\varepsilon(\lambda)=1$, we have 
\begin{equation}
\Av_i|\mathbb{M}^1,\cA\rangle=\varepsilon(\lambda)|\mathbb{M}^1,\cA\rangle=|\mathbb{M}^1,\cA\rangle.
\end{equation}
Combining Eqs.~\eqref{eq:Bfact} and \eqref{eq:lambdaPerm}, we have
\begin{equation}
    \Bf_j|\mathbb{M}^1,\cA\rangle= \Tr \Gamma(1_{\cA}) |\mathbb{M}^1,\cA\rangle=d_{\Gamma} |\mathbb{M}^1,\cA\rangle.
\end{equation}
This, together with the Peter-Weyl formula for Hopf algebra, further implies that $\Bf_j|\mathbb{M}^1,\cA\rangle=|\mathbb{M}^1,\cA\rangle$.
\end{proof}

\subsection{Non-invertible symmetries}
The Hopf cluster state Hamiltonian in Eq.~\eqref{eq:Hamiltonian1d} has two non-invertible global symmetries characterized by $\cA$ and $\Rep(\cA)^{\rm rev}$ for the open boundary condition.
Notice that the symmetry $\Rep(\cA)^{\rm rev}$ is in fact the Grothendieck ring $\operatorname{Gr}(\Rep(\cA)^{\rm rev})$ of $\Rep(\cA)^{\rm rev}$, which has the basis as characters for irreducible representations of $\cA$. Since character $\chi_{\Gamma}$ is in the dual Hopf algebra $\bar{\cA}$, we have an embedding
\begin{equation}
   \operatorname{Gr}(\Rep(\cA)^{\rm rev}) \hookrightarrow  \bar{\cA}^{\rm op}, 
\end{equation}
where superscript `$\textrm{op}$' means that we take opposite algebra, viz., $a\cdot^{\rm op} b:= b\cdot a$.
This implies that the symmetry algebra can be embedded into $\cA\times \bar{\cA}$.

Let us first consider the Hopf $\cA$-symmetry, which is defined as an operator acting on all odd vertices:
\begin{equation}
    F_h=\sum_{(h)} \XL_{h^{(1)}}(1)\otimes \XL_{h^{(2)}}(3) \otimes \cdots \otimes \XL_{h^{(L)}}(2L-1).
\end{equation}
It's clear that
\begin{equation}
    F_gF_h=F_{gh},\quad F_{1_{\cA}}=I.
\end{equation}
Notice that this symmetry is not the same as on-site group symmetry, where $U_g = \otimes_j U_g^{(j)}$ and each local site carries the same group element $g$. Here, for different vertices, we assign a different component of the comultiplication of $h \in \cA$.

\begin{proposition} \label{prop:A-symm}
      The Hopf cluster state Hamiltonian $ H_{\rm cluster}$ in Eq.~\eqref{eq:Hamiltonian1d} commutes with $F_g$ for all $g\in \cA$, thus it has a Hopf $\cA$-symmetry.
\end{proposition}

\begin{proof}
Notice that $\Av$, when acting on Hopf qudit for odd vertices, is a action from the left-hand side, while  $F_g$ acts from the right-hand side, thus $[F_g,\Av]=0$.
To prove $[\Bf,F_g]=0$, consider Hopf tensor network in Eq.~\eqref{eq:TENSORPROOF}, applying $F_g$, we replace $x^{(3)}_i$ with $x^{(3)}S( g^{(\frac{i+1}{2})})$ and $x^{(3)}_{i+2}$ with $x^{(3)}_{i+2} S( g^{(\frac{i+1}{2} +1)})$. Then applying $B^{\Gamma}_{i+2}$ gives the factor
\begin{equation}
\begin{aligned}
    & \Tr' \Gamma[S^2(g^{(\frac{i+1}{2} +1)})S(x^{(4)}_i) \cdots  x_{i+2}^{(5)} S(g^{(\frac{i+1}{2} +2)}) ]\\
    = & \Tr'\Gamma[ S(g^{(\frac{i+1}{2} +2)}) g^{(\frac{i+1}{2} +1)} S(x^{(4)}_i) \cdots  x_{i+2}^{(5)}] \\
    =&\varepsilon(g^{(\frac{i+1}{2} +1)})\Tr'\Gamma[  S(x^{(4)}_i) \cdots  x_{i+2}^{(5)}].
\end{aligned}  
\end{equation}
This implies that $[\Bf_{i+1}^{\Gamma},F_g]=0$, leading us to the required result.
\end{proof}

Consider an open cluster chain with two odd vertices at the ends \footnote{For other cases of open chains that may have one odd vertex and one even vertex, or two even vertices at the ends, the analysis is similar.}, when acting $F_g$ on the corresponding Hopf cluster state, it is equivalent to acting $\sum_{(g)} \tilde{\XL}_{g^{(1)}} \otimes \tilde{\XR}_{g^{(2)}}$ on virtual legs at two ends. This can be derived using Eq.~\eqref{eq:TNX1}. Represented as Hopf tensor network, we obtain (the $\sum_{(g)}$ for comultiplication has be omitted):
\begin{align}
     \sum_{(g)}  \begin{aligned}
        \begin{tikzpicture}  
        \draw[line width=.6pt,black] (-0.8,0) -- (6.8,0);
        \draw[line width=.6pt,black] (0,0) -- (0.8,0);
        \draw[line width=.6pt,black] (0,0) -- (-0.8,0);
           \draw[line width=.6pt,blue] (0,0) -- (0,1.5);
        \draw[line width=.6pt,blue] (6,0) -- (6,1.5);
         \draw[line width=.6pt,blue] (1,0) -- (1,1.5);
        \draw[line width=.6pt,blue] (2,0) -- (2,1.5);
     %  \draw[line width=.6pt,blue] (3,0) -- (3,1.2);
       \draw[line width=.6pt,blue] (4,0) -- (4,1.5);
        \draw[line width=.6pt,blue] (5,0) -- (5,1.5);
         \draw[black,fill=red] (0,0) circle (0.2); 
         \draw[black,fill=white] (-0.3,0.5) rectangle ++ (0.8,0.6); 
             \draw[black,fill=white] (1.7,0.5) rectangle ++ (0.8,0.6); 
                 \draw[black,fill=white] (3.5,0.5) rectangle ++ (1.2,0.6); 
                     \draw[black,fill=white] (5.7,0.5) rectangle ++ (0.8,0.6); 
        \draw[fill=teal]  (0.8,-0.2) rectangle ++(.4,.4);
                \draw[line width=0pt, white,fill=white]  (2.7,-0.2) rectangle ++(.6,.4);
                        \draw[fill=teal]  (4.8,-0.2) rectangle ++(.4,.4);
                             \draw[black,fill=yellow] (-1,-.2) rectangle ++ (0.4,0.4); 
                                \draw[black,fill=yellow] (6.6,-.2) rectangle ++ (0.4,0.4); 
        \draw[black,fill=red] (2,0) circle (0.2);
                   \draw[black,fill=red] (4,0) circle (0.2); 
                        \draw[black,fill=red] (6,0) circle (0.2);
        \node[ line width=0.6pt, dashed, draw opacity=0.5] (a) at (0,0){\small {\color{white}$\lambda$}};
            \node[ line width=0.6pt, dashed, draw opacity=0.5] (a) at (2,0){\small {\color{white}$\lambda$}};
                \node[ line width=0.6pt, dashed, draw opacity=0.5] (a) at (4,0){\small {\color{white}$\lambda$}};
                    \node[ line width=0.6pt, dashed, draw opacity=0.5] (a) at (6,0){\small {\color{white}$\lambda$}};
        \node[ line width=0.6pt, dashed, draw opacity=0.5] (a) at (-0.7,0){\small {\color{black} $x$ } };
        \node[ line width=0.6pt, dashed, draw opacity=0.5] (a) at (6.9,0){\small {\color{black} $y$ } };
        \node[ line width=0.6pt, dashed, draw opacity=0.5] (a) at (1.05,0){\small {\color{white} $\mu$ } };
        \node[ line width=0.6pt, dashed, draw opacity=0.5] (a) at (3.05,0){\small {\color{black} $\cdots$ } };
        \node[ line width=0.6pt, dashed, draw opacity=0.5] (a) at (3.05,0.8){\small {\color{black} $\cdots$ } };
        \node[ line width=0.6pt, dashed, draw opacity=0.5] (a) at (5.05,0){\small {\color{white} $\mu$ } };
         \node[ line width=0.6pt, dashed, draw opacity=0.5] (a) at (0.15,0.8){\small {\color{black} $\XL_{g^{\cone}}$ } };
        \node[ line width=0.6pt, dashed, draw opacity=0.5] (a) at (2.15,0.8){\small {\color{black} $\XL_{g^{\ctwo}}$ } };
        \node[ line width=0.6pt, dashed, draw opacity=0.5] (a) at (4.15,0.8){\small {\color{black} $\XL_{g^{(n-1)}}$ } };
        \node[ line width=0.6pt, dashed, draw opacity=0.5] (a) at (6.15,0.8){\small {\color{black} $\XL_{g^{(n)}}$ } };
        \end{tikzpicture}
    \end{aligned}\nonumber\\
      = \sum_{(g)} \begin{aligned}
        \begin{tikzpicture}  
        \draw[line width=.6pt,black] (-2,0) -- (7.8,0);
        \draw[line width=.6pt,black] (0,0) -- (0.8,0);
        \draw[line width=.6pt,black] (0,0) -- (-0.8,0);
           \draw[line width=.6pt,blue] (0,0) -- (0,1.);
        \draw[line width=.6pt,blue] (6,0) -- (6,1.);
         \draw[line width=.6pt,blue] (1,0) -- (1,1.);
        \draw[line width=.6pt,blue] (2,0) -- (2,1.);
     %  \draw[line width=.6pt,blue] (3,0) -- (3,1.2);
       \draw[line width=.6pt,blue] (4,0) -- (4,1.);
        \draw[line width=.6pt,blue] (5,0) -- (5,1.);
         \draw[black,fill=red] (0,0) circle (0.2); 
        % \draw[black,fill=white] (-0.3,0.5) rectangle ++ (0.8,0.6); 
         %    \draw[black,fill=white] (1.7,0.5) rectangle ++ (0.8,0.6); 
          %       \draw[black,fill=white] (3.5,0.5) rectangle ++ (1.2,0.6); 
           %          \draw[black,fill=white] (5.7,0.5) rectangle ++ (0.8,0.6); 
        \draw[fill=teal]  (0.8,-0.2) rectangle ++(.4,.4);
                \draw[line width=0pt, white,fill=white]  (2.7,-0.2) rectangle ++(.6,.4);
                        \draw[fill=teal]  (4.8,-0.2) rectangle ++(.4,.4);
                            \draw[black,fill=white] (-1.3,-.2) rectangle ++ (0.8,0.8); 
                             \draw[black,fill=white] (6.6,-.3) rectangle ++ (0.8,0.8);      
                             \draw[black,fill=yellow] (-2,-.2) rectangle ++ (0.4,0.4); 
                                \draw[black,fill=yellow] (7.6,-.2) rectangle ++ (0.4,0.4); 
        \draw[black,fill=red] (2,0) circle (0.2);
                   \draw[black,fill=red] (4,0) circle (0.2); 
                        \draw[black,fill=red] (6,0) circle (0.2);
        \node[ line width=0.6pt, dashed, draw opacity=0.5] (a) at (0,0){\small {\color{white}$\lambda$}};
            \node[ line width=0.6pt, dashed, draw opacity=0.5] (a) at (2,0){\small {\color{white}$\lambda$}};
                \node[ line width=0.6pt, dashed, draw opacity=0.5] (a) at (4,0){\small {\color{white}$\lambda$}};
                    \node[ line width=0.6pt, dashed, draw opacity=0.5] (a) at (6,0){\small {\color{white}$\lambda$}};
        \node[ line width=0.6pt, dashed, draw opacity=0.5] (a) at (-0.8,0.13){\small {\color{black} $\tilde{\XL}_{g^{\cone}}$ } };
        \node[ line width=0.6pt, dashed, draw opacity=0.5] (a) at (7.1,0.12){\small {\color{black} $\tilde{\XR}_{g^{\ctwo}}$ } };
           \node[ line width=0.6pt, dashed, draw opacity=0.5] (a) at (-1.8,0){\small {\color{black} $x$ } };
        \node[ line width=0.6pt, dashed, draw opacity=0.5] (a) at (7.9,0){\small {\color{black} $y$ } };
        \node[ line width=0.6pt, dashed, draw opacity=0.5] (a) at (1.05,0){\small {\color{white} $\mu$ } };
        \node[ line width=0.6pt, dashed, draw opacity=0.5] (a) at (3.05,0){\small {\color{black} $\cdots$ } };
        \node[ line width=0.6pt, dashed, draw opacity=0.5] (a) at (3.05,0.8){\small {\color{black} $\cdots$ } };
        \node[ line width=0.6pt, dashed, draw opacity=0.5] (a) at (5.05,0){\small {\color{white} $\mu$ } };
     %    \node[ line width=0.6pt, dashed, draw opacity=0.5] (a) at (0.15,0.8){\small {\color{black} $\XL_{g^{\cone}}$ } };
      %  \node[ line width=0.6pt, dashed, draw opacity=0.5] (a) at (2.15,0.8){\small {\color{black} $\XL_{g^{\ctwo}}$ } };
       % \node[ line width=0.6pt, dashed, draw opacity=0.5] (a) at (4.15,0.8){\small {\color{black} $\XL_{g^{(n-1)}}$ } };
        %\node[ line width=0.6pt, dashed, draw opacity=0.5] (a) at (6.15,0.8){\small {\color{black} $\XL_{g^{(n)}}$ } };
        \end{tikzpicture}
    \end{aligned}
\end{align}
where $x,y\in \cA$ are boundary elements that we choose to take the inner product with dangling legs at two ends of the chain.

We can also introduce a non-invertible global $\Rep(\cA)^{\rm rev}$-symmetry that acts on all even vertices:
\begin{equation}\label{eq:Zstring}
    D_{\Gamma}=\Tr' [\tilde{Z}_{\Gamma}(2)\otimes \tilde{Z}_{\Gamma}(4)\otimes \cdots \otimes \tilde{Z}_{\Gamma}(2L)].
\end{equation}
They satisfy a fusion rule
\begin{equation}
    D_{\Gamma}D_{\Phi}=D_{\Phi\otimes \Gamma},
\end{equation}
which means that they form a $\Rep(\cA)^{\rm rev}$ symmetry.
Note that this differs from the definition in Ref.~\cite{fechisin2023noninvertible} due to the cocommutative nature of the group algebra $\Delta(g) = g \otimes g$ for $g \in G$. Thus, they use $Z_{\Gamma}$ to construct the global symmetry. However, for a general Hopf algebra, this is not the case, since
\begin{equation}
    Z_{\Gamma}Z_{\Phi}|x\rangle =\sum_{(x)} x^{\cone} \otimes \Gamma(x^{\ctwo}) \otimes \Phi(x^{\cthree})\neq \sum_{(x)} x^{\cone} \otimes \Phi(x^{\ctwo}) \otimes \Gamma(x^{\cthree})=  Z_{\Phi}Z_{\Gamma}|x\rangle.
\end{equation}
The operator $\Tr' [Z_{\Gamma}(2) \otimes Z_{\Gamma}(4) \otimes \cdots \otimes Z_{\Gamma}(2L)]$ does not commute with $\Bf$, and therefore is no longer a symmetry operator.

\begin{proposition}\label{prop:Rep-symm}
      The Hopf cluster state Hamiltonian $ H_{\rm cluster}$ in Eq.~\eqref{eq:Hamiltonian1d} commutes with $D_{\Gamma}$ for all $\Gamma\in \Rep(\cA)^{\rm rev}$, thus it has a $\Rep(\cA)^{\rm rev}$-symmetry.
\end{proposition}

\begin{proof}
    The commutativity relation $[\Bf^{\Phi}_{2k},D_{\Gamma}]$ is clear from the definition. For odd vertex of $\Bf^{\Phi}=Z_{\Phi}^{\ddagger}(2k-1)\otimes Z_{\Phi}(2k)\otimes Z_{\Phi}(2k+1)$, the corresponding operator is $Z_{\Phi}(2k)$. We have $[Z_{\Phi}(2k),\tilde{Z}_{\Gamma}(2k)]=0$, which implies $[\Bf^{\Phi}(2k),D_{\Gamma}]=0$ and further $[\Bf(2k),D_{\Gamma}]=0$ .
    
    Consider an odd vertex, $\Av_{2k+1} = \XL_{\lambda^{\cone}} \otimes \XR_{\lambda^{\cthree}} \otimes \XR_{\lambda^{\ctwo}}$. For the product state $|\cdots, x_{2k}, x_{2k+1}, x_{2k+2}, \cdots \rangle$, we have
    \begin{equation}
    \begin{aligned}
        & D^{\Gamma}     \Av_{2k+1}  |\cdots,x_{2k},x_{2k+1},x_{2k+2},\cdots \rangle \\
         =&\Tr'\Gamma(\cdots x_{2k}^{\cone}S(\lambda^{\ctwo}) \lambda^{\cthree}x_{2k}^{\cone}\cdots)  |\cdots,x_{2k}^{\ctwo}S(\lambda^{\cone}),\lambda^{(5)}x_{2k+1},\lambda^{(4)}x_{2k+2}^{\ctwo},\cdots \rangle \\
         =&\varepsilon(\lambda^{(2)})\Tr'\Gamma(\cdots x_{2k}^{\cone}x_{2k}^{\cone}\cdots)  |\cdots,x_{2k}^{\ctwo}S(\lambda^{\cone}),\lambda^{(4)}x_{2k+1},\lambda^{(3)}x_{2k+2}^{\ctwo},\cdots \rangle\\
         =&\Tr'\Gamma(\cdots x_{2k}^{\cone}x_{2k}^{\cone}\cdots)  |\cdots,x_{2k}^{\ctwo}S(\lambda^{\cone}),\lambda^{(3)}x_{2k+1},\lambda^{(2)}x_{2k+2}^{\ctwo},\cdots \rangle\\
         =& \Av_{2k+1}  D^{\Gamma}      |\cdots,x_{2k},x_{2k+1},x_{2k+2},\cdots \rangle 
    \end{aligned}.
    \end{equation}
This completes the proof.
\end{proof}

Similar to $\cA$-symmetry, the bulk symmetry operators can also be pushed to boundaries. A Hopf Z-string operator (Eq.~\eqref{eq:Zstring}) can be effectively regarded as two boundary Z-operators $\Tr' \tilde{Z}^{\ddagger}_{\Gamma} \otimes Z_{\Gamma}^{\ddagger}$.
This can be derived based on Eqs.~\eqref{eq:TNZ3}-\eqref{eq:TNZ2}, if we take open boundary condition for Z-string operator($\alpha$ and $\beta$ as left and right indices), represented as Hopf tensor network, we have 
\begin{align}
    \begin{aligned}
        \begin{tikzpicture}  
        \draw[line width=.6pt,black] (-0.8,0) -- (6.8,0);
        \draw[line width=.6pt,black] (0,0) -- (0.8,0);
        \draw[line width=.6pt,black] (0,0) -- (-0.8,0);
           \draw[line width=.6pt,blue] (0,0) -- (0,1.5);
        \draw[line width=.6pt,blue] (6,0) -- (6,1.5);
         \draw[line width=.6pt,blue] (1,0) -- (1,1.5);
        \draw[line width=.6pt,blue] (2,0) -- (2,1.5);
        \draw[line width=.6pt,cyan] (0.4,0.8) -- (1.96,.8);
        \draw[line width=.6pt,cyan] (2.7,0.8) -- (2.05,.8);
        \draw[line width=.6pt,cyan] (3.3,0.8) -- (3.95,.8);
        \draw[line width=.6pt,cyan] (5.6,0.8) -- (4.05,.8);
     %  \draw[line width=.6pt,blue] (3,0) -- (3,1.2);
           \draw[cyan,fill=cyan] (0.4,0.8) circle (0.1);
                   \node[ line width=0.6pt, dashed, draw opacity=0.5] (a) at (.43,0.8){\tiny {\color{white} $\alpha$ } };
                            \draw[cyan,fill=cyan] (5.6,0.8) circle (0.1);
                   \node[ line width=0.6pt, dashed, draw opacity=0.5] (a) at (5.63,0.8){\tiny {\color{white} $\beta$ } };
       \draw[line width=.6pt,blue] (4,0) -- (4,1.5);
        \draw[line width=.6pt,blue] (5,0) -- (5,1.5);
         \draw[black,fill=red] (0,0) circle (0.2); 
         \draw[black,fill=white] (.7,0.5) rectangle ++ (0.6,0.6); 
             \draw[black,fill=white] (4.7,0.5) rectangle ++ (0.6,0.6); 
                 %\draw[black,fill=white] (3.5,0.5) rectangle ++ (1.2,0.6); 
                     %\draw[black,fill=white] (5.7,0.5) rectangle ++ (0.8,0.6); 
        \draw[fill=teal]  (0.8,-0.2) rectangle ++(.4,.4);
                \draw[line width=0pt, white,fill=white]  (2.7,-0.2) rectangle ++(.6,.4);
                        \draw[fill=teal]  (4.8,-0.2) rectangle ++(.4,.4);
                             \draw[black,fill=yellow] (-1,-.2) rectangle ++ (0.4,0.4); 
                                \draw[black,fill=yellow] (6.6,-.2) rectangle ++ (0.4,0.4); 
        \draw[black,fill=red] (2,0) circle (0.2);
                   \draw[black,fill=red] (4,0) circle (0.2); 
                        \draw[black,fill=red] (6,0) circle (0.2);
        \node[ line width=0.6pt, dashed, draw opacity=0.5] (a) at (0,0){\small {\color{white}$\lambda$}};
            \node[ line width=0.6pt, dashed, draw opacity=0.5] (a) at (2,0){\small {\color{white}$\lambda$}};
                \node[ line width=0.6pt, dashed, draw opacity=0.5] (a) at (4,0){\small {\color{white}$\lambda$}};
                    \node[ line width=0.6pt, dashed, draw opacity=0.5] (a) at (6,0){\small {\color{white}$\lambda$}};
        \node[ line width=0.6pt, dashed, draw opacity=0.5] (a) at (-0.7,0){\small {\color{black} $x$ } };
        \node[ line width=0.6pt, dashed, draw opacity=0.5] (a) at (6.9,0){\small {\color{black} $y$ } };
        \node[ line width=0.6pt, dashed, draw opacity=0.5] (a) at (1.05,0){\small {\color{white} $\mu$ } };
        \node[ line width=0.6pt, dashed, draw opacity=0.5] (a) at (3.05,0){\small {\color{black} $\cdots$ } };
        \node[ line width=0.6pt, dashed, draw opacity=0.5] (a) at (3.05,0.8){\small {\color{cyan} $\cdots$ } };
        \node[ line width=0.6pt, dashed, draw opacity=0.5] (a) at (5.05,0){\small {\color{white} $\mu$ } };
         \node[ line width=0.6pt, dashed, draw opacity=0.5] (a) at (1,0.8){\small {\color{black} $\tilde{Z}_{\Gamma}$} };
          \node[ line width=0.6pt, dashed, draw opacity=0.5] (a) at (5,0.8){\small {\color{black} $\tilde{Z}_{\Gamma}$} };
        \end{tikzpicture}
    \end{aligned}\nonumber\\
      =  \begin{aligned}
        \begin{tikzpicture}  
        \draw[line width=.6pt,black] (-2,0) -- (7.8,0);
        \draw[line width=.6pt,black] (0,0) -- (0.8,0);
        \draw[line width=.6pt,black] (0,0) -- (-0.8,0);
           \draw[line width=.6pt,blue] (0,0) -- (0,1.);
        \draw[line width=.6pt,blue] (6,0) -- (6,1.);
         \draw[line width=.6pt,blue] (1,0) -- (1,1.);
        \draw[line width=.6pt,blue] (2,0) -- (2,1.);
     %  \draw[line width=.6pt,blue] (3,0) -- (3,1.2);
       \draw[line width=.6pt,blue] (4,0) -- (4,1.);
        \draw[line width=.6pt,blue] (5,0) -- (5,1.);
        \draw[line width=.6pt,cyan] (-1.7,0.4) -- (-.05,.4);
        \draw[line width=.6pt,cyan] (.05,0.4) -- (0.95,.4);
         \draw[line width=.6pt,cyan] (1.05,0.4) -- (1.95,.4);
                 \draw[line width=.6pt,cyan] (2.05,0.4) -- (2.7,.4);
                      \draw[line width=.6pt,cyan] (3.3,0.4) -- (3.95,.4);
                          \draw[line width=.6pt,cyan] (4.05,0.4) -- (4.95,.4);
                              \draw[line width=.6pt,cyan] (5.05,0.4) -- (5.95,.4);
                                  \draw[line width=.6pt,cyan] (6.05,0.4) -- (6.95,.4);
                                      \draw[line width=.6pt,cyan] (7.05,0.4) -- (7.6,.4);
         \draw[black,fill=red] (0,0) circle (0.2); 
        % \draw[black,fill=white] (-0.3,0.5) rectangle ++ (0.8,0.6); 
         %    \draw[black,fill=white] (1.7,0.5) rectangle ++ (0.8,0.6); 
          %       \draw[black,fill=white] (3.5,0.5) rectangle ++ (1.2,0.6); 
           %          \draw[black,fill=white] (5.7,0.5) rectangle ++ (0.8,0.6); 
        \draw[fill=teal]  (0.8,-0.2) rectangle ++(.4,.4);
                \draw[line width=0pt, white,fill=white]  (2.7,-0.2) rectangle ++(.6,.4);
                        \draw[fill=teal]  (4.8,-0.2) rectangle ++(.4,.4);
                            \draw[black,fill=white] (-1.3,-.2) rectangle ++ (0.6,0.8); 
                             \draw[black,fill=white] (6.6,-.2) rectangle ++ (0.6,0.8);      
                             \draw[black,fill=yellow] (-2,-.2) rectangle ++ (0.4,0.4); 
                                \draw[black,fill=yellow] (7.6,-.2) rectangle ++ (0.4,0.4);
               \draw[cyan,fill=cyan] (-1.7,0.4) circle (0.1);
                   \node[ line width=0.6pt, dashed, draw opacity=0.5] (a) at (-1.68,0.4){\tiny {\color{white} $\alpha$ } };
                            \draw[cyan,fill=cyan] (7.6,0.4) circle (0.1);
                   \node[ line width=0.6pt, dashed, draw opacity=0.5] (a) at (7.62,0.4){\tiny {\color{white} $\beta$ } };                   
        \draw[black,fill=red] (2,0) circle (0.2);
                   \draw[black,fill=red] (4,0) circle (0.2); 
                        \draw[black,fill=red] (6,0) circle (0.2);
        \node[ line width=0.6pt, dashed, draw opacity=0.5] (a) at (0,0){\small {\color{white}$\lambda$}};
            \node[ line width=0.6pt, dashed, draw opacity=0.5] (a) at (2,0){\small {\color{white}$\lambda$}};
                \node[ line width=0.6pt, dashed, draw opacity=0.5] (a) at (4,0){\small {\color{white}$\lambda$}};
                    \node[ line width=0.6pt, dashed, draw opacity=0.5] (a) at (6,0){\small {\color{white}$\lambda$}};
        \node[ line width=0.6pt, dashed, draw opacity=0.5] (a) at (-0.9,0.1){\small {\color{black} ${Z}^{\ddagger}_{\Gamma}$ } };
        \node[ line width=0.6pt, dashed, draw opacity=0.5] (a) at (6.95,0.1){\small {\color{black} ${Z}_{\Gamma}^{\ddagger}$ } };
           \node[ line width=0.6pt, dashed, draw opacity=0.5] (a) at (-1.8,0){\small {\color{black} $x$ } };
        \node[ line width=0.6pt, dashed, draw opacity=0.5] (a) at (7.9,0){\small {\color{black} $y$ } };
        \node[ line width=0.6pt, dashed, draw opacity=0.5] (a) at (1.05,0){\small {\color{white} $\mu$ } };
        \node[ line width=0.6pt, dashed, draw opacity=0.5] (a) at (3.05,0){\small {\color{black} $\cdots$ } };
        \node[ line width=0.6pt, dashed, draw opacity=0.5] (a) at (3.05,0.4){\small {\color{cyan} $\cdots$ } };
        \node[ line width=0.6pt, dashed, draw opacity=0.5] (a) at (5.05,0){\small {\color{white} $\mu$ } };
     %    \node[ line width=0.6pt, dashed, draw opacity=0.5] (a) at (0.15,0.8){\small {\color{black} $\XL_{g^{\cone}}$ } };
      %  \node[ line width=0.6pt, dashed, draw opacity=0.5] (a) at (2.15,0.8){\small {\color{black} $\XL_{g^{\ctwo}}$ } };
       % \node[ line width=0.6pt, dashed, draw opacity=0.5] (a) at (4.15,0.8){\small {\color{black} $\XL_{g^{(n-1)}}$ } };
        %\node[ line width=0.6pt, dashed, draw opacity=0.5] (a) at (6.15,0.8){\small {\color{black} $\XL_{g^{(n)}}$ } };
        \end{tikzpicture}
    \end{aligned}
\end{align}
where we use cyan dot to represent the free indices in representation space.

Using the embedding of $\Rep(\cA)^{\mathrm{rev}}$ into $\bar{\cA}^{\mathrm{op}}$, we see that the $Z$-string operator $D_{\Gamma}$ in Eq.~\eqref{eq:Zstring} can be regarded as taking the comultiplication $\Delta_{L}(\chi_{\Gamma})$ (where $L$ is the number of even vertices) of $\chi_{\Gamma}$ and assigning the component $\chi_{\Gamma}^{(i)}$ to the $i$-th even vertex. This implies that the categorical symmetry $\Rep(\cA)^{\mathrm{rev}}$ is a special case of Hopf symmetry, and this is a general phenomenon when we only consider the fusion structure of the category.

\section{Correspondence between 1d Hopf cluster state model and quasi-1d Hopf quantum double model}
\label{sec:QD}

As mentioned before, the 1d cluster state model is intricately connected to Kitaev's quantum double model. In this section, we establish this connection by mapping a 1d cluster state lattice Hamiltonian into a quasi-1d quantum double model for Hopf qudits. This correspondence can assist us in providing a weak Hopf generalization of the 1d cluster state model based on our previous work \cite{Jia2023weak}.

\vspace{1em}
\emph{Qubit case.} --- 
It's worthwhile to start by discussing the qubit case, where the quantum double model is commonly known as the toric code model \cite{Kitaev2003}.
For quantum double lattice, the vertex operators and face operators are of the forms:
\begin{equation}
    A_v=\prod_{j\in \partial v}X_j,\quad B_f=\prod_{j\in \partial f} Z_j,
\end{equation}
where $v,f$ are vertices and faces of the quantum double lattice.
There are two kinds of formalisms for the quantum double lattice. In edge-lattice one, as described in Ref.~\cite{Kitaev2003}, the qudit is placed on the edges, allowing discussion of vertices and faces. The other formalism is the chessboard representation, commonly used in the quantum error-correction community. 
As shown in Fig.~\ref{fig:QDcluster}, the edge-lattice one draw as a gray dotted lattice, where the vertices are represented as star symbols.
The chessboard representation is given by light green and light orange triangles.

The lattice Hamiltonian reads 
\begin{equation}
    H=-\sum_v A_v-\sum_f B_f.
\end{equation}
Consider a quasi-1d quantum double lattice, composed of vertex and face plaquettes arranged in a sequence, represented using a chessboard layout as shown in Fig.~\ref{fig:QDRibbon}.
In the case of periodic boundary conditions (Fig.~\ref{fig:QDRibbon}), when viewing the lattice as a 2d disk, it becomes apparent that the vertex plaquette operator surrounded by all face plaquettes can be expressed as a product of all vertex plaquette operators. Consequently, we can eliminate this particular vertex plaquette from the lattice. More precisely, consider a bicycle-wheel-like quantum double lattice on a disk. The boundary vertex operator is given by \( A_v = X_{e_1}X_{e_2}X_{s_3} \), where the \( X_{e_1} \) and \( X_{e_2} \) operators act on the "tire" edges, and the \( X_{s_3} \) operator acts on the "spoke" edge. In contrast, the bulk vertex operator is \( A_u = X_{s_1}X_{s_2} \cdots X_{s_n} \), where all \( X_{s_i} \) operators act on the "spoke" edges. We observe that \( \prod_{v: \text{tire}} A_v = A_u \). Here, for each "tire" edge, the \( X \) operator acts twice, effectively becoming the identity. This implies that there is no need to impose an additional \( A_u = 1 \) condition to obtain the ground state. Furthermore, note that removing the bulk vertex operator introduces a rough boundary at the center of the disk \cite{jia2022electricmagnetic}.

\begin{figure}[t]
    \centering
    \includegraphics[width=10cm]{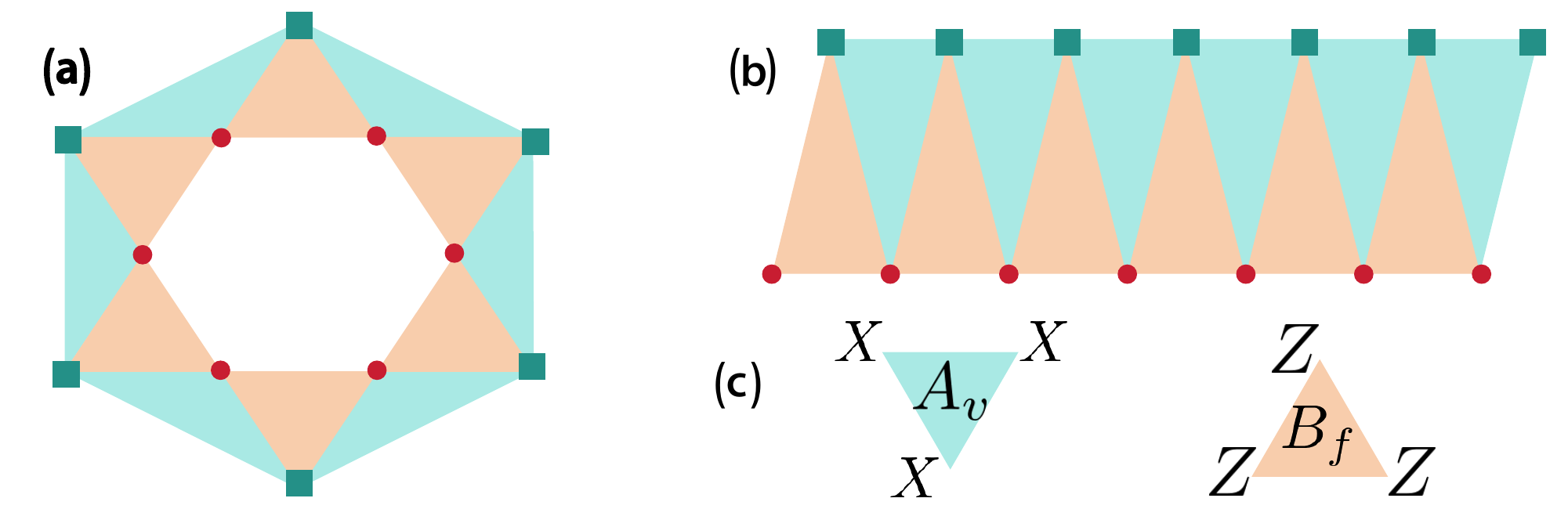}
    \caption{Illustration depicting the quasi-1d toric code model represented in a chess-board layout: (a) periodic boundary conditions and (b) open boundary conditions. The vertex plaquettes are depicted as light green triangles, while the face plaquettes are represented as light orange triangles. The red circles represent odd vertices, while the teal squares denote even vertices in the original Hopf cluster lattice.}
    \label{fig:QDRibbon}
\end{figure}

Via folding the 1d cluster state model into a zigzag form (Fig.~\ref{fig:QDcluster}), we see that the cluster state model in Eq.~\eqref{eq:ClusterHamiltonianCSS} is completely the same with the quasi-1d toric code model. To summarize, we have

\begin{proposition}
    The 1d  qubit cluster state model is equivalent to a quasi-1d toric code model.
\end{proposition}

It's noteworthy to mention that the cluster state Hamiltonian in Eq.~\eqref{eq:ClusterHamiltonianNormal} is equivalent to the quasi-1d toric code model expressed in Wen's formulation \cite{wen2003quantum}.
Notice that the global symmetries of the cluster state model can be expressed 
\begin{equation}
    F=\otimes_{i:\rm odd} X_i,\quad D=\otimes_{j:\rm even} Z_j.
\end{equation}
We see that in the quantum double model, this is just the closed string operator, $F$ is a closed magnetic string and $D$ is a closed electric string.

\vspace{1em}
\emph{Hopf qudit case.} --- 
In the case of finite groups and more general Hopf algebras, establishing this equivalence is not straightforward, as the vertex and face operators appear to differ significantly from those of the quantum double model. However, as we will observe, they are still equivalent, akin to the qubit case.

To proceed with establishing the equivalence between the 1d Hopf cluster state model and the quasi-1d Hopf quantum double model, let's first revisit the definition of local stabilizers in the Hopf quantum double model.

The vertex operators are constructed using left regular actions: $L_+^h=\XR_h$ and $L_-^h=\XL_h$. We have adopted the notation $L_{\pm}$ following Kitaev's original work \cite{Kitaev2003} to enhance clarity in the discussion.
The edge operator are constructed from \cite{Buerschaper2013a,chen2021ribbon,jia2023boundary,Jia2023weak}:
\begin{align}
		T_+^{\varphi}|x\rangle=&|\varphi \rightharpoonup x\rangle=| \sum_{(x)}\langle \varphi, x^{(2)}\rangle x^{(1)}\rangle,\\
		T_-^{\varphi}|x\rangle
        =&|x\leftharpoonup \bar{S}(\varphi) \rangle 
		=| \sum_{(x)}\langle \bar{S}(\varphi), x^{(1)}\rangle x^{(2)}\rangle 
		=| \sum_{(x)}\langle \varphi, S(x^{(1)})\rangle x^{(2)}\rangle,
	\end{align}
where $\varphi\in\bar{\cA}$ is linear functional over $\cA$ and we have used Sweedler's arrow notation. Refer to our previous work \cite{jia2023boundary,Jia2023weak} for details about the notation we have used.

Let $j$ be a directed edge with endpoint $v$. We define $L^h(j,v)$ as follows: if $v$ is the origin of $j$, then $L^h(j,v)=L_+^h(j)$; otherwise, $L^h(j,v)=L_{-}^h(j)$. 
For edge $j$ surrounding face $f$, if $f$ is on the left hand side, we choose $T^{\varphi}(j,f)=T_{-}^{\varphi}(j)$, otherwise, we choose $T^{\varphi}(j,f)=T_+^{\varphi}(j)$.
For our purposes, for a given site $s=(v,f)$, we order the edges around the vertex  $v$ and the face $f$ clockwise with the origin $s$, which is different from that in Refs.~\cite{jia2023boundary,Jia2023weak}.
This is also the reason that we choose to use the above convention.
The vertex operators and face operators on a site are defined as
\begin{equation}   
\begin{aligned}
        \begin{tikzpicture}   
           \draw[line width=.6pt,-latex,black] (0,0)--(0,1);
         \draw[line width=.6pt,black] (0,0.8)--(0,1.8);
        \draw[black,fill=black] (0,1.8) circle (0.08); 
           \draw[black,fill=black] (0,0) circle (0.08);     
           \node[ line width=0.6pt, dashed, draw opacity=0.5] (a) at (0,-0.4){$L_+$};
           \node[line width=0.6pt, dashed, draw opacity=0.5] (a) at (0,2.1){$L_{-}$};
           \node[line width=0.6pt, dashed, draw opacity=0.5] (a) at (0.4,0.9){$T_+$};
           \node[line width=0.6pt, dashed, draw opacity=0.5] (a) at (-0.4,0.9){$T_-$};
        \end{tikzpicture}
    \end{aligned}\quad 
    \begin{aligned}
		A^{h}(s)= \sum_{(h)} L^{h^{(1)}}(j_1,v)\otimes \cdots \otimes L^{h^{(n)}}(j_n,v),\\
		B^{\varphi}(s)= \sum_{(\varphi)} T^{\varphi^{(1)}}(j_1,f)\otimes \cdots \otimes T^{\varphi^{(n)}}(j_n,f).
	\end{aligned}
 \end{equation} 
For example, consider the following configuration
	\begin{equation}
		\begin{aligned}
			\begin{tikzpicture}
				\draw[-latex,black] (-1,0) -- (0,0); 
				\draw[-latex,black] (0,0) -- (1,0); 
				\draw[-latex,black] (0,0) -- (0,1); 
				\draw[-latex,black] (0,-1) -- (0,0); 
				\draw[-latex,black] (0,1) -- (1,1);
				\draw[-latex,black] (1,0) -- (1,1);  
				\draw[line width=0.5pt, dotted, red] (0,0) -- (0.5,0.5);
				\draw [fill = black] (0,0) circle (1.2pt);
				\draw [fill = black] (0.5,0.5) circle (1.2pt);
                    \draw[red,-latex] (0.2,0.2) arc (45:-250:0.3);
                    \draw[red,-latex] (0.4,0.4) arc (225:-90:0.3);
				\node[ line width=0.2pt, dashed, draw opacity=0.5] (a) at (-0.2,-0.2){$v$};
				\node[ line width=0.2pt, dashed, draw opacity=0.5] (a) at (0.7,0.7){$f$};
				\node[ line width=0.2pt, dashed, draw opacity=0.5] (a) at (-1.2,0){$x_5$};
				\node[ line width=0.2pt, dashed, draw opacity=0.5] (a) at (0,-1.2){$x_6$};
				\node[ line width=0.2pt, dashed, draw opacity=0.5] (a) at (0.8,-0.3){$x_1$};
				\node[ line width=0.2pt, dashed, draw opacity=0.5] (a) at (1.3,0.5){$x_2$};
				\node[ line width=0.2pt, dashed, draw opacity=0.5] (a) at (0.5,1.2){$x_3$};
				\node[ line width=0.2pt, dashed, draw opacity=0.5] (a) at (-0.2,0.5){$x_4$};
			\end{tikzpicture}
		\end{aligned}
	\end{equation}
the corresponding quantum double vertex and face operators are defined as
\begin{equation}
  \begin{aligned}
    &A^h(s)=\sum_{(h)} L_-^{h^{(1)}}(j_1) \otimes L_+^{h^{(2)}}(j_6)\otimes L_+^{h^{(3)}}(j_5) \otimes L_-^{h^{(4)}}(j_4) ,\\
    & B^{\varphi}(s)=\sum_{(\varphi)}T_{-}^{\varphi^{(1)}} (j_4) \otimes T_{-}^{\varphi^{(2)}} (j_3) \otimes T_{+}^{\varphi^{(3)}} (j_2) \otimes T_{+}^{\varphi^{(4)}} (j_1).
    \end{aligned}
\end{equation}
The local stabilizers for the quantum double model are denoted as $A_v=A^{\lambda}_v$ and $B_f=B^{\Lambda}_f$, where $\lambda\in \mathcal{A}$ and $\Lambda\in \bar{\mathcal{A}}$ are Haar integrals. Because of the cocommutativity of Haar integrals, their definition is independent of the initial site, relying solely on the vertex and face, respectively.

The equivalence between the 1d Hopf cluster state model and the quasi-1d quantum double model can be summarized as follows:

\begin{proposition}\label{prop:AB-QD}
By folding the 1D cluster lattice into a zigzag configuration (Fig.~\ref{fig:QDcluster}), the 1d Hopf cluster state Hamiltonian aligns with the quasi-1d Hopf quantum double Hamiltonian, where the odd vertex and even vertex operators correspond to the vertex and face operators of the quantum double model, respectively:
\begin{equation}
    \Av_{2k+1}=A_v^{\lambda},\quad \Bf_{2k}=B_f^{\Lambda}.
\end{equation}
If we look at the quantum double lattice, this means that the Hopf cluster state model is equivalent to a quantum double model with a rough boundary and a smooth boundary defined on a ladder (Fig.~\ref{fig:ladder}).
\end{proposition}

\begin{proof}
By folding the 1D cluster lattice into a zigzag configuration (Fig.~\ref{fig:QDcluster}), we observe that the odd vertex operator $\Av_i$ corresponds to the vertex operator $A_v^{\lambda}$ by definition.
For even vertex operators $\Bf_j$'s, the proof is more complicated.
Notice that, for $\Gamma\in \Irr(\cA)$, when acting $\Bf^{\Gamma}_j$ on $x_{j-1},x_{j},x_{j+1}$, we have
\begin{equation}
    \Bf^{\Gamma}_j|x_{j-1},x_{j},x_{j+1}\rangle =\Tr' [\Gamma(S(x_{j-1}^{\cone})x_j^{\ctwo}x_{j+1}^{\ctwo})] |x_{j-1}^{\ctwo},x_{j}^{\cone},x_{j+1}^{\cone}\rangle.
\end{equation}
Using the character $\chi_{\Gamma}\in \bar{\cA}$, we see 
\begin{equation}
    \Bf^{\Gamma}_j|x_{j-1},x_{j},x_{j+1}\rangle  =\sum_{\chi_{\Gamma}}{\chi_{\Gamma}^{\cone}}(S(x_{j-1}^{\cone})) \chi_{\Gamma}^{\ctwo}(x_j^{\ctwo}) \chi_{\Gamma}^{\cthree}(x_{j+1}^{\ctwo}) 
    |x_{j-1}^{\ctwo},x_{j}^{\cone},x_{j+1}^{\cone}\rangle.
\end{equation}
This means that 
\begin{equation}
     \Bf^{\Gamma}_j=B^{\chi_{\Gamma}}.
\end{equation}
Therefore
\begin{equation}
     \Bf_j=\sum_{\Gamma\in \Irr(\cA)} \frac{d_{\Gamma}}{|\cA|} B^{\chi_{\Gamma}}.
\end{equation}
From Eq.~\eqref{eq:HaarLambda1}, we obtain
\begin{equation}
    \Bf^{\Gamma}_j=B_f^{\Lambda}.
\end{equation}
This completes the proof.
\end{proof}

We have demonstrated that the 1d cluster state model is equivalent to the quantum double model with a smooth boundary, where the internal vertex operator for the 2d bulk is omitted.
In this case, the boundary data is chosen to be identical to that of the bulk \cite{jia2022electricmagnetic,jia2023boundary}. 
The global symmetries of the Hopf cluster state model are also equivalent to closed ribbon operators of the quasi-1d Hopf quantum double model. 
Since we have removed the internal vertex operator for a bicycle-wheel lattice (take periodic boundary condition for the gray lattice in Fig.~\ref{fig:QDcluster}) and focus solely on the face operators and vertex operators on the boundary, this implies that we must first perform some coarse-graining (entanglement renormalization \cite{buerschaper2013electric}) on the bulk to map the bulk lattice into a bicycle-wheel lattice and then project the bulk to obtain the cluster state model \cite{albert2021spin}.

There is another way to interpret the correspondence between the cluster state model and quantum double model. As shown in Fig.~\ref{fig:QDRibbon}, we can regard the cluster state model as a quantum double model on a very thin ribbon whose one boundary (teal) is smooth, while the other boundary (red) is rough \cite{jia2022electricmagnetic,jia2023boundary}. 
Also, note that the result in this section asserts that the cluster state model is a specific choice of boundary conditions for the quantum double model. 
A natural question arises: What model corresponds to other boundary conditions of the Hopf quantum double model? This question can be addressed by employing SymTFT or topological holography \cite{jia2024SymTFT}.

\section{SymTFT and Hopf ladder model}
\label{sec:SymTFT}

In the previous section, we show that the Hopf cluster state model can be regarded as a quasi-1d Hopf quantum double model with a smooth boundary and a rough boundary. 
In fact, this point of view can be generalized by considering general boundary conditions in the framework of SymTFT. The SymTFT (also called topological holography) provides a general framework for us to understand the fusion category symmetry for both gapped and gapless phases \cite{Kong2020algebraic,SchaferNameki2024ICTP,huang2023topologicalholo,bhardwaj2024lattice,freed2024topSymTFT,gaiotto2021orbifold,bhardwaj2023generalizedcharge,apruzzi2023symmetry,bhardwaj2024gappedphases,Zhang2024anomaly,Ji2020categoricalsym}.
A thorough discussion on the application of SymTFT to Hopf and weak Hopf lattice gauge theory will be presented in our forthcoming work \cite{jia2024SymTFT}. Here, we only summarize some key results.

The SymTFT is a triple $(\mathcal{Z}(\eB_{\rm sym}), \eB_{\rm sym}, \eB_{\rm phys})$ that forms a sandwich, where: (i) $\mathcal{Z}(\eB_{\rm sym})$ is a 2d topological phase; (ii) $\eB_{\rm sym}$ is called the symmetry boundary and characterizes the fusion category symmetry of the (1+1)D phase. The symmetry boundary must be a gapped topological boundary; (iii) The physical boundary $\eB_{\rm phys}$ need not be gapped or topological.
The SymTFT for (1+1)D phase is defined on a sandwich manifold $\Sigma \times [0,1]$ with $\operatorname{dim} \Sigma = 1$.
In this section, we only consider the case that $\Sigma$ is closed, i.e., it is a circle.
We place $\mathcal{Z}(\eB_{\rm sym})$ in the bulk $\Sigma \times (0,1)$, and position the two boundaries $\eB_{\rm sym}$ and $\eB_{\rm phys}$ on $\Sigma \times \{0\}$ and $\Sigma \times \{1\}$, respectively. After compactification over the interval $[0,1]$, we obtain a (1+1)D theory.

To understand the Hopf cluster state model in the SymTFT framework, consider the fusion category symmetry given by a fusion category $\eB_{\rm sym}=\Rep(\cA)$ for some Hopf algebra $\cA$. The corresponding SymTFT consists of the following data:
\begin{itemize}
    \item For the (2+1)D bulk, we put a Hopf lattice gauge theory whose topological excitation is given by the representation category $\Rep(D(\cA))$ of the quantum double $D(\cA)$, notice that $\Rep(D(\cA))$ is braided monoidal equivalent to the Drinfeld center $\mathcal{Z}(\eB_{\rm sym})$ of the multifusion symmetry $\eB_{\rm sym}\simeq \Rep(\cA)$.
    \item The symmetry boundary is a topological boundary condition $\eB_{\rm sym} = \Rep(\cA)$ of the bulk lattice gauge phase $\Rep(D(\cA))$, which corresponds to a smooth boundary for which all magnetic fluxes are condensed~\cite{jia2023boundary}. This topological boundary condition
    can be characterized in at least three different ways: (i) A Lagrangian algebra $L_m\in \Rep(D(\cA))$ which is a superposition of all magnetic fluxes; (ii) A module category $\EuScript{M}=\mathsf{Rep}(\cA)$ over $\Rep(\cA)$; (iii) A  comodule algebra $K=\cA$ over Hopf algebra $\cA$.
    \item The physical boundary condition $\eB_{\rm phys}$ is chosen as a rough boundary. This is also a gapped topological boundary condition, meaning that after compactification, the resulting (1+1)D phase is gapped. The Lagrangian algebra $L_e$ of this boundary is a superposition of all electric charges. The module category is $\eN = \mathsf{Vect}$. The comodule algebra is trivial comodule algebra $\mathbb{C}$.

    \item The local operators are generated by Wilson lines that can end at both boundaries. For our Hopf lattice gauge theory, the Wilson lines are ribbon operators connecting the two boundaries. Since the two boundaries we choose are fully electric-condensate boundary and a fully magnetic-condensate boundary, there is only one Wilson operator that can connect both boundaries, which is the one that carries the trivial topological charge. This guarantees the ground state space is non-degenerate:
    \begin{equation}
      \text{GSD}=\dim \Hom(L_m\otimes L_e,\mathbf{1})=1,
    \end{equation}  
    where $\mathbf{1}$ is the vacuum charge in the bulk.
    The fusion category symmetry acts on local operators via half-braiding, which is realized in quantum double mode via half-braiding ribbon operators \cite{jia2023boundary}.
    Since the only Wilson line carries the vacuum charge, this implies that the fusion category symmetry $\eB_{\rm sym} = \Rep(\cA)$ is fully preserved, with no symmetry breaking.
\end{itemize}

Notice that since both of two boundaries are gapped, we can choose any one of them as the symmetry boundary. We have chosen smooth boundary as symmetry boundary and see the there is a $\Rep(\cA)$ symmetry. We can also choose the rough boundary as the symmetry boundary, in this case the symmetry is give by $\Rep(\bar{\cA})$, the representation category of dual Hopf algebras. Since each object in $\Rep(\bar{\cA})$ can be regarded as an elements in $\cA$, we have an embedding 
\begin{equation}
  \mathsf{Fun}_{\Rep(\cA)} (\mathsf{Vect},\mathsf{Vect}) \simeq   \Rep(\bar{\cA})\hookrightarrow \cA.
\end{equation}
All these symmetry elements are cocommutative due to the fact that character functions are cocommutative.
This matches well with our analysis of the Hopf symmetry of the cluster state model using the explicit lattice construction: on a closed manifold, the symmetry element must be cocommutative.

Now let us consider the general case, where we choose two different general boundaries for the Hopf quantum double model. One boundary is given by the $\cA$-comodule algebra $K$, and the other is given by the $\cA$-comodule algebra $J$.
For the bulk \(\Sigma \times (0,1)\), we assign a Hopf algebra $\cA$. 
On the symmetry boundary \(\Sigma \times \{0\}\), we assign a Hopf  comodule algebra \(K\). Depending of the orientation of the boundary, we can choose left or right comodule algebras as input data \cite{jia2023boundary,Jia2023weak}.
Similarly, on the physical boundary \(\Sigma \times \{1\}\), we assign a Hopf comodule algebra \(J\).
The compactification over the interval \([0,1]\) can be realized by considering an ultra-thin sandwich lattice, which is a ladder lattice denoted by \(\mathbb{M}^1\). See Figure~\ref{fig:ladder} for a depiction. The edges for symmetry boundary are drawn in red, the edges for physical boundary are drawn in blue, and the bulk edge is drawn in black.

\begin{figure}[t]
    \centering
    \includegraphics[width=10cm]{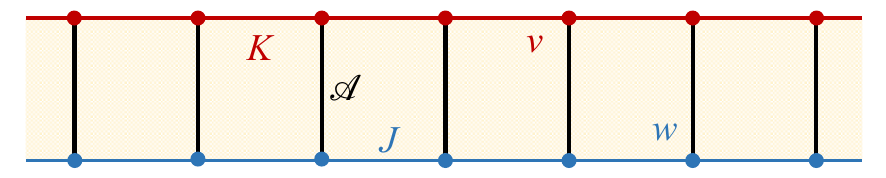}
    \caption{Depiction of the ladder lattice of SymTFT sandwich for Hopf lattice gauge theory.}
    \label{fig:ladder}
\end{figure}

There are three types of local operators: the face operator $\Bf_f$, which is defined using the Haar measure of $\cA$; the red vertex operator $\Av_v$; and the blue vertex operator $\Av_w$, both of which are defined using the symmetric separability idempotent of the comodule algebra. The explicit expressions are given in Refs.~\cite{jia2023boundary, Jia2023weak}. The Hamiltonian is defined as
\begin{equation}
   H_{\rm ladder} = -\sum_f \Bf_f - \sum_{v:~\text{s.b.}} \Av^K_v - \sum_{w:~\text{p.b.}} \Av^J_w.
\end{equation}
This model will be called the Hopf ladder model~\cite{jia2024SymTFT}. The Hopf cluster state model is a special case of this model with $K = \mathbb{C}$ and $J = \cA$, as we have pointed out. Also notice that the Hopf ladder model can be solved exactly using the Hopf tensor network \cite{jia2023boundary,Jia2023weak}.
If we choose the two boundaries to be the same, $J = K$, we obtain the symmetry breaking phase, in which the fusion category symmetry $\eC = {_K}\mathsf{Mod}_K^H$ \cite{jia2023boundary,Jia2023weak} is fully broken.
In this case, the ground state space is degenerate, and the degeneracy is determined by the number of fusion channels between the anyon-condensates of the two boundaries.
To ensure that the model gives an SPT phase protected by $\eC = {_K}\mathsf{Mod}_K^H$, we need to require that there are no nontrivial fusion channels between the two boundary anyon-condensates. The existence of nontrivial fusion channels will partially or fully break the fusion category symmetry~\cite{huang2023topologicalholo, jia2024SymTFT}.

\section{Graph and hypergraph states from the pairing of weak Hopf algebras}
\label{sec:GraphHypergraphHopf}

In the previous sections, our main focus was to construct the Hopf graph state based on a bipartite graph by generalizing the CSS-type qubit graph state (Eq.~\eqref{eq:CSScluster}). In this part, we will consider the general hypergraph, of which a graph is a special case, and discuss how to construct the Hopf hypergraph state based on the decorated hypergraph by generalizing Eq.~\eqref{eq:QubiGraphState}.

A hypergraph $G=(V,E)$ consists of a vertex set $V=\{v_1,\cdots,v_n\}$, and a hyperedge set $E=\{e_1,\cdots,e_m\} \subset 2^V$ (i.e., a $k$-hyperedge is a set $e_i=\{v_{i_1},\cdots,v_{i_k}\}\subset V$). 
A $k$-uniform hypergraph is a hypergraph whose each hyperedge contains precisely $k$ vertices. A graph is $2$-uniform hypergraph.
We can also introduce the direction of the $k$-hyperedge, which is an ordered $k$-tuple $e=(v_{i_1},\cdots,v_{i_k})$. For a $k$-hyperedge, there are $k!$ possible choices of hyperedge directions.
The hyperedges of a multi-hypergraph are allowed to appear repeatedly in the hyperedge set, the number of copies of a hyperedge $e$ is called multiplicity of $e$ and is denoted as $m_e$.
A hypergraph whose all hyperedges are directed is called a directed hypergraph.
A hypergraph is called connected if, for any two vertices, there is a set of hyperedges (a hyper-path) that connect them.
If a hypergraph is not connected, it can be broken into several connected components, therefore, it suffices to investigate connected hypergraph only.

\subsection{Generalized hypergraph state for qudit}

In order to construct the quantum state from a given directed hypergraph, we assign a corresponding Hilbert space $\mathcal{H}_{v_i}$ to each vertex $v_i$ (e.g., for discrete variable case, $\mathcal{H}_{v_i}=\mathbb{C}^d$; for continuous variable case, $\mathcal{H}_{v_i}=L^2(\mathbb{R})$). To each $k$-hyperedge $e_i=\{v_{i_1},\cdots,v_{i_k}\}$,
we assign a $k$-linear functional $\phi_{e_i}: \otimes_j \mathcal{H}_{v_{i_j}} \to \mathbb{C}$. By definition, $\phi_{e_i}$ is determined by the values $\phi_{e_i} (i_1 ,\cdots ,i_k)$ over basis elements $I=(i_1,\cdots,i_k)\in \mathbb{Z}_{d_1}\times \cdots \mathbb{Z}_{d_k}$.
In this way, we obtain a decorated hypergraph $\mathcal{G}=(V,E,\{\mathcal{H}_{v_i}\},\{\phi_{e_i}\})$.
The corresponding state $|\mathcal{G}\rangle$ can be constructed using a quantum circuit that will be discussed later.
Notice that for hyperedge with no direction, $\phi_e$ needs to be symmetric for all variables. More precisely, for a $k$-hyperedge, $\phi_e(i_{\sigma(1)} ,\cdots ,i_{\sigma(k)})=\phi_{e_i} (i_1,\cdots ,i_k)$ holds for arbitrary $k$-permutation $\sigma\in \operatorname{Aut}(I)\simeq \mathfrak{S}_k$.

Consider a $k$-hyperedge $e$, we can introduce the hyperedge gate $U_{e}$ to realize the edge function $\phi_e$ (to ensure the unitarity of $U_e$, here we choose $\phi_e=e^{i\theta_e}$ with $\theta_e$ a real function, a more general case will be discussed in the sequel):
\begin{equation}
\begin{aligned}
    	U_e |I\rangle	 = \phi_e(I)|I\rangle	 = e^{i\theta_e(I)} |I\rangle.
\end{aligned}
\end{equation} 
This $U_e$ is in fact a controlled operation
\begin{align}
&U_e= \sum_{I_{e\setminus \{v_n\}}} |I_{e\setminus \{v_n\}}\rangle \langle I_{e\setminus \{v_n\}}|\otimes U_{I_{e\setminus \{v_n\}}},
\end{align}
where  $ U_{I_{e\setminus \{v_n\}}}$ is a diagonal unitary matrix with diagonal entries equal to  $e^{i\theta_e(i_1,\cdots,i_{n-1},0)}$, $\cdots$, $e^{i\theta_e(i_1,\cdots,i_{n-1},d_n-1)}$.
For the edge with no direction, using the permutation symmetry of the hyperedge function, there will also be a corresponding symmetry of the hyperedge gate, namely, it does not depend on the choice of target vertex on which the controlled-operation acts.
For instance, consider the qubit case, if $e$ is a 2-hyperedge (viz., an edge), 
if $e$ is not directed, then $\theta_e(i,j)=\theta_e(j,i)$,
\begin{align}
    U_e=&|0\rangle \langle 0| \otimes \operatorname{diag}(e^{ia},1)+ |1\rangle \langle 1| \otimes \operatorname{diag}(1,e^{i b}) \nonumber \\
    =&\operatorname{diag}(e^{ia},1) \otimes|0\rangle \langle 0| +\operatorname{diag}(1,e^{i b}) \otimes |1\rangle \langle 1|,
\end{align}
where $a,b$ are two independent real parameters and we have omitted the overall phase factor. If we take $a=0$ and $b=\pi$, $U_e$ becomes the controlled-$Z$ gate.

There are several crucial classes of hyperedge gates, one prototypical class is the controlled power gate.
Choose a single-particle gate $V$ and set $C^0(V)=V$, we define recursively 
\begin{equation}
    C^k(V)=C[C^{k-1}(V)]=\sum_{q}|q\rangle \langle q| \otimes (C^{k-1}(V))^q.
\end{equation}
It's easy to check that
\begin{equation}
    C^k(Z)=\sum_{q_1,\cdots,q_k} |q_1\cdots q_k\rangle \langle q_1\cdots q_k| \otimes V^{q_1\times \cdots \times q_k}.
\end{equation}
Examples of controlled power gates include the following:
\begin{itemize}
    \item For qubit, the controlled-$Z$ gate $C(Z)$ and $k$-controlled-$Z$ gate $C^{k}(Z)$ (with $C^0(Z)=Z$).
    \item For qudit, using the Heisenberg-Weyl operator $Z_d=\sum_{q\in \mathbb{Z}_d} w_d^q |q\rangle \langle q|$ with $w_d$ the $d$-th root of unity, we have $C(Z_d)$ and $C^{k}(Z_d)$ (with $C^0(Z_d)=Z_d$).
    \item For fixed diagonal qudit  unitary gate $P_{\theta}=\operatorname{diag}(1,e^{i\theta},\cdots,e^{ i (d-1)\theta})$, $C^k(P_{\theta})$ does not depend on the choice of target vertex. Setting $\theta=2 \pi/d$, the $k$-controlled-$Z_d$ gate is obtained.
\end{itemize}

Another gate we need is the generalized Hadamard gate \cite{cui2015generalized}
\begin{equation}\label{eq:HaarState}
	H|0\rangle =|h\rangle,
\end{equation}
where $h\in \Cbb^d$ is a image of $|0\rangle$. 
For qubit case, $H$ is just Hadamard gate.
We can now give the definition of the hypergraph state
\begin{definition}
For an arbitrary decorated hypergraph $\mathcal{G}=(V,E,\{\mathcal{H}_{v_i}\},\{\phi_{e_i}\})$, there is a corresponding generalized hypergraph state defined as 
\begin{equation} \label{eq:GraphState}
	|\mathcal{G}\rangle =( \prod_{e\in E} U_e^{m_e}) ( \otimes_{v\in V} H_{v}) |0\rangle^{\otimes |V|},
\end{equation}
where $H_v$'s are generalized Hadamard gates acting on each vertex $v$ and $U_e$'s are hyperedge gates determined by hyperedge function $\phi_e$, and $m_e$ is the hyperedge multiplicity.  
Note that, hyperedge unitary operators $U_e$'s commute with each other in general, thus we need not to specify the order in which these gates are applied.
\end{definition}

The general procedure for finding stabilizers of a given generalized hypergraph state $|\mathcal{G}\rangle$ is as follows: firstly, find the stabilizers of state $|h\rangle^{\otimes|V|}$, this can be done by choosing local operators $K_{v}$ whose $+1$ eigenstate is integral state $|h\rangle$, i.e., $K_v=|h\rangle \langle h|\oplus A$ (Notice that $A$ is required to be a normal operator $[A,A^{\dagger}]=0$ with all eigenvalues unequal to $+1$). Then, using the trick of conjugation operation we can obtain the stabilizers of the hypergraph state  $T_v=(\prod_{e}U_e )K_v (\prod_{e}U_e)^{\dagger}$.

\subsection{Weak Hopf hypergraph state}

The above construction of generalized hypergraph states can be extended to Hopf qudits. It's worth noting that all constructions in this section also apply to weak Hopf algebras.

To use the comultiplication here, we need a local ordering structure around each vertex $v$.
 Suppose there are $n=|N_E(v)|$ hyperedges connected to $v$, we need to introduce and ordering of the set $\{v\}\cup N_E(v)$. For example, consider a vertex $v$ connected to four hyperedges $e_1,\cdots,e_4$, we need to order $\{v,e_1,e_2,e_3,e_4\}$:
\begin{equation}\label{eq:vertexOrder}
\begin{aligned}
        \begin{tikzpicture}
           \draw[line width=.6pt,black] (0,0)--(0,-1);
           \draw[line width=.6pt,black] (0,0)--(0,1);
           \draw[line width=1.6pt,red] (0,0)--(.4,-.4);
           \draw[line width=.6pt,black] (0,0)--(1,0);
           \draw[line width=.6pt,black] (0,0)--(-1,0);
           \draw[black,fill=lightgray] (0,0) circle (0.2);
           \node[ line width=0.6pt, dashed, draw opacity=0.5] (a) at (0,0){$v$};
           \node[ line width=0.6pt, dashed, draw opacity=0.5] (a) at (0.6,-0.6){$2$};
           \node[ line width=0.6pt, dashed, draw opacity=0.5] (a) at (0,1.3){$1$};
           \node[ line width=0.6pt, dashed, draw opacity=0.5] (a) at (0,-1.3){$3$};
         \node[ line width=0.6pt, dashed, draw opacity=0.5] (a) at (-1.3,0){$4$};
          \node[ line width=0.6pt, dashed, draw opacity=0.5] (a) at (1.3,0){$5$};
          \node[ line width=0.6pt, dashed, draw opacity=0.5] (a) at (0.2,0.6){$e_1$};
          \node[ line width=0.6pt, dashed, draw opacity=0.5] (a) at (0.2,-0.6){$e_2$};
         \node[ line width=0.6pt, dashed, draw opacity=0.5] (a) at (0.6,0.2){$e_4$};
         \node[ line width=0.6pt, dashed, draw opacity=0.5] (a) at (-0.6,0.2){$e_3$};
        \end{tikzpicture}
    \end{aligned}.
\end{equation}
where we use the red line to denote $v$, the black line to denote the hyperedges, the ordering in this case is $(e_1,v,e_2,e_3,e_4)$.
If we put $h\in \cA$ on the vertex $v$, we introduce comultiplication $\Delta_{|N_E(v)|+1}(h)$ and put the respective component according to the local ordering to the set $\{v\}\cup N_e(v)$. We denote it as 
\begin{equation}
    h\mapsto \sum_{(h)}h^{(v)}\otimes h^{(e_{v,1})}\otimes \cdots \otimes h^{(e_{v,|N_E(v)|})}. 
\end{equation}
E.g., for the example in Eq.~\eqref{eq:vertexOrder}, $h^{(v)}=h^{(2)}$, $h^{e_2}=h^{(3)}$, and so on.

For the directed hyperedge $e$, we denote the ordered vertices in $e$ as $v_{e,1},\cdots v_{e,m}$. For example, consider a 3-hyperedge
\begin{equation}
  \begin{aligned}
        \begin{tikzpicture}
           \draw[line width=.6pt,black] (0,0)--(0,-1);
           \draw[line width=.6pt,black] (0,0)--(0,1);
           \draw[line width=.6pt,black] (0.2,1)--(1.8,1);
           \draw[line width=1.6pt,red] (0,0)--(.4,-.4);
           \draw[line width=1.6pt,red] (2,0)--(2.4,-.4);
          \draw[line width=1.6pt,red] (1,1)--(1.4,1.4);
           \draw[line width=.6pt,black] (0,0)--(3,0);
           \draw[line width=.6pt,black] (0,0)--(-1,0);
          \draw[line width=.6pt,black] (1,1)--(0,0);
            \draw[line width=.6pt,black] (1,1)--(1,2);
            \draw[line width=.6pt,black] (1,1)--(2,0);
              \draw[line width=.6pt,black] (2,-1)--(2,0);
          \fill[lime]   (-0.2,-0.2) -- (2.2,-0.2) -- (1,1.2) -- cycle;
           \draw[line width=.6pt,black] (1,1)--(1,0.5);
           \draw[line width=.6pt,black] (2,0)--(1.2,0.26);
           \draw[line width=.6pt,black] (0,0)--(0.8,0.26);
           \draw[black,fill=lightgray] (0,0) circle (0.3);     
           \draw[black,fill=lightgray] (2,0) circle (0.3); 
           \draw[black,fill=lightgray] (1,1) circle (0.3);  
           \node[ line width=0.6pt, dashed, draw opacity=0.5] (a) at (0,0){$v_{e,1}$};
           \node[line width=0.6pt, dashed, draw opacity=0.5] (a) at (2,0){$v_{e,2}$};
           \node[line width=0.6pt, dashed, draw opacity=0.5] (a) at (1,1){$v_{e,3}$};
           \node[line width=0.6pt, dashed, draw opacity=0.5] (a) at (1,0.3){$e$};
        \end{tikzpicture}
    \end{aligned},  
\end{equation}
there is an ordering of the vertices inside the edge, the edge function $\psi_e:\cA^{\otimes m}\to \Cbb$ acts on the elements of the corresponding vertices according to this ordering of $N_V(e)$.

\begin{definition}[(Weak) Hopf hypergraph state]
  Consider a directed hypergraph $\mathcal{G}$ with local ordering around each vertex,  for a given (weak) Hopf algebra, we put $h_v\in \cA$  for each vertex $v$ and put a linear function $\psi_e: \cA\otimes \cdots \otimes \cA \to \Cbb$. 
  The corresponding (weak) Hopf hypergraph state is given by
  \begin{equation}
      |\mathcal{G}\rangle= \sum_{(h_v),v\in V} [\prod_{e\in E} \psi_e(h^{e}_{v_{e,1}},\cdots, h^{(e)}_{v_{e,m}})] \bigotimes_{v\in V} |h_v^{(v)} \rangle.
  \end{equation}
\end{definition}

A crucial class of the edge function is defined by choosing 
 $\phi \in \bar{\cA}$, and define
\begin{equation}
   \psi_e(h^{(e)}_{v_{e,1}}, \cdots, h^{(e)}_{v_{e,m}}) = \phi(h^{(e)}_{v_{e,1}} \cdots h^{(e)}_{v_{e,m}}).
\end{equation}
In other words, we first multiply $h^{(e)}_{v_{e,1}}, \cdots, h^{(e)}_{v_{e,m}}$ and then apply $\phi$.
This can be equivalently regarded as taking the comultiplication of $\phi$ and putting the corresponding component $\phi^{(i)}$ to the vertices $v_{e,i}$.
Then using the pairing between $\cA$ and $\bar{\cA}$, we have 
\begin{equation}
    \phi(h^{(e)}_{v_{e,1}} \cdots h^{(e)}_{v_{e,m}})=\sum_{(\phi)}\phi^{(1)}(h^{(e)}_{v_{e,1}})\cdots \phi^{(m)}(h^{(e)}_{v_{e,m}}).
\end{equation}
The  (weak) Hopf quantum double state \cite{Buerschaper2013a,girelli2021semidual,jia2023boundary,Jia2023weak} is a special case of our (weak) Hopf hypergraph state with $\phi$ chosen as Haar integral in $\bar{\cA}$ and the local ordering satisfy some other constraints. On a square lattice, it's a $4$-uniform hypergraph state.

\section{Discussion}
\label{sec: Discussion}
In this paper, we propose a general construction of the Hopf cluster state. We also construct the lattice Hamiltonian for the 1d Hopf cluster state and discuss its non-invertible symmetries. This state is a potential candidate for the 1d SPT phase, differing from both the anyonic chain construction \cite{Feiguin2007interacting} and the construction based on Hopf comodule algebras \cite{inamura2022lattice}.
The discussion of the SPT phase for this model will be given in our future work \cite{jia2024cluster}.
The connection we have established between the Hopf cluster state model and the quasi-1d quantum double model is also of significant interest.

Despite the progress that has been made, there are still some open problems. As mentioned, our construction does not apply to weak Hopf algebras. However, in investigating fusion category symmetry, it's essential to consider weak Hopf algebras. This is because any given fusion category symmetry can be viewed as a representation category of some weak Hopf algebra. Notably, there exist fusion categories that cannot be represented as the representation category of Hopf algebras. We leave the solution to this problem for future study \cite{jia2024SymTFT}, leveraging the correspondence between the 1d cluster state model and the quasi-1d quantum double model. The application of our model in MQBC can be implemented in a similar way as that in Ref.~\cite{fechisin2023noninvertible}. This direction will also be discussed in our future work.

\begin{acknowledgments}
I would like to extend my sincere gratitude to Dagomir Kaszlikowski for his invaluable support. I am deeply thankful to Sheng Tan and Jinsong Wu for their warm hospitality and enriching discussions during my stay at BIMSA, where part of this work was carried out. Furthermore, I am  thankful to Daniel Bulmash, Yuting Hu and Liang Kong for the productive discussion. Additionally, I acknowledge Jos\'{e} Garre Rubio, Sakura Sch\"{a}fer-Nameki, and Apoorv Tiwari for bringing my attention to some pertinent works on non-invertible symmetry and providing valuable comments.
This work is supported by the National Research Foundation in Singapore and A*STAR under its CQT Bridging Grant.

\end{acknowledgments}

\bibliographystyle{apsrev4-1-title}
\bibliography{mybib}

\begin{thebibliography}{83}%
\makeatletter
\providecommand \@ifxundefined [1]{%
 \@ifx{#1\undefined}
}%
\providecommand \@ifnum [1]{%
 \ifnum #1\expandafter \@firstoftwo
 \else \expandafter \@secondoftwo
 \fi
}%
\providecommand \@ifx [1]{%
 \ifx #1\expandafter \@firstoftwo
 \else \expandafter \@secondoftwo
 \fi
}%
\providecommand \natexlab [1]{#1}%
\providecommand \enquote  [1]{``#1''}%
\providecommand \bibnamefont  [1]{#1}%
\providecommand \bibfnamefont [1]{#1}%
\providecommand \citenamefont [1]{#1}%
\providecommand \href@noop [0]{\@secondoftwo}%
\providecommand \href [0]{\begingroup \@sanitize@url \@href}%
\providecommand \@href[1]{\@@startlink{#1}\@@href}%
\providecommand \@@href[1]{\endgroup#1\@@endlink}%
\providecommand \@sanitize@url [0]{\catcode `\\12\catcode `\$12\catcode
  `\&12\catcode `\#12\catcode `\^12\catcode `\_12\catcode `\%12\relax}%
\providecommand \@@startlink[1]{}%
\providecommand \@@endlink[0]{}%
\providecommand \url  [0]{\begingroup\@sanitize@url \@url }%
\providecommand \@url [1]{\endgroup\@href {#1}{\urlprefix }}%
\providecommand \urlprefix  [0]{URL }%
\providecommand \Eprint [0]{\href }%
\providecommand \doibase [0]{http://dx.doi.org/}%
\providecommand \selectlanguage [0]{\@gobble}%
\providecommand \bibinfo  [0]{\@secondoftwo}%
\providecommand \bibfield  [0]{\@secondoftwo}%
\providecommand \translation [1]{[#1]}%
\providecommand \BibitemOpen [0]{}%
\providecommand \bibitemStop [0]{}%
\providecommand \bibitemNoStop [0]{.\EOS\space}%
\providecommand \EOS [0]{\spacefactor3000\relax}%
\providecommand \BibitemShut  [1]{\csname bibitem#1\endcsname}%
\let\auto@bib@innerbib\@empty
%</preamble>
\bibitem [{\citenamefont {Nielsen}\ and\ \citenamefont
  {Chuang}(2010)}]{Nielsen2010}%
  \BibitemOpen
  \bibfield  {author} {\bibinfo {author} {\bibfnamefont {M.~A.}\ \bibnamefont
  {Nielsen}}\ and\ \bibinfo {author} {\bibfnamefont {I.~L.}\ \bibnamefont
  {Chuang}},\ }\href
  {http://www.cambridge.org/cn/academic/subjects/physics/quantum-physics-quantum-information-and-quantum-computation/quantum-computation-and-quantum-information-10th-anniversary-edition?format=PB&isbn=9781107002173#eUc7irgofUw4ZEd5.97}
  {\emph {\bibinfo {title} {Quantum computation and quantum information}}}\
  (\bibinfo  {publisher} {Cambridge university press},\ \bibinfo {year}
  {2010})\BibitemShut {NoStop}%
\bibitem [{\citenamefont {Preskill}(1998)}]{preskill1998}%
  \BibitemOpen
  \bibfield  {author} {\bibinfo {author} {\bibfnamefont {J.}~\bibnamefont
  {Preskill}},\ }\href {http://www.theory.caltech.edu/people/preskill/ph229/}
  {\enquote {\bibinfo {title} {Lecture notes for physics 229: Quantum
  information and computation},}\ } (\bibinfo {year} {1998})\BibitemShut
  {NoStop}%
\bibitem [{\citenamefont {Deutsch}(1989)}]{deutsch1989quantum}%
  \BibitemOpen
  \bibfield  {author} {\bibinfo {author} {\bibfnamefont {D.~E.}\ \bibnamefont
  {Deutsch}},\ }\bibfield  {title} {\enquote {\bibinfo {title} {Quantum
  computational networks},}\ }\href
  {https://royalsocietypublishing.org/doi/10.1098/rspa.1989.0099} {\bibfield
  {journal} {\bibinfo  {journal} {Proceedings of the royal society of London.
  A. mathematical and physical sciences}\ }\textbf {\bibinfo {volume} {425}},\
  \bibinfo {pages} {73} (\bibinfo {year} {1989})}\BibitemShut {NoStop}%
\bibitem [{\citenamefont {Raussendorf}\ and\ \citenamefont
  {Briegel}(2001)}]{Raussendorf2001}%
  \BibitemOpen
  \bibfield  {author} {\bibinfo {author} {\bibfnamefont {R.}~\bibnamefont
  {Raussendorf}}\ and\ \bibinfo {author} {\bibfnamefont {H.~J.}\ \bibnamefont
  {Briegel}},\ }\bibfield  {title} {\enquote {\bibinfo {title} {A one-way
  quantum computer},}\ }\href {\doibase 10.1103/PhysRevLett.86.5188} {\bibfield
   {journal} {\bibinfo  {journal} {Phys. Rev. Lett.}\ }\textbf {\bibinfo
  {volume} {86}},\ \bibinfo {pages} {5188} (\bibinfo {year}
  {2001})}\BibitemShut {NoStop}%
\bibitem [{\citenamefont {Raussendorf}\ \emph {et~al.}(2003)\citenamefont
  {Raussendorf}, \citenamefont {Browne},\ and\ \citenamefont
  {Briegel}}]{Raussendorf2003measurement}%
  \BibitemOpen
  \bibfield  {author} {\bibinfo {author} {\bibfnamefont {R.}~\bibnamefont
  {Raussendorf}}, \bibinfo {author} {\bibfnamefont {D.~E.}\ \bibnamefont
  {Browne}}, \ and\ \bibinfo {author} {\bibfnamefont {H.~J.}\ \bibnamefont
  {Briegel}},\ }\bibfield  {title} {\enquote {\bibinfo {title}
  {Measurement-based quantum computation on cluster states},}\ }\href {\doibase
  10.1103/PhysRevA.68.022312} {\bibfield  {journal} {\bibinfo  {journal} {Phys.
  Rev. A}\ }\textbf {\bibinfo {volume} {68}},\ \bibinfo {pages} {022312}
  (\bibinfo {year} {2003})},\ \Eprint {http://arxiv.org/abs/quant-ph/0301052}
  {arXiv:quant-ph/0301052 [quant-ph]} \BibitemShut {NoStop}%
\bibitem [{\citenamefont {Nielsen}(2006)}]{nielsen2006cluster}%
  \BibitemOpen
  \bibfield  {author} {\bibinfo {author} {\bibfnamefont {M.~A.}\ \bibnamefont
  {Nielsen}},\ }\bibfield  {title} {\enquote {\bibinfo {title} {Cluster-state
  quantum computation},}\ }\href
  {https://www.sciencedirect.com/science/article/abs/pii/S0034487706800145}
  {\bibfield  {journal} {\bibinfo  {journal} {Reports on Mathematical Physics}\
  }\textbf {\bibinfo {volume} {57}},\ \bibinfo {pages} {147} (\bibinfo {year}
  {2006})},\ \Eprint {http://arxiv.org/abs/quant-ph/0504097}
  {arXiv:quant-ph/0504097 [quant-ph]} \BibitemShut {NoStop}%
\bibitem [{\citenamefont {Briegel}\ \emph {et~al.}(2009)\citenamefont
  {Briegel}, \citenamefont {Browne}, \citenamefont {D{\"u}r}, \citenamefont
  {Raussendorf},\ and\ \citenamefont {Van~den Nest}}]{briegel2009measurement}%
  \BibitemOpen
  \bibfield  {author} {\bibinfo {author} {\bibfnamefont {H.~J.}\ \bibnamefont
  {Briegel}}, \bibinfo {author} {\bibfnamefont {D.~E.}\ \bibnamefont {Browne}},
  \bibinfo {author} {\bibfnamefont {W.}~\bibnamefont {D{\"u}r}}, \bibinfo
  {author} {\bibfnamefont {R.}~\bibnamefont {Raussendorf}}, \ and\ \bibinfo
  {author} {\bibfnamefont {M.}~\bibnamefont {Van~den Nest}},\ }\bibfield
  {title} {\enquote {\bibinfo {title} {Measurement-based quantum
  computation},}\ }\href {https://www.nature.com/articles/nphys1157} {\bibfield
   {journal} {\bibinfo  {journal} {Nature Physics}\ }\textbf {\bibinfo {volume}
  {5}},\ \bibinfo {pages} {19} (\bibinfo {year} {2009})},\ \Eprint
  {http://arxiv.org/abs/0910.1116} {arXiv:0910.1116 [quant-ph]} \BibitemShut
  {NoStop}%
\bibitem [{\citenamefont {Hein}\ \emph {et~al.}(2006)\citenamefont {Hein},
  \citenamefont {D{\"u}r}, \citenamefont {Eisert}, \citenamefont {Raussendorf},
  \citenamefont {Nest},\ and\ \citenamefont {Briegel}}]{hein2006entanglement}%
  \BibitemOpen
  \bibfield  {author} {\bibinfo {author} {\bibfnamefont {M.}~\bibnamefont
  {Hein}}, \bibinfo {author} {\bibfnamefont {W.}~\bibnamefont {D{\"u}r}},
  \bibinfo {author} {\bibfnamefont {J.}~\bibnamefont {Eisert}}, \bibinfo
  {author} {\bibfnamefont {R.}~\bibnamefont {Raussendorf}}, \bibinfo {author}
  {\bibfnamefont {M.}~\bibnamefont {Nest}}, \ and\ \bibinfo {author}
  {\bibfnamefont {H.-J.}\ \bibnamefont {Briegel}},\ }\href@noop {} {\enquote
  {\bibinfo {title} {Entanglement in graph states and its applications},}\ }
  (\bibinfo {year} {2006}),\ \Eprint {http://arxiv.org/abs/quant-ph/0602096}
  {arXiv:quant-ph/0602096 [quant-ph]} \BibitemShut {NoStop}%
\bibitem [{\citenamefont {Qu}\ \emph {et~al.}(2013)\citenamefont {Qu},
  \citenamefont {Wang}, \citenamefont {Li},\ and\ \citenamefont
  {Bao}}]{Qu2013encoding}%
  \BibitemOpen
  \bibfield  {author} {\bibinfo {author} {\bibfnamefont {R.}~\bibnamefont
  {Qu}}, \bibinfo {author} {\bibfnamefont {J.}~\bibnamefont {Wang}}, \bibinfo
  {author} {\bibfnamefont {Z.-s.}\ \bibnamefont {Li}}, \ and\ \bibinfo {author}
  {\bibfnamefont {Y.-r.}\ \bibnamefont {Bao}},\ }\bibfield  {title} {\enquote
  {\bibinfo {title} {Encoding hypergraphs into quantum states},}\ }\href
  {\doibase 10.1103/PhysRevA.87.022311} {\bibfield  {journal} {\bibinfo
  {journal} {Phys. Rev. A}\ }\textbf {\bibinfo {volume} {87}},\ \bibinfo
  {pages} {022311} (\bibinfo {year} {2013})},\ \Eprint
  {http://arxiv.org/abs/1211.3911} {arXiv:1211.3911 [quant-ph]} \BibitemShut
  {NoStop}%
\bibitem [{\citenamefont {Rossi}\ \emph {et~al.}(2013)\citenamefont {Rossi},
  \citenamefont {Huber}, \citenamefont {Bru{\ss}},\ and\ \citenamefont
  {Macchiavello}}]{rossi2013quantum}%
  \BibitemOpen
  \bibfield  {author} {\bibinfo {author} {\bibfnamefont {M.}~\bibnamefont
  {Rossi}}, \bibinfo {author} {\bibfnamefont {M.}~\bibnamefont {Huber}},
  \bibinfo {author} {\bibfnamefont {D.}~\bibnamefont {Bru{\ss}}}, \ and\
  \bibinfo {author} {\bibfnamefont {C.}~\bibnamefont {Macchiavello}},\
  }\bibfield  {title} {\enquote {\bibinfo {title} {Quantum hypergraph
  states},}\ }\href
  {https://iopscience.iop.org/article/10.1088/1367-2630/15/11/113022}
  {\bibfield  {journal} {\bibinfo  {journal} {New Journal of Physics}\ }\textbf
  {\bibinfo {volume} {15}},\ \bibinfo {pages} {113022} (\bibinfo {year}
  {2013})},\ \Eprint {http://arxiv.org/abs/1211.5554} {arXiv:1211.5554
  [quant-ph]} \BibitemShut {NoStop}%
\bibitem [{\citenamefont {Steinhoff}\ \emph {et~al.}(2017)\citenamefont
  {Steinhoff}, \citenamefont {Ritz}, \citenamefont {Miklin},\ and\
  \citenamefont {G\"{u}hne}}]{Steinhoff2017qudit}%
  \BibitemOpen
  \bibfield  {author} {\bibinfo {author} {\bibfnamefont {F.~E.~S.}\
  \bibnamefont {Steinhoff}}, \bibinfo {author} {\bibfnamefont {C.}~\bibnamefont
  {Ritz}}, \bibinfo {author} {\bibfnamefont {N.~I.}\ \bibnamefont {Miklin}}, \
  and\ \bibinfo {author} {\bibfnamefont {O.}~\bibnamefont {G\"{u}hne}},\
  }\bibfield  {title} {\enquote {\bibinfo {title} {Qudit hypergraph states},}\
  }\href {\doibase 10.1103/PhysRevA.95.052340} {\bibfield  {journal} {\bibinfo
  {journal} {Phys. Rev. A}\ }\textbf {\bibinfo {volume} {95}},\ \bibinfo
  {pages} {052340} (\bibinfo {year} {2017})},\ \Eprint
  {http://arxiv.org/abs/1612.06418} {arXiv:1612.06418 [quant-ph]} \BibitemShut
  {NoStop}%
\bibitem [{\citenamefont {Xiong}\ \emph {et~al.}(2018)\citenamefont {Xiong},
  \citenamefont {Zhen}, \citenamefont {Cao}, \citenamefont {Chen},\ and\
  \citenamefont {Chen}}]{Xiong2017qudit}%
  \BibitemOpen
  \bibfield  {author} {\bibinfo {author} {\bibfnamefont {F.-L.}\ \bibnamefont
  {Xiong}}, \bibinfo {author} {\bibfnamefont {Y.-Z.}\ \bibnamefont {Zhen}},
  \bibinfo {author} {\bibfnamefont {W.-F.}\ \bibnamefont {Cao}}, \bibinfo
  {author} {\bibfnamefont {K.}~\bibnamefont {Chen}}, \ and\ \bibinfo {author}
  {\bibfnamefont {Z.-B.}\ \bibnamefont {Chen}},\ }\bibfield  {title} {\enquote
  {\bibinfo {title} {Qudit hypergraph states and their properties},}\ }\href
  {\doibase 10.1103/PhysRevA.97.012323} {\bibfield  {journal} {\bibinfo
  {journal} {Phys. Rev. A}\ }\textbf {\bibinfo {volume} {97}},\ \bibinfo
  {pages} {012323} (\bibinfo {year} {2018})},\ \Eprint
  {http://arxiv.org/abs/1701.07733} {arXiv:1701.07733 [quant-ph]} \BibitemShut
  {NoStop}%
\bibitem [{\citenamefont {Cui}\ \emph {et~al.}(2015)\citenamefont {Cui},
  \citenamefont {Yu},\ and\ \citenamefont {Zeng}}]{cui2015generalized}%
  \BibitemOpen
  \bibfield  {author} {\bibinfo {author} {\bibfnamefont {S.~X.}\ \bibnamefont
  {Cui}}, \bibinfo {author} {\bibfnamefont {N.}~\bibnamefont {Yu}}, \ and\
  \bibinfo {author} {\bibfnamefont {B.}~\bibnamefont {Zeng}},\ }\bibfield
  {title} {\enquote {\bibinfo {title} {Generalized graph states based on
  hadamard matrices},}\ }\href
  {https://aip.scitation.org/doi/10.1063/1.4926427} {\bibfield  {journal}
  {\bibinfo  {journal} {Journal of Mathematical Physics}\ }\textbf {\bibinfo
  {volume} {56}},\ \bibinfo {pages} {072201} (\bibinfo {year} {2015})},\
  \Eprint {http://arxiv.org/abs/1502.07195} {arXiv:1502.07195 [quant-ph]}
  \BibitemShut {NoStop}%
\bibitem [{\citenamefont {Brell}(2015)}]{brell2015generalized}%
  \BibitemOpen
  \bibfield  {author} {\bibinfo {author} {\bibfnamefont {C.~G.}\ \bibnamefont
  {Brell}},\ }\bibfield  {title} {\enquote {\bibinfo {title} {Generalized
  cluster states based on finite groups},}\ }\href
  {https://iopscience.iop.org/article/10.1088/1367-2630/17/2/023029} {\bibfield
   {journal} {\bibinfo  {journal} {New Journal of Physics}\ }\textbf {\bibinfo
  {volume} {17}},\ \bibinfo {pages} {023029} (\bibinfo {year} {2015})},\
  \Eprint {http://arxiv.org/abs/1408.6237} {arXiv:1408.6237 [quant-ph]}
  \BibitemShut {NoStop}%
\bibitem [{\citenamefont {Fechisin}\ \emph {et~al.}(2023)\citenamefont
  {Fechisin}, \citenamefont {Tantivasadakarn},\ and\ \citenamefont
  {Albert}}]{fechisin2023noninvertible}%
  \BibitemOpen
  \bibfield  {author} {\bibinfo {author} {\bibfnamefont {C.}~\bibnamefont
  {Fechisin}}, \bibinfo {author} {\bibfnamefont {N.}~\bibnamefont
  {Tantivasadakarn}}, \ and\ \bibinfo {author} {\bibfnamefont {V.~V.}\
  \bibnamefont {Albert}},\ }\href@noop {} {\enquote {\bibinfo {title}
  {Non-invertible symmetry-protected topological order in a group-based cluster
  state},}\ } (\bibinfo {year} {2023}),\ \Eprint
  {http://arxiv.org/abs/2312.09272} {arXiv:2312.09272 [cond-mat.str-el]}
  \BibitemShut {NoStop}%
\bibitem [{\citenamefont {Walschaers}\ \emph {et~al.}(2018)\citenamefont
  {Walschaers}, \citenamefont {Sarkar}, \citenamefont {Parigi},\ and\
  \citenamefont {Treps}}]{Walschaers2018tailoring}%
  \BibitemOpen
  \bibfield  {author} {\bibinfo {author} {\bibfnamefont {M.}~\bibnamefont
  {Walschaers}}, \bibinfo {author} {\bibfnamefont {S.}~\bibnamefont {Sarkar}},
  \bibinfo {author} {\bibfnamefont {V.}~\bibnamefont {Parigi}}, \ and\ \bibinfo
  {author} {\bibfnamefont {N.}~\bibnamefont {Treps}},\ }\bibfield  {title}
  {\enquote {\bibinfo {title} {Tailoring non-gaussian continuous-variable graph
  states},}\ }\href {\doibase 10.1103/PhysRevLett.121.220501} {\bibfield
  {journal} {\bibinfo  {journal} {Phys. Rev. Lett.}\ }\textbf {\bibinfo
  {volume} {121}},\ \bibinfo {pages} {220501} (\bibinfo {year} {2018})},\
  \Eprint {http://arxiv.org/abs/1804.09444} {arXiv:1804.09444 [quant-ph]}
  \BibitemShut {NoStop}%
\bibitem [{\citenamefont {Moore}(2019)}]{Moore2019quantum}%
  \BibitemOpen
  \bibfield  {author} {\bibinfo {author} {\bibfnamefont {D.~W.}\ \bibnamefont
  {Moore}},\ }\bibfield  {title} {\enquote {\bibinfo {title} {Quantum
  hypergraph states in continuous variables},}\ }\href {\doibase
  10.1103/PhysRevA.100.062301} {\bibfield  {journal} {\bibinfo  {journal}
  {Phys. Rev. A}\ }\textbf {\bibinfo {volume} {100}},\ \bibinfo {pages}
  {062301} (\bibinfo {year} {2019})},\ \Eprint
  {http://arxiv.org/abs/1909.03871} {arXiv:1909.03871 [quant-ph]} \BibitemShut
  {NoStop}%
\bibitem [{\citenamefont {Looi}\ \emph {et~al.}(2008)\citenamefont {Looi},
  \citenamefont {Yu}, \citenamefont {Gheorghiu},\ and\ \citenamefont
  {Griffiths}}]{Looi2008QECC}%
  \BibitemOpen
  \bibfield  {author} {\bibinfo {author} {\bibfnamefont {S.~Y.}\ \bibnamefont
  {Looi}}, \bibinfo {author} {\bibfnamefont {L.}~\bibnamefont {Yu}}, \bibinfo
  {author} {\bibfnamefont {V.}~\bibnamefont {Gheorghiu}}, \ and\ \bibinfo
  {author} {\bibfnamefont {R.~B.}\ \bibnamefont {Griffiths}},\ }\bibfield
  {title} {\enquote {\bibinfo {title} {Quantum-error-correcting codes using
  qudit graph states},}\ }\href {\doibase 10.1103/PhysRevA.78.042303}
  {\bibfield  {journal} {\bibinfo  {journal} {Phys. Rev. A}\ }\textbf {\bibinfo
  {volume} {78}},\ \bibinfo {pages} {042303} (\bibinfo {year} {2008})},\
  \Eprint {http://arxiv.org/abs/0712.1979} {arXiv:0712.1979 [quant-ph]}
  \BibitemShut {NoStop}%
\bibitem [{\citenamefont {Markham}\ and\ \citenamefont
  {Sanders}(2008)}]{Markham2008graph}%
  \BibitemOpen
  \bibfield  {author} {\bibinfo {author} {\bibfnamefont {D.}~\bibnamefont
  {Markham}}\ and\ \bibinfo {author} {\bibfnamefont {B.~C.}\ \bibnamefont
  {Sanders}},\ }\bibfield  {title} {\enquote {\bibinfo {title} {Graph states
  for quantum secret sharing},}\ }\href {\doibase 10.1103/PhysRevA.78.042309}
  {\bibfield  {journal} {\bibinfo  {journal} {Phys. Rev. A}\ }\textbf {\bibinfo
  {volume} {78}},\ \bibinfo {pages} {042309} (\bibinfo {year} {2008})},\
  \Eprint {http://arxiv.org/abs/0808.1532} {arXiv:0808.1532 [quant-ph]}
  \BibitemShut {NoStop}%
\bibitem [{\citenamefont {Keet}\ \emph {et~al.}(2010)\citenamefont {Keet},
  \citenamefont {Fortescue}, \citenamefont {Markham},\ and\ \citenamefont
  {Sanders}}]{Keet2010quantum}%
  \BibitemOpen
  \bibfield  {author} {\bibinfo {author} {\bibfnamefont {A.}~\bibnamefont
  {Keet}}, \bibinfo {author} {\bibfnamefont {B.}~\bibnamefont {Fortescue}},
  \bibinfo {author} {\bibfnamefont {D.}~\bibnamefont {Markham}}, \ and\
  \bibinfo {author} {\bibfnamefont {B.~C.}\ \bibnamefont {Sanders}},\
  }\bibfield  {title} {\enquote {\bibinfo {title} {Quantum secret sharing with
  qudit graph states},}\ }\href {\doibase 10.1103/PhysRevA.82.062315}
  {\bibfield  {journal} {\bibinfo  {journal} {Phys. Rev. A}\ }\textbf {\bibinfo
  {volume} {82}},\ \bibinfo {pages} {062315} (\bibinfo {year} {2010})},\
  \Eprint {http://arxiv.org/abs/1004.4619} {arXiv:1004.4619 [quant-ph]}
  \BibitemShut {NoStop}%
\bibitem [{\citenamefont {Son}\ \emph {et~al.}(2012)\citenamefont {Son},
  \citenamefont {Amico},\ and\ \citenamefont {Vedral}}]{son2012topological}%
  \BibitemOpen
  \bibfield  {author} {\bibinfo {author} {\bibfnamefont {W.}~\bibnamefont
  {Son}}, \bibinfo {author} {\bibfnamefont {L.}~\bibnamefont {Amico}}, \ and\
  \bibinfo {author} {\bibfnamefont {V.}~\bibnamefont {Vedral}},\ }\bibfield
  {title} {\enquote {\bibinfo {title} {Topological order in 1d cluster state
  protected by symmetry},}\ }\href@noop {} {\bibfield  {journal} {\bibinfo
  {journal} {Quantum Information Processing}\ }\textbf {\bibinfo {volume}
  {11}},\ \bibinfo {pages} {1961} (\bibinfo {year} {2012})},\ \Eprint
  {http://arxiv.org/abs/1111.7173} {arXiv:1111.7173 [quant-ph]} \BibitemShut
  {NoStop}%
\bibitem [{\citenamefont {Seifnashri}\ and\ \citenamefont
  {Shao}(2024)}]{seifnashri2024cluster}%
  \BibitemOpen
  \bibfield  {author} {\bibinfo {author} {\bibfnamefont {S.}~\bibnamefont
  {Seifnashri}}\ and\ \bibinfo {author} {\bibfnamefont {S.-H.}\ \bibnamefont
  {Shao}},\ }\href@noop {} {\enquote {\bibinfo {title} {Cluster state as a
  non-invertible symmetry protected topological phase},}\ } (\bibinfo {year}
  {2024}),\ \Eprint {http://arxiv.org/abs/2404.01369} {arXiv:2404.01369
  [cond-mat.str-el]} \BibitemShut {NoStop}%
\bibitem [{\citenamefont {Kitaev}(2003)}]{Kitaev2003}%
  \BibitemOpen
  \bibfield  {author} {\bibinfo {author} {\bibfnamefont {A.}~\bibnamefont
  {Kitaev}},\ }\bibfield  {title} {\enquote {\bibinfo {title} {Fault-tolerant
  quantum computation by anyons},}\ }\href {\doibase
  https://doi.org/10.1016/S0003-4916(02)00018-0} {\bibfield  {journal}
  {\bibinfo  {journal} {Annals of Physics}\ }\textbf {\bibinfo {volume}
  {303}},\ \bibinfo {pages} {2 } (\bibinfo {year} {2003})},\ \Eprint
  {http://arxiv.org/abs/quant-ph/9707021} {arXiv:quant-ph/9707021 [quant-ph]}
  \BibitemShut {NoStop}%
\bibitem [{\citenamefont {Albert}\ \emph {et~al.}(2021)\citenamefont {Albert},
  \citenamefont {Aasen}, \citenamefont {Xu}, \citenamefont {Ji}, \citenamefont
  {Alicea},\ and\ \citenamefont {Preskill}}]{albert2021spin}%
  \BibitemOpen
  \bibfield  {author} {\bibinfo {author} {\bibfnamefont {V.~V.}\ \bibnamefont
  {Albert}}, \bibinfo {author} {\bibfnamefont {D.}~\bibnamefont {Aasen}},
  \bibinfo {author} {\bibfnamefont {W.}~\bibnamefont {Xu}}, \bibinfo {author}
  {\bibfnamefont {W.}~\bibnamefont {Ji}}, \bibinfo {author} {\bibfnamefont
  {J.}~\bibnamefont {Alicea}}, \ and\ \bibinfo {author} {\bibfnamefont
  {J.}~\bibnamefont {Preskill}},\ }\href@noop {} {\enquote {\bibinfo {title}
  {Spin chains, defects, and quantum wires for the quantum-double edge},}\ }
  (\bibinfo {year} {2021}),\ \Eprint {http://arxiv.org/abs/2111.12096}
  {arXiv:2111.12096 [cond-mat.str-el]} \BibitemShut {NoStop}%
\bibitem [{\citenamefont {Cordova}\ \emph {et~al.}(2022)\citenamefont
  {Cordova}, \citenamefont {Dumitrescu}, \citenamefont {Intriligator},\ and\
  \citenamefont {Shao}}]{cordova2022snowmass}%
  \BibitemOpen
  \bibfield  {author} {\bibinfo {author} {\bibfnamefont {C.}~\bibnamefont
  {Cordova}}, \bibinfo {author} {\bibfnamefont {T.~T.}\ \bibnamefont
  {Dumitrescu}}, \bibinfo {author} {\bibfnamefont {K.}~\bibnamefont
  {Intriligator}}, \ and\ \bibinfo {author} {\bibfnamefont {S.-H.}\
  \bibnamefont {Shao}},\ }\href@noop {} {\enquote {\bibinfo {title} {Snowmass
  white paper: Generalized symmetries in quantum field theory and beyond},}\ }
  (\bibinfo {year} {2022}),\ \Eprint {http://arxiv.org/abs/2205.09545}
  {arXiv:2205.09545 [hep-th]} \BibitemShut {NoStop}%
\bibitem [{\citenamefont {Brennan}\ and\ \citenamefont
  {Hong}(2023)}]{brennan2023introduction}%
  \BibitemOpen
  \bibfield  {author} {\bibinfo {author} {\bibfnamefont {T.~D.}\ \bibnamefont
  {Brennan}}\ and\ \bibinfo {author} {\bibfnamefont {S.}~\bibnamefont {Hong}},\
  }\href@noop {} {\enquote {\bibinfo {title} {Introduction to generalized
  global symmetries in {QFT} and particle physics},}\ } (\bibinfo {year}
  {2023}),\ \Eprint {http://arxiv.org/abs/2306.00912} {arXiv:2306.00912
  [hep-ph]} \BibitemShut {NoStop}%
\bibitem [{\citenamefont {McGreevy}(2023)}]{mcgreevy2023generalized}%
  \BibitemOpen
  \bibfield  {author} {\bibinfo {author} {\bibfnamefont {J.}~\bibnamefont
  {McGreevy}},\ }\bibfield  {title} {\enquote {\bibinfo {title} {Generalized
  symmetries in condensed matter},}\ }\href
  {https://www.annualreviews.org/content/journals/10.1146/annurev-conmatphys-040721-021029}
  {\bibfield  {journal} {\bibinfo  {journal} {Annual Review of Condensed Matter
  Physics}\ }\textbf {\bibinfo {volume} {14}},\ \bibinfo {pages} {57} (\bibinfo
  {year} {2023})},\ \Eprint {http://arxiv.org/abs/2204.03045} {arXiv:2204.03045
  [cond-mat.str-el]} \BibitemShut {NoStop}%
\bibitem [{\citenamefont {Luo}\ \emph {et~al.}(2024)\citenamefont {Luo},
  \citenamefont {Wang},\ and\ \citenamefont {Wang}}]{luo2023lecture}%
  \BibitemOpen
  \bibfield  {author} {\bibinfo {author} {\bibfnamefont {R.}~\bibnamefont
  {Luo}}, \bibinfo {author} {\bibfnamefont {Q.-R.}\ \bibnamefont {Wang}}, \
  and\ \bibinfo {author} {\bibfnamefont {Y.-N.}\ \bibnamefont {Wang}},\
  }\bibfield  {title} {\enquote {\bibinfo {title} {Lecture notes on generalized
  symmetries and applications},}\ }\href {\doibase
  https://doi.org/10.1016/j.physrep.2024.02.002} {\bibfield  {journal}
  {\bibinfo  {journal} {Physics Reports}\ }\textbf {\bibinfo {volume} {1065}},\
  \bibinfo {pages} {1} (\bibinfo {year} {2024})},\ \Eprint
  {http://arxiv.org/abs/2307.09215} {arXiv:2307.09215 [hep-th]} \BibitemShut
  {NoStop}%
\bibitem [{\citenamefont {Shao}(2024)}]{shao2024whats}%
  \BibitemOpen
  \bibfield  {author} {\bibinfo {author} {\bibfnamefont {S.-H.}\ \bibnamefont
  {Shao}},\ }\href@noop {} {\enquote {\bibinfo {title} {What's done cannot be
  undone: Tasi lectures on non-invertible symmetries},}\ } (\bibinfo {year}
  {2024}),\ \Eprint {http://arxiv.org/abs/2308.00747} {arXiv:2308.00747
  [hep-th]} \BibitemShut {NoStop}%
\bibitem [{\citenamefont {Sch\"{a}fer-Nameki}(2024)}]{SchaferNameki2024ICTP}%
  \BibitemOpen
  \bibfield  {author} {\bibinfo {author} {\bibfnamefont {S.}~\bibnamefont
  {Sch\"{a}fer-Nameki}},\ }\bibfield  {title} {\enquote {\bibinfo {title}
  {{ICTP} lecture on (non-)invertible generalized symmetries},}\ }\href
  {\doibase https://doi.org/10.1016/j.physrep.2024.01.007} {\bibfield
  {journal} {\bibinfo  {journal} {Physics Reports}\ }\textbf {\bibinfo {volume}
  {1063}},\ \bibinfo {pages} {1} (\bibinfo {year} {2024})},\ \Eprint
  {http://arxiv.org/abs/2305.18296} {arXiv:2305.18296 [hep-th]} \BibitemShut
  {NoStop}%
\bibitem [{\citenamefont {Bhardwaj}\ \emph
  {et~al.}(2024{\natexlab{a}})\citenamefont {Bhardwaj}, \citenamefont
  {Bottini}, \citenamefont {Fraser-Taliente}, \citenamefont {Gladden},
  \citenamefont {Gould}, \citenamefont {Platschorre},\ and\ \citenamefont
  {Tillim}}]{Bhardwaj2024lecture}%
  \BibitemOpen
  \bibfield  {author} {\bibinfo {author} {\bibfnamefont {L.}~\bibnamefont
  {Bhardwaj}}, \bibinfo {author} {\bibfnamefont {L.~E.}\ \bibnamefont
  {Bottini}}, \bibinfo {author} {\bibfnamefont {L.}~\bibnamefont
  {Fraser-Taliente}}, \bibinfo {author} {\bibfnamefont {L.}~\bibnamefont
  {Gladden}}, \bibinfo {author} {\bibfnamefont {D.~S.}\ \bibnamefont {Gould}},
  \bibinfo {author} {\bibfnamefont {A.}~\bibnamefont {Platschorre}}, \ and\
  \bibinfo {author} {\bibfnamefont {H.}~\bibnamefont {Tillim}},\ }\bibfield
  {title} {\enquote {\bibinfo {title} {Lectures on generalized symmetries},}\
  }\href {\doibase https://doi.org/10.1016/j.physrep.2023.11.002} {\bibfield
  {journal} {\bibinfo  {journal} {Physics Reports}\ }\textbf {\bibinfo {volume}
  {1051}},\ \bibinfo {pages} {1} (\bibinfo {year} {2024}{\natexlab{a}})},\
  \Eprint {http://arxiv.org/abs/2307.07547} {arXiv:2307.07547 [hep-th]}
  \BibitemShut {NoStop}%
\bibitem [{\citenamefont {Delcamp}\ and\ \citenamefont
  {Tiwari}(2024)}]{delcamp2024higher}%
  \BibitemOpen
  \bibfield  {author} {\bibinfo {author} {\bibfnamefont {C.}~\bibnamefont
  {Delcamp}}\ and\ \bibinfo {author} {\bibfnamefont {A.}~\bibnamefont
  {Tiwari}},\ }\bibfield  {title} {\enquote {\bibinfo {title} {Higher
  categorical symmetries and gauging in two-dimensional spin systems},}\ }\href
  {https://scipost.org/10.21468/SciPostPhys.16.4.110} {\bibfield  {journal}
  {\bibinfo  {journal} {SciPost Physics}\ }\textbf {\bibinfo {volume} {16}},\
  \bibinfo {pages} {110} (\bibinfo {year} {2024})},\ \Eprint
  {http://arxiv.org/abs/2301.01259} {arXiv:2301.01259 [hep-th]} \BibitemShut
  {NoStop}%
\bibitem [{\citenamefont {Fr\"{o}hlich}\ \emph {et~al.}(2004)\citenamefont
  {Fr\"{o}hlich}, \citenamefont {Fuchs}, \citenamefont {Runkel},\ and\
  \citenamefont {Schweigert}}]{Frohlich2004kramers}%
  \BibitemOpen
  \bibfield  {author} {\bibinfo {author} {\bibfnamefont {J.}~\bibnamefont
  {Fr\"{o}hlich}}, \bibinfo {author} {\bibfnamefont {J.}~\bibnamefont {Fuchs}},
  \bibinfo {author} {\bibfnamefont {I.}~\bibnamefont {Runkel}}, \ and\ \bibinfo
  {author} {\bibfnamefont {C.}~\bibnamefont {Schweigert}},\ }\bibfield  {title}
  {\enquote {\bibinfo {title} {Kramers-wannier duality from conformal
  defects},}\ }\href {\doibase 10.1103/PhysRevLett.93.070601} {\bibfield
  {journal} {\bibinfo  {journal} {Phys. Rev. Lett.}\ }\textbf {\bibinfo
  {volume} {93}},\ \bibinfo {pages} {070601} (\bibinfo {year} {2004})},\
  \Eprint {http://arxiv.org/abs/cond-mat/0404051} {arXiv:cond-mat/0404051
  [cond-mat.stat-mech]} \BibitemShut {NoStop}%
\bibitem [{\citenamefont {Fuchs}\ \emph {et~al.}(2007)\citenamefont {Fuchs},
  \citenamefont {Gaberdiel}, \citenamefont {Runkel},\ and\ \citenamefont
  {Schweigert}}]{fuchs2007topological}%
  \BibitemOpen
  \bibfield  {author} {\bibinfo {author} {\bibfnamefont {J.}~\bibnamefont
  {Fuchs}}, \bibinfo {author} {\bibfnamefont {M.~R.}\ \bibnamefont
  {Gaberdiel}}, \bibinfo {author} {\bibfnamefont {I.}~\bibnamefont {Runkel}}, \
  and\ \bibinfo {author} {\bibfnamefont {C.}~\bibnamefont {Schweigert}},\
  }\bibfield  {title} {\enquote {\bibinfo {title} {Topological defects for the
  free boson cft},}\ }\href
  {https://iopscience.iop.org/article/10.1088/1751-8113/40/37/016} {\bibfield
  {journal} {\bibinfo  {journal} {Journal of Physics A: Mathematical and
  Theoretical}\ }\textbf {\bibinfo {volume} {40}},\ \bibinfo {pages} {11403}
  (\bibinfo {year} {2007})},\ \Eprint {http://arxiv.org/abs/0705.3129}
  {arXiv:0705.3129 [hep-th]} \BibitemShut {NoStop}%
\bibitem [{\citenamefont {Fr{\"o}hlich}\ \emph {et~al.}(2010)\citenamefont
  {Fr{\"o}hlich}, \citenamefont {Fuchs}, \citenamefont {Runkel},\ and\
  \citenamefont {Schweigert}}]{frohlich2010defect}%
  \BibitemOpen
  \bibfield  {author} {\bibinfo {author} {\bibfnamefont {J.}~\bibnamefont
  {Fr{\"o}hlich}}, \bibinfo {author} {\bibfnamefont {J.}~\bibnamefont {Fuchs}},
  \bibinfo {author} {\bibfnamefont {I.}~\bibnamefont {Runkel}}, \ and\ \bibinfo
  {author} {\bibfnamefont {C.}~\bibnamefont {Schweigert}},\ }\bibfield  {title}
  {\enquote {\bibinfo {title} {Defect lines, dualities and generalised
  orbifolds},}\ }in\ \href
  {https://www.worldscientific.com/doi/abs/10.1142/9789814304634_0056} {\emph
  {\bibinfo {booktitle} {XVIth International Congress On Mathematical Physics:
  (With DVD-ROM)}}}\ (\bibinfo {organization} {World Scientific},\ \bibinfo
  {year} {2010})\ pp.\ \bibinfo {pages} {608--613},\ \Eprint
  {http://arxiv.org/abs/0909.5013} {arXiv:0909.5013 [math-ph]} \BibitemShut
  {NoStop}%
\bibitem [{\citenamefont {Bhardwaj}\ and\ \citenamefont
  {Tachikawa}(2018)}]{bhardwaj2018finite}%
  \BibitemOpen
  \bibfield  {author} {\bibinfo {author} {\bibfnamefont {L.}~\bibnamefont
  {Bhardwaj}}\ and\ \bibinfo {author} {\bibfnamefont {Y.}~\bibnamefont
  {Tachikawa}},\ }\bibfield  {title} {\enquote {\bibinfo {title} {On finite
  symmetries and their gauging in two dimensions},}\ }\href@noop {} {\bibfield
  {journal} {\bibinfo  {journal} {Journal of High Energy Physics}\ }\textbf
  {\bibinfo {volume} {2018}},\ \bibinfo {pages} {1} (\bibinfo {year} {2018})},\
  \Eprint {http://arxiv.org/abs/1704.02330} {arXiv:1704.02330 [hep-th]}
  \BibitemShut {NoStop}%
\bibitem [{\citenamefont {Chang}\ \emph {et~al.}(2019)\citenamefont {Chang},
  \citenamefont {Lin}, \citenamefont {Shao}, \citenamefont {Wang},\ and\
  \citenamefont {Yin}}]{chang2019topological}%
  \BibitemOpen
  \bibfield  {author} {\bibinfo {author} {\bibfnamefont {C.-M.}\ \bibnamefont
  {Chang}}, \bibinfo {author} {\bibfnamefont {Y.-H.}\ \bibnamefont {Lin}},
  \bibinfo {author} {\bibfnamefont {S.-H.}\ \bibnamefont {Shao}}, \bibinfo
  {author} {\bibfnamefont {Y.}~\bibnamefont {Wang}}, \ and\ \bibinfo {author}
  {\bibfnamefont {X.}~\bibnamefont {Yin}},\ }\bibfield  {title} {\enquote
  {\bibinfo {title} {Topological defect lines and renormalization group flows
  in two dimensions},}\ }\href
  {https://link.springer.com/article/10.1007/JHEP01(2019)026} {\bibfield
  {journal} {\bibinfo  {journal} {Journal of High Energy Physics}\ }\textbf
  {\bibinfo {volume} {2019}},\ \bibinfo {pages} {1} (\bibinfo {year} {2019})},\
  \Eprint {http://arxiv.org/abs/1802.04445} {arXiv:1802.04445 [hep-th]}
  \BibitemShut {NoStop}%
\bibitem [{\citenamefont {Thorngren}\ and\ \citenamefont
  {Wang}(2024)}]{thorngren2019fusion}%
  \BibitemOpen
  \bibfield  {author} {\bibinfo {author} {\bibfnamefont {R.}~\bibnamefont
  {Thorngren}}\ and\ \bibinfo {author} {\bibfnamefont {Y.}~\bibnamefont
  {Wang}},\ }\bibfield  {title} {\enquote {\bibinfo {title} {Fusion category
  symmetry. part i. anomaly in-flow and gapped phases},}\ }\href
  {https://link.springer.com/article/10.1007/JHEP04(2024)132} {\bibfield
  {journal} {\bibinfo  {journal} {Journal of High Energy Physics}\ }\textbf
  {\bibinfo {volume} {2024}},\ \bibinfo {pages} {1} (\bibinfo {year} {2024})},\
  \Eprint {http://arxiv.org/abs/1912.02817} {arXiv:1912.02817 [hep-th]}
  \BibitemShut {NoStop}%
\bibitem [{\citenamefont {Thorngren}\ and\ \citenamefont
  {Wang}(2021)}]{thorngren2021fusion}%
  \BibitemOpen
  \bibfield  {author} {\bibinfo {author} {\bibfnamefont {R.}~\bibnamefont
  {Thorngren}}\ and\ \bibinfo {author} {\bibfnamefont {Y.}~\bibnamefont
  {Wang}},\ }\href@noop {} {\enquote {\bibinfo {title} {Fusion category
  symmetry ii: Categoriosities at c = 1 and beyond},}\ } (\bibinfo {year}
  {2021}),\ \Eprint {http://arxiv.org/abs/2106.12577} {arXiv:2106.12577
  [hep-th]} \BibitemShut {NoStop}%
\bibitem [{\citenamefont {Komargodski}\ \emph {et~al.}(2021)\citenamefont
  {Komargodski}, \citenamefont {Ohmori}, \citenamefont {Roumpedakis},\ and\
  \citenamefont {Seifnashri}}]{komargodski2021symmetries}%
  \BibitemOpen
  \bibfield  {author} {\bibinfo {author} {\bibfnamefont {Z.}~\bibnamefont
  {Komargodski}}, \bibinfo {author} {\bibfnamefont {K.}~\bibnamefont {Ohmori}},
  \bibinfo {author} {\bibfnamefont {K.}~\bibnamefont {Roumpedakis}}, \ and\
  \bibinfo {author} {\bibfnamefont {S.}~\bibnamefont {Seifnashri}},\ }\bibfield
   {title} {\enquote {\bibinfo {title} {Symmetries and strings of adjoint
  qcd2},}\ }\href {https://link.springer.com/article/10.1007/JHEP03(2021)103}
  {\bibfield  {journal} {\bibinfo  {journal} {Journal of High Energy Physics}\
  }\textbf {\bibinfo {volume} {2021}},\ \bibinfo {pages} {1} (\bibinfo {year}
  {2021})},\ \Eprint {http://arxiv.org/abs/2008.07567} {arXiv:2008.07567
  [hep-th]} \BibitemShut {NoStop}%
\bibitem [{\citenamefont {Inamura}(2022)}]{inamura2022lattice}%
  \BibitemOpen
  \bibfield  {author} {\bibinfo {author} {\bibfnamefont {K.}~\bibnamefont
  {Inamura}},\ }\bibfield  {title} {\enquote {\bibinfo {title} {On lattice
  models of gapped phases with fusion category symmetries},}\ }\href
  {https://link.springer.com/article/10.1007/JHEP03(2022)036} {\bibfield
  {journal} {\bibinfo  {journal} {Journal of High Energy Physics}\ }\textbf
  {\bibinfo {volume} {2022}},\ \bibinfo {pages} {1} (\bibinfo {year} {2022})},\
  \Eprint {http://arxiv.org/abs/2110.12882} {arXiv:2110.12882
  [cond-mat.str-el]} \BibitemShut {NoStop}%
\bibitem [{\citenamefont {Inamura}(2023)}]{inamura2023fermionization}%
  \BibitemOpen
  \bibfield  {author} {\bibinfo {author} {\bibfnamefont {K.}~\bibnamefont
  {Inamura}},\ }\bibfield  {title} {\enquote {\bibinfo {title} {Fermionization
  of fusion category symmetries in 1+1 dimensions},}\ }\href
  {https://link.springer.com/article/10.1007/JHEP10(2023)101} {\bibfield
  {journal} {\bibinfo  {journal} {Journal of High Energy Physics}\ }\textbf
  {\bibinfo {volume} {2023}},\ \bibinfo {pages} {1} (\bibinfo {year} {2023})},\
  \Eprint {http://arxiv.org/abs/2206.13159} {arXiv:2206.13159
  [cond-mat.str-el]} \BibitemShut {NoStop}%
\bibitem [{\citenamefont {Gaiotto}\ \emph {et~al.}(2015)\citenamefont
  {Gaiotto}, \citenamefont {Kapustin}, \citenamefont {Seiberg},\ and\
  \citenamefont {Willett}}]{gaiotto2015generalized}%
  \BibitemOpen
  \bibfield  {author} {\bibinfo {author} {\bibfnamefont {D.}~\bibnamefont
  {Gaiotto}}, \bibinfo {author} {\bibfnamefont {A.}~\bibnamefont {Kapustin}},
  \bibinfo {author} {\bibfnamefont {N.}~\bibnamefont {Seiberg}}, \ and\
  \bibinfo {author} {\bibfnamefont {B.}~\bibnamefont {Willett}},\ }\bibfield
  {title} {\enquote {\bibinfo {title} {Generalized global symmetries},}\
  }\href@noop {} {\bibfield  {journal} {\bibinfo  {journal} {Journal of High
  Energy Physics}\ }\textbf {\bibinfo {volume} {2015}},\ \bibinfo {pages} {1}
  (\bibinfo {year} {2015})},\ \Eprint {http://arxiv.org/abs/1412.5148}
  {arXiv:1412.5148 [hep-th]} \BibitemShut {NoStop}%
\bibitem [{\citenamefont {Kapustin}\ and\ \citenamefont
  {Thorngren}(2017)}]{kapustin2017higher}%
  \BibitemOpen
  \bibfield  {author} {\bibinfo {author} {\bibfnamefont {A.}~\bibnamefont
  {Kapustin}}\ and\ \bibinfo {author} {\bibfnamefont {R.}~\bibnamefont
  {Thorngren}},\ }\bibfield  {title} {\enquote {\bibinfo {title} {Higher
  symmetry and gapped phases of gauge theories},}\ }\href
  {https://link-springer-com.libproxy1.nus.edu.sg/chapter/10.1007/978-3-319-59939-7_5}
  {\bibfield  {journal} {\bibinfo  {journal} {Algebra, Geometry, and Physics in
  the 21st Century: Kontsevich Festschrift}\ ,\ \bibinfo {pages} {177}}
  (\bibinfo {year} {2017})}\BibitemShut {NoStop}%
\bibitem [{\citenamefont {Gomes}(2023)}]{gomes2023introduction}%
  \BibitemOpen
  \bibfield  {author} {\bibinfo {author} {\bibfnamefont {P.~R.}\ \bibnamefont
  {Gomes}},\ }\bibfield  {title} {\enquote {\bibinfo {title} {An introduction
  to higher-form symmetries},}\ }\href
  {https://scipost.org/10.21468/SciPostPhysLectNotes.74} {\bibfield  {journal}
  {\bibinfo  {journal} {SciPost Physics Lecture Notes}\ ,\ \bibinfo {pages}
  {074}} (\bibinfo {year} {2023})},\ \Eprint {http://arxiv.org/abs/2303.01817}
  {arXiv:2303.01817 [hep-th]} \BibitemShut {NoStop}%
\bibitem [{\citenamefont {Bhardwaj}\ \emph
  {et~al.}(2024{\natexlab{b}})\citenamefont {Bhardwaj}, \citenamefont
  {Bottini}, \citenamefont {Schafer-Nameki},\ and\ \citenamefont
  {Tiwari}}]{bhardwaj2024illustrating}%
  \BibitemOpen
  \bibfield  {author} {\bibinfo {author} {\bibfnamefont {L.}~\bibnamefont
  {Bhardwaj}}, \bibinfo {author} {\bibfnamefont {L.~E.}\ \bibnamefont
  {Bottini}}, \bibinfo {author} {\bibfnamefont {S.}~\bibnamefont
  {Schafer-Nameki}}, \ and\ \bibinfo {author} {\bibfnamefont {A.}~\bibnamefont
  {Tiwari}},\ }\href@noop {} {\enquote {\bibinfo {title} {Illustrating the
  categorical landau paradigm in lattice models},}\ } (\bibinfo {year}
  {2024}{\natexlab{b}}),\ \Eprint {http://arxiv.org/abs/2405.05302}
  {arXiv:2405.05302 [cond-mat.str-el]} \BibitemShut {NoStop}%
\bibitem [{\citenamefont {Bhardwaj}\ \emph
  {et~al.}(2024{\natexlab{c}})\citenamefont {Bhardwaj}, \citenamefont
  {Bottini}, \citenamefont {Schafer-Nameki},\ and\ \citenamefont
  {Tiwari}}]{bhardwaj2024lattice}%
  \BibitemOpen
  \bibfield  {author} {\bibinfo {author} {\bibfnamefont {L.}~\bibnamefont
  {Bhardwaj}}, \bibinfo {author} {\bibfnamefont {L.~E.}\ \bibnamefont
  {Bottini}}, \bibinfo {author} {\bibfnamefont {S.}~\bibnamefont
  {Schafer-Nameki}}, \ and\ \bibinfo {author} {\bibfnamefont {A.}~\bibnamefont
  {Tiwari}},\ }\href@noop {} {\enquote {\bibinfo {title} {Lattice models for
  phases and transitions with non-invertible symmetries},}\ } (\bibinfo {year}
  {2024}{\natexlab{c}}),\ \Eprint {http://arxiv.org/abs/2405.05964}
  {arXiv:2405.05964 [cond-mat.str-el]} \BibitemShut {NoStop}%
\bibitem [{\citenamefont {Jia}()}]{jia2024cluster}%
  \BibitemOpen
  \bibfield  {author} {\bibinfo {author} {\bibfnamefont {Z.}~\bibnamefont
  {Jia}},\ }\href@noop {} {\enquote {\bibinfo {title} {{Cluster
  symmetry-protected topological phases from Hopf symmetries}},}\ }\bibinfo
  {howpublished} {in preparation}\BibitemShut {NoStop}%
\bibitem [{\citenamefont {Buerschaper}\ \emph
  {et~al.}(2013{\natexlab{a}})\citenamefont {Buerschaper}, \citenamefont
  {Mombelli}, \citenamefont {Christandl},\ and\ \citenamefont
  {Aguado}}]{Buerschaper2013a}%
  \BibitemOpen
  \bibfield  {author} {\bibinfo {author} {\bibfnamefont {O.}~\bibnamefont
  {Buerschaper}}, \bibinfo {author} {\bibfnamefont {J.~M.}\ \bibnamefont
  {Mombelli}}, \bibinfo {author} {\bibfnamefont {M.}~\bibnamefont
  {Christandl}}, \ and\ \bibinfo {author} {\bibfnamefont {M.}~\bibnamefont
  {Aguado}},\ }\bibfield  {title} {\enquote {\bibinfo {title} {A hierarchy of
  topological tensor network states},}\ }\href {\doibase 10.1063/1.4773316}
  {\bibfield  {journal} {\bibinfo  {journal} {Journal of Mathematical Physics}\
  }\textbf {\bibinfo {volume} {54}},\ \bibinfo {pages} {012201} (\bibinfo
  {year} {2013}{\natexlab{a}})},\ \Eprint {http://arxiv.org/abs/1007.5283}
  {arXiv:1007.5283 [cond-mat.str-el]} \BibitemShut {NoStop}%
\bibitem [{\citenamefont {Yan}\ \emph {et~al.}(2022)\citenamefont {Yan},
  \citenamefont {Chen},\ and\ \citenamefont {Cui}}]{chen2021ribbon}%
  \BibitemOpen
  \bibfield  {author} {\bibinfo {author} {\bibfnamefont {B.}~\bibnamefont
  {Yan}}, \bibinfo {author} {\bibfnamefont {P.}~\bibnamefont {Chen}}, \ and\
  \bibinfo {author} {\bibfnamefont {S.}~\bibnamefont {Cui}},\ }\bibfield
  {title} {\enquote {\bibinfo {title} {Ribbon operators in the generalized
  {K}itaev quantum double model based on {H}opf algebras},}\ }\href
  {https://iopscience.iop.org/article/10.1088/1751-8121/ac552c/meta} {\bibfield
   {journal} {\bibinfo  {journal} {Journal of Physics A: Mathematical and
  Theoretical}\ } (\bibinfo {year} {2022})},\ \Eprint
  {http://arxiv.org/abs/2105.08202} {arXiv:2105.08202 [cond-mat.str-el]}
  \BibitemShut {NoStop}%
\bibitem [{\citenamefont {Jia}\ \emph {et~al.}(2023{\natexlab{a}})\citenamefont
  {Jia}, \citenamefont {Kaszlikowski},\ and\ \citenamefont
  {Tan}}]{jia2023boundary}%
  \BibitemOpen
  \bibfield  {author} {\bibinfo {author} {\bibfnamefont {Z.}~\bibnamefont
  {Jia}}, \bibinfo {author} {\bibfnamefont {D.}~\bibnamefont {Kaszlikowski}}, \
  and\ \bibinfo {author} {\bibfnamefont {S.}~\bibnamefont {Tan}},\ }\bibfield
  {title} {\enquote {\bibinfo {title} {Boundary and domain wall theories of 2d
  generalized quantum double model},}\ }\href
  {https://link.springer.com/article/10.1007/JHEP07(2023)160} {\bibfield
  {journal} {\bibinfo  {journal} {Journal of High Energy Physics}\ }\textbf
  {\bibinfo {volume} {2023}},\ \bibinfo {pages} {1} (\bibinfo {year}
  {2023}{\natexlab{a}})},\ \Eprint {http://arxiv.org/abs/2207.03970}
  {arXiv:2207.03970 [quant-ph]} \BibitemShut {NoStop}%
\bibitem [{\citenamefont {Jia}\ \emph {et~al.}(2023{\natexlab{b}})\citenamefont
  {Jia}, \citenamefont {Tan}, \citenamefont {Kaszlikowski},\ and\ \citenamefont
  {Chang}}]{Jia2023weak}%
  \BibitemOpen
  \bibfield  {author} {\bibinfo {author} {\bibfnamefont {Z.}~\bibnamefont
  {Jia}}, \bibinfo {author} {\bibfnamefont {S.}~\bibnamefont {Tan}}, \bibinfo
  {author} {\bibfnamefont {D.}~\bibnamefont {Kaszlikowski}}, \ and\ \bibinfo
  {author} {\bibfnamefont {L.}~\bibnamefont {Chang}},\ }\bibfield  {title}
  {\enquote {\bibinfo {title} {On weak {H}opf symmetry and weak {H}opf quantum
  double model},}\ }\href {\doibase 10.1007/s00220-023-04792-9} {\bibfield
  {journal} {\bibinfo  {journal} {Communications in Mathematical Physics}\
  }\textbf {\bibinfo {volume} {402}},\ \bibinfo {pages} {3045} (\bibinfo {year}
  {2023}{\natexlab{b}})},\ \Eprint {http://arxiv.org/abs/2302.08131}
  {arXiv:2302.08131 [hep-th]} \BibitemShut {NoStop}%
\bibitem [{\citenamefont {Abe}(2004)}]{abe2004hopf}%
  \BibitemOpen
  \bibfield  {author} {\bibinfo {author} {\bibfnamefont {E.}~\bibnamefont
  {Abe}},\ }\href@noop {} {\emph {\bibinfo {title} {Hopf algebras}}},\ \bibinfo
  {series} {Cambridge Tracts in Mathematics}, Vol.~\bibinfo {volume} {74}\
  (\bibinfo  {publisher} {Cambridge University Press},\ \bibinfo {year}
  {2004})\ pp.\ \bibinfo {pages} {xii+284}\BibitemShut {NoStop}%
\bibitem [{\citenamefont {Kassel}(1995)}]{kassel2012quantum}%
  \BibitemOpen
  \bibfield  {author} {\bibinfo {author} {\bibfnamefont {C.}~\bibnamefont
  {Kassel}},\ }\href {https://link.springer.com/book/10.1007/978-1-4612-0783-2}
  {\emph {\bibinfo {title} {Quantum groups}}},\ \bibinfo {series} {Graduate
  Texts in Mathematics}, Vol.\ \bibinfo {volume} {155}\ (\bibinfo  {publisher}
  {Springer-Verlag, New York},\ \bibinfo {year} {1995})\ pp.\ \bibinfo {pages}
  {xii+531}\BibitemShut {NoStop}%
\bibitem [{\citenamefont {Turaev}(2016)}]{turaev2016quantum}%
  \BibitemOpen
  \bibfield  {author} {\bibinfo {author} {\bibfnamefont {V.~G.}\ \bibnamefont
  {Turaev}},\ }\href
  {https://www.degruyter.com/document/doi/10.1515/9783110435221/html} {\emph
  {\bibinfo {title} {Quantum invariants of knots and 3-manifolds}}},\
  Vol.~\bibinfo {volume} {18}\ (\bibinfo  {publisher} {De Gruyter},\ \bibinfo
  {year} {2016})\ pp.\ \bibinfo {pages} {xii+592}\BibitemShut {NoStop}%
\bibitem [{\citenamefont {Bakalov}\ and\ \citenamefont
  {Kirillov}(2001)}]{bakalov2001lectures}%
  \BibitemOpen
  \bibfield  {author} {\bibinfo {author} {\bibfnamefont {B.}~\bibnamefont
  {Bakalov}}\ and\ \bibinfo {author} {\bibfnamefont {A.~A.}\ \bibnamefont
  {Kirillov}},\ }\href {http://bookstore.ams.org/ulect-21} {\emph {\bibinfo
  {title} {Lectures on tensor categories and modular functors}}},\
  Vol.~\bibinfo {volume} {21}\ (\bibinfo  {publisher} {American Mathematical
  Soc.},\ \bibinfo {year} {2001})\BibitemShut {NoStop}%
\bibitem [{\citenamefont {Larson}\ and\ \citenamefont
  {Radford}(1988)}]{larson1988semisimple}%
  \BibitemOpen
  \bibfield  {author} {\bibinfo {author} {\bibfnamefont {R.~G.}\ \bibnamefont
  {Larson}}\ and\ \bibinfo {author} {\bibfnamefont {D.~E.}\ \bibnamefont
  {Radford}},\ }\bibfield  {title} {\enquote {\bibinfo {title} {Semisimple
  cosemisimple hopf algebras},}\ }\href {https://www.jstor.org/stable/2374545}
  {\bibfield  {journal} {\bibinfo  {journal} {American Journal of Mathematics}\
  }\textbf {\bibinfo {volume} {110}},\ \bibinfo {pages} {187} (\bibinfo {year}
  {1988})}\BibitemShut {NoStop}%
\bibitem [{\citenamefont {Larson}(1971)}]{larson1971characters}%
  \BibitemOpen
  \bibfield  {author} {\bibinfo {author} {\bibfnamefont {R.~G.}\ \bibnamefont
  {Larson}},\ }\bibfield  {title} {\enquote {\bibinfo {title} {Characters of
  hopf algebras},}\ }\href@noop {} {\bibfield  {journal} {\bibinfo  {journal}
  {Journal of algebra}\ }\textbf {\bibinfo {volume} {17}},\ \bibinfo {pages}
  {352} (\bibinfo {year} {1971})}\BibitemShut {NoStop}%
\bibitem [{\citenamefont {Etingof}\ \emph {et~al.}(2016)\citenamefont
  {Etingof}, \citenamefont {Gelaki}, \citenamefont {Nikshych},\ and\
  \citenamefont {Ostrik}}]{etingof2016tensor}%
  \BibitemOpen
  \bibfield  {author} {\bibinfo {author} {\bibfnamefont {P.}~\bibnamefont
  {Etingof}}, \bibinfo {author} {\bibfnamefont {S.}~\bibnamefont {Gelaki}},
  \bibinfo {author} {\bibfnamefont {D.}~\bibnamefont {Nikshych}}, \ and\
  \bibinfo {author} {\bibfnamefont {V.}~\bibnamefont {Ostrik}},\ }\href
  {https://bookstore.ams.org/surv-205} {\emph {\bibinfo {title} {Tensor
  categories}}},\ Vol.\ \bibinfo {volume} {205}\ (\bibinfo  {publisher}
  {American Mathematical Soc.},\ \bibinfo {year} {2016})\ pp.\ \bibinfo {pages}
  {xvi+343}\BibitemShut {NoStop}%
\bibitem [{\citenamefont {Nikshych}(2002)}]{nikshych2002structure}%
  \BibitemOpen
  \bibfield  {author} {\bibinfo {author} {\bibfnamefont {D.}~\bibnamefont
  {Nikshych}},\ }\bibfield  {title} {\enquote {\bibinfo {title} {On the
  structure of weak hopf algebras},}\ }\href
  {https://www.sciencedirect.com/science/article/pii/S0001870802920815}
  {\bibfield  {journal} {\bibinfo  {journal} {Advances in Mathematics}\
  }\textbf {\bibinfo {volume} {170}},\ \bibinfo {pages} {257} (\bibinfo {year}
  {2002})},\ \Eprint {http://arxiv.org/abs/math/0106010} {arXiv:math/0106010
  [math.QA]} \BibitemShut {NoStop}%
\bibitem [{\citenamefont {Bridgeman}\ \emph {et~al.}(2023)\citenamefont
  {Bridgeman}, \citenamefont {Lootens},\ and\ \citenamefont
  {Verstraete}}]{bridgeman2023invertible}%
  \BibitemOpen
  \bibfield  {author} {\bibinfo {author} {\bibfnamefont {J.~C.}\ \bibnamefont
  {Bridgeman}}, \bibinfo {author} {\bibfnamefont {L.}~\bibnamefont {Lootens}},
  \ and\ \bibinfo {author} {\bibfnamefont {F.}~\bibnamefont {Verstraete}},\
  }\bibfield  {title} {\enquote {\bibinfo {title} {Invertible bimodule
  categories and generalized {S}chur orthogonality},}\ }\href
  {https://link.springer.com/article/10.1007/s00220-023-04781-y} {\bibfield
  {journal} {\bibinfo  {journal} {Communications in Mathematical Physics}\
  }\textbf {\bibinfo {volume} {402}},\ \bibinfo {pages} {2691} (\bibinfo {year}
  {2023})},\ \Eprint {http://arxiv.org/abs/2211.01947} {arXiv:2211.01947
  [math.QA]} \BibitemShut {NoStop}%
\bibitem [{\citenamefont {Buerschaper}\ \emph
  {et~al.}(2013{\natexlab{b}})\citenamefont {Buerschaper}, \citenamefont
  {Christandl}, \citenamefont {Kong},\ and\ \citenamefont
  {Aguado}}]{buerschaper2013electric}%
  \BibitemOpen
  \bibfield  {author} {\bibinfo {author} {\bibfnamefont {O.}~\bibnamefont
  {Buerschaper}}, \bibinfo {author} {\bibfnamefont {M.}~\bibnamefont
  {Christandl}}, \bibinfo {author} {\bibfnamefont {L.}~\bibnamefont {Kong}}, \
  and\ \bibinfo {author} {\bibfnamefont {M.}~\bibnamefont {Aguado}},\
  }\bibfield  {title} {\enquote {\bibinfo {title} {Electric--magnetic duality
  of lattice systems with topological order},}\ }\href
  {https://www.sciencedirect.com/science/article/abs/pii/S0550321313004367?via%3Dihub}
  {\bibfield  {journal} {\bibinfo  {journal} {Nuclear Physics B}\ }\textbf
  {\bibinfo {volume} {876}},\ \bibinfo {pages} {619} (\bibinfo {year}
  {2013}{\natexlab{b}})},\ \Eprint {http://arxiv.org/abs/1006.5823}
  {arXiv:1006.5823 [cond-mat.str-el]} \BibitemShut {NoStop}%
\bibitem [{\citenamefont {Bais}\ \emph {et~al.}(2003)\citenamefont {Bais},
  \citenamefont {Schroers},\ and\ \citenamefont {Slingerland}}]{bais2003hopf}%
  \BibitemOpen
  \bibfield  {author} {\bibinfo {author} {\bibfnamefont {A.~F.}\ \bibnamefont
  {Bais}}, \bibinfo {author} {\bibfnamefont {B.~J.}\ \bibnamefont {Schroers}},
  \ and\ \bibinfo {author} {\bibfnamefont {J.~K.}\ \bibnamefont
  {Slingerland}},\ }\bibfield  {title} {\enquote {\bibinfo {title} {Hopf
  symmetry breaking and confinement in (2+1)-dimensional gauge theory},}\
  }\href {https://doi.org/10.1088/1126-6708/2003/05/068} {\bibfield  {journal}
  {\bibinfo  {journal} {Journal of High Energy Physics}\ }\textbf {\bibinfo
  {volume} {2003}},\ \bibinfo {pages} {068} (\bibinfo {year} {2003})},\ \Eprint
  {http://arxiv.org/abs/hep-th/0205114} {arXiv:hep-th/0205114 [hep-th]}
  \BibitemShut {NoStop}%
\bibitem [{\citenamefont {Jia}\ \emph {et~al.}(2024{\natexlab{a}})\citenamefont
  {Jia}, \citenamefont {Tan},\ and\ \citenamefont
  {Kaszlikowski}}]{jia2024weakTube}%
  \BibitemOpen
  \bibfield  {author} {\bibinfo {author} {\bibfnamefont {Z.}~\bibnamefont
  {Jia}}, \bibinfo {author} {\bibfnamefont {S.}~\bibnamefont {Tan}}, \ and\
  \bibinfo {author} {\bibfnamefont {D.}~\bibnamefont {Kaszlikowski}},\
  }\bibfield  {title} {\enquote {\bibinfo {title} {Weak {H}opf symmetry and
  tube algebra of the generalized multifusion string-net model},}\ }\href
  {https://doi.org/10.1007/JHEP07(2024)207} {\bibfield  {journal} {\bibinfo
  {journal} {Journal of High Enegy Physics}\ }\textbf {\bibinfo {volume}
  {07}},\ \bibinfo {pages} {207} (\bibinfo {year} {2024}{\natexlab{a}})},\
  \Eprint {http://arxiv.org/abs/2403.04446} {arXiv:2403.04446 [hep-th]}
  \BibitemShut {NoStop}%
\bibitem [{\citenamefont {B\"{o}hm}\ \emph {et~al.}(1999)\citenamefont
  {B\"{o}hm}, \citenamefont {Nill},\ and\ \citenamefont
  {Szlach\'{a}nyi}}]{BOHM1998weak}%
  \BibitemOpen
  \bibfield  {author} {\bibinfo {author} {\bibfnamefont {G.}~\bibnamefont
  {B\"{o}hm}}, \bibinfo {author} {\bibfnamefont {F.}~\bibnamefont {Nill}}, \
  and\ \bibinfo {author} {\bibfnamefont {K.}~\bibnamefont {Szlach\'{a}nyi}},\
  }\bibfield  {title} {\enquote {\bibinfo {title} {Weak {H}opf algebras: I.
  {I}ntegral theory and ${C}^*$-structure},}\ }\href {\doibase
  https://doi.org/10.1006/jabr.1999.7984} {\bibfield  {journal} {\bibinfo
  {journal} {Journal of Algebra}\ }\textbf {\bibinfo {volume} {221}},\ \bibinfo
  {pages} {385 } (\bibinfo {year} {1999})},\ \Eprint
  {http://arxiv.org/abs/math/9805116} {arXiv:math/9805116 [math.QA]}
  \BibitemShut {NoStop}%
\bibitem [{\citenamefont {Feiguin}\ \emph {et~al.}(2007)\citenamefont
  {Feiguin}, \citenamefont {Trebst}, \citenamefont {Ludwig}, \citenamefont
  {Troyer}, \citenamefont {Kitaev}, \citenamefont {Wang},\ and\ \citenamefont
  {Freedman}}]{Feiguin2007interacting}%
  \BibitemOpen
  \bibfield  {author} {\bibinfo {author} {\bibfnamefont {A.}~\bibnamefont
  {Feiguin}}, \bibinfo {author} {\bibfnamefont {S.}~\bibnamefont {Trebst}},
  \bibinfo {author} {\bibfnamefont {A.~W.~W.}\ \bibnamefont {Ludwig}}, \bibinfo
  {author} {\bibfnamefont {M.}~\bibnamefont {Troyer}}, \bibinfo {author}
  {\bibfnamefont {A.}~\bibnamefont {Kitaev}}, \bibinfo {author} {\bibfnamefont
  {Z.}~\bibnamefont {Wang}}, \ and\ \bibinfo {author} {\bibfnamefont {M.~H.}\
  \bibnamefont {Freedman}},\ }\bibfield  {title} {\enquote {\bibinfo {title}
  {Interacting anyons in topological quantum liquids: The golden chain},}\
  }\href {\doibase 10.1103/PhysRevLett.98.160409} {\bibfield  {journal}
  {\bibinfo  {journal} {Phys. Rev. Lett.}\ }\textbf {\bibinfo {volume} {98}},\
  \bibinfo {pages} {160409} (\bibinfo {year} {2007})},\ \Eprint
  {http://arxiv.org/abs/cond-mat/0612341} {arXiv:cond-mat/0612341
  [cond-mat.str-el]} \BibitemShut {NoStop}%
\bibitem [{\citenamefont {Girelli}\ \emph {et~al.}(2021)\citenamefont
  {Girelli}, \citenamefont {Osei},\ and\ \citenamefont
  {Osumanu}}]{girelli2021semidual}%
  \BibitemOpen
  \bibfield  {author} {\bibinfo {author} {\bibfnamefont {F.}~\bibnamefont
  {Girelli}}, \bibinfo {author} {\bibfnamefont {P.~K.}\ \bibnamefont {Osei}}, \
  and\ \bibinfo {author} {\bibfnamefont {A.}~\bibnamefont {Osumanu}},\
  }\bibfield  {title} {\enquote {\bibinfo {title} {Semidual {K}itaev lattice
  model and tensor network representation},}\ }\href
  {https://link.springer.com/article/10.1007/JHEP09(2021)210} {\bibfield
  {journal} {\bibinfo  {journal} {Journal of High Energy Physics}\ }\textbf
  {\bibinfo {volume} {2021}},\ \bibinfo {pages} {1} (\bibinfo {year} {2021})},\
  \Eprint {http://arxiv.org/abs/1709.00522} {arXiv:1709.00522 [math.QA]}
  \BibitemShut {NoStop}%
\bibitem [{\citenamefont {Or\'{u}s}(2014)}]{Orus2014tensornet}%
  \BibitemOpen
  \bibfield  {author} {\bibinfo {author} {\bibfnamefont {R.}~\bibnamefont
  {Or\'{u}s}},\ }\bibfield  {title} {\enquote {\bibinfo {title} {A practical
  introduction to tensor networks: Matrix product states and projected
  entangled pair states},}\ }\href {\doibase
  https://doi.org/10.1016/j.aop.2014.06.013} {\bibfield  {journal} {\bibinfo
  {journal} {Annals of Physics}\ }\textbf {\bibinfo {volume} {349}},\ \bibinfo
  {pages} {117 } (\bibinfo {year} {2014})},\ \Eprint
  {http://arxiv.org/abs/1306.2164} {arXiv:1306.2164 [cond-mat.str-el]}
  \BibitemShut {NoStop}%
\bibitem [{\citenamefont {Cirac}\ \emph {et~al.}(2021)\citenamefont {Cirac},
  \citenamefont {P\'erez-Garc\'{\i}a}, \citenamefont {Schuch},\ and\
  \citenamefont {Verstraete}}]{Cirac2021MPSreivew}%
  \BibitemOpen
  \bibfield  {author} {\bibinfo {author} {\bibfnamefont {J.~I.}\ \bibnamefont
  {Cirac}}, \bibinfo {author} {\bibfnamefont {D.}~\bibnamefont
  {P\'erez-Garc\'{\i}a}}, \bibinfo {author} {\bibfnamefont {N.}~\bibnamefont
  {Schuch}}, \ and\ \bibinfo {author} {\bibfnamefont {F.}~\bibnamefont
  {Verstraete}},\ }\bibfield  {title} {\enquote {\bibinfo {title} {Matrix
  product states and projected entangled pair states: Concepts, symmetries,
  theorems},}\ }\href {\doibase 10.1103/RevModPhys.93.045003} {\bibfield
  {journal} {\bibinfo  {journal} {Rev. Mod. Phys.}\ }\textbf {\bibinfo {volume}
  {93}},\ \bibinfo {pages} {045003} (\bibinfo {year} {2021})},\ \Eprint
  {http://arxiv.org/abs/2011.12127} {arXiv:2011.12127 [quant-ph]} \BibitemShut
  {NoStop}%
\bibitem [{\citenamefont {Molnar}\ \emph {et~al.}(2022)\citenamefont {Molnar},
  \citenamefont {de~Alarc{\'o}n}, \citenamefont {Garre-Rubio}, \citenamefont
  {Schuch}, \citenamefont {Cirac},\ and\ \citenamefont
  {P{\'e}rez-Garc{\'\i}a}}]{molnar2022matrix}%
  \BibitemOpen
  \bibfield  {author} {\bibinfo {author} {\bibfnamefont {A.}~\bibnamefont
  {Molnar}}, \bibinfo {author} {\bibfnamefont {A.~R.}\ \bibnamefont
  {de~Alarc{\'o}n}}, \bibinfo {author} {\bibfnamefont {J.}~\bibnamefont
  {Garre-Rubio}}, \bibinfo {author} {\bibfnamefont {N.}~\bibnamefont {Schuch}},
  \bibinfo {author} {\bibfnamefont {J.~I.}\ \bibnamefont {Cirac}}, \ and\
  \bibinfo {author} {\bibfnamefont {D.}~\bibnamefont {P{\'e}rez-Garc{\'\i}a}},\
  }\bibfield  {title} {\enquote {\bibinfo {title} {Matrix product operator
  algebras {I}: representations of weak {H}opf algebras and projected entangled
  pair states},}\ }\href {https://arxiv.org/abs/2204.05940} {\bibfield
  {journal} {\bibinfo  {journal} {arXiv preprint arXiv:2204.05940}\ } (\bibinfo
  {year} {2022})}\BibitemShut {NoStop}%
\bibitem [{\citenamefont {Garre-Rubio}\ \emph {et~al.}(2023)\citenamefont
  {Garre-Rubio}, \citenamefont {Lootens},\ and\ \citenamefont
  {Moln{\'{a}}r}}]{GarreRubio2023classifyingphases}%
  \BibitemOpen
  \bibfield  {author} {\bibinfo {author} {\bibfnamefont {J.}~\bibnamefont
  {Garre-Rubio}}, \bibinfo {author} {\bibfnamefont {L.}~\bibnamefont
  {Lootens}}, \ and\ \bibinfo {author} {\bibfnamefont {A.}~\bibnamefont
  {Moln{\'{a}}r}},\ }\bibfield  {title} {\enquote {\bibinfo {title}
  {Classifying phases protected by matrix product operator symmetries using
  matrix product states},}\ }\href {\doibase 10.22331/q-2023-02-21-927}
  {\bibfield  {journal} {\bibinfo  {journal} {{Quantum}}\ }\textbf {\bibinfo
  {volume} {7}},\ \bibinfo {pages} {927} (\bibinfo {year} {2023})},\ \Eprint
  {http://arxiv.org/abs/2203.12563} {arXiv:2203.12563 [cond-mat.str-el]}
  \BibitemShut {NoStop}%
\bibitem [{\citenamefont {Kong}\ \emph {et~al.}(2020)\citenamefont {Kong},
  \citenamefont {Lan}, \citenamefont {Wen}, \citenamefont {Zhang},\ and\
  \citenamefont {Zheng}}]{Kong2020algebraic}%
  \BibitemOpen
  \bibfield  {author} {\bibinfo {author} {\bibfnamefont {L.}~\bibnamefont
  {Kong}}, \bibinfo {author} {\bibfnamefont {T.}~\bibnamefont {Lan}}, \bibinfo
  {author} {\bibfnamefont {X.-G.}\ \bibnamefont {Wen}}, \bibinfo {author}
  {\bibfnamefont {Z.-H.}\ \bibnamefont {Zhang}}, \ and\ \bibinfo {author}
  {\bibfnamefont {H.}~\bibnamefont {Zheng}},\ }\bibfield  {title} {\enquote
  {\bibinfo {title} {Algebraic higher symmetry and categorical symmetry: A
  holographic and entanglement view of symmetry},}\ }\href {\doibase
  10.1103/PhysRevResearch.2.043086} {\bibfield  {journal} {\bibinfo  {journal}
  {Phys. Rev. Res.}\ }\textbf {\bibinfo {volume} {2}},\ \bibinfo {pages}
  {043086} (\bibinfo {year} {2020})},\ \Eprint
  {http://arxiv.org/abs/2005.14178} {arXiv:2005.14178 [cond-mat.str-el]}
  \BibitemShut {NoStop}%
\bibitem [{\citenamefont {Jia}\ \emph {et~al.}(2024{\natexlab{b}})\citenamefont
  {Jia}, \citenamefont {Kaszlikowski},\ and\ \citenamefont
  {Tan}}]{jia2022electricmagnetic}%
  \BibitemOpen
  \bibfield  {author} {\bibinfo {author} {\bibfnamefont {Z.}~\bibnamefont
  {Jia}}, \bibinfo {author} {\bibfnamefont {D.}~\bibnamefont {Kaszlikowski}}, \
  and\ \bibinfo {author} {\bibfnamefont {S.}~\bibnamefont {Tan}},\ }\bibfield
  {title} {\enquote {\bibinfo {title} {Electric-magnetic duality and
  $\mathbb{Z}_2$ symmetry enriched cyclic abelian lattice gauge theory},}\
  }\href {https://doi.org/10.1088/1751-8121/ad5123} {\bibfield  {journal}
  {\bibinfo  {journal} {Journal of Physics A: Mathematical and Theoretical}\
  }\textbf {\bibinfo {volume} {57}},\ \bibinfo {pages} {255203} (\bibinfo
  {year} {2024}{\natexlab{b}})},\ \Eprint {http://arxiv.org/abs/2201.12361}
  {arXiv:2201.12361 [quant-ph]} \BibitemShut {NoStop}%
\bibitem [{\citenamefont {Wen}(2003)}]{wen2003quantum}%
  \BibitemOpen
  \bibfield  {author} {\bibinfo {author} {\bibfnamefont {X.-G.}\ \bibnamefont
  {Wen}},\ }\bibfield  {title} {\enquote {\bibinfo {title} {Quantum orders in
  an exact soluble model},}\ }\href {\doibase 10.1103/PhysRevLett.90.016803}
  {\bibfield  {journal} {\bibinfo  {journal} {Phys. Rev. Lett.}\ }\textbf
  {\bibinfo {volume} {90}},\ \bibinfo {pages} {016803} (\bibinfo {year}
  {2003})},\ \Eprint {http://arxiv.org/abs/quant-ph/0205004}
  {arXiv:quant-ph/0205004 [quant-ph]} \BibitemShut {NoStop}%
\bibitem [{\citenamefont {Jia}(2024)}]{jia2024SymTFT}%
  \BibitemOpen
  \bibfield  {author} {\bibinfo {author} {\bibfnamefont {Z.}~\bibnamefont
  {Jia}},\ }\href@noop {} {\enquote {\bibinfo {title} {{SymTFT perspective on
  (1+1)d lattice models of weak Hopf non-invertible symmetry-protected
  topological phases}},}\ }\bibinfo {howpublished} {in preparation} (\bibinfo
  {year} {2024})\BibitemShut {NoStop}%
\bibitem [{\citenamefont {Huang}\ and\ \citenamefont
  {Cheng}(2023)}]{huang2023topologicalholo}%
  \BibitemOpen
  \bibfield  {author} {\bibinfo {author} {\bibfnamefont {S.-J.}\ \bibnamefont
  {Huang}}\ and\ \bibinfo {author} {\bibfnamefont {M.}~\bibnamefont {Cheng}},\
  }\href {https://arxiv.org/abs/2310.16878} {\enquote {\bibinfo {title}
  {Topological holography, quantum criticality, and boundary states},}\ }
  (\bibinfo {year} {2023}),\ \Eprint {http://arxiv.org/abs/2310.16878}
  {arXiv:2310.16878 [cond-mat.str-el]} \BibitemShut {NoStop}%
\bibitem [{\citenamefont {Freed}\ \emph {et~al.}(2024)\citenamefont {Freed},
  \citenamefont {Moore},\ and\ \citenamefont {Teleman}}]{freed2024topSymTFT}%
  \BibitemOpen
  \bibfield  {author} {\bibinfo {author} {\bibfnamefont {D.~S.}\ \bibnamefont
  {Freed}}, \bibinfo {author} {\bibfnamefont {G.~W.}\ \bibnamefont {Moore}}, \
  and\ \bibinfo {author} {\bibfnamefont {C.}~\bibnamefont {Teleman}},\ }\href
  {https://arxiv.org/abs/2209.07471} {\enquote {\bibinfo {title} {Topological
  symmetry in quantum field theory},}\ } (\bibinfo {year} {2024}),\ \Eprint
  {http://arxiv.org/abs/2209.07471} {arXiv:2209.07471 [hep-th]} \BibitemShut
  {NoStop}%
\bibitem [{\citenamefont {Gaiotto}\ and\ \citenamefont
  {Kulp}(2021)}]{gaiotto2021orbifold}%
  \BibitemOpen
  \bibfield  {author} {\bibinfo {author} {\bibfnamefont {D.}~\bibnamefont
  {Gaiotto}}\ and\ \bibinfo {author} {\bibfnamefont {J.}~\bibnamefont {Kulp}},\
  }\bibfield  {title} {\enquote {\bibinfo {title} {Orbifold groupoids},}\
  }\href {https://link.springer.com/article/10.1007/JHEP02(2021)132} {\bibfield
   {journal} {\bibinfo  {journal} {Journal of High Energy Physics}\ }\textbf
  {\bibinfo {volume} {2021}},\ \bibinfo {pages} {1} (\bibinfo {year} {2021})},\
  \Eprint {http://arxiv.org/abs/2008.05960} {arXiv:2008.05960 [hep-th]}
  \BibitemShut {NoStop}%
\bibitem [{\citenamefont {Bhardwaj}\ and\ \citenamefont
  {Schafer-Nameki}(2023)}]{bhardwaj2023generalizedcharge}%
  \BibitemOpen
  \bibfield  {author} {\bibinfo {author} {\bibfnamefont {L.}~\bibnamefont
  {Bhardwaj}}\ and\ \bibinfo {author} {\bibfnamefont {S.}~\bibnamefont
  {Schafer-Nameki}},\ }\href {https://arxiv.org/abs/2305.17159} {\enquote
  {\bibinfo {title} {Generalized charges, part ii: Non-invertible symmetries
  and the symmetry tft},}\ } (\bibinfo {year} {2023}),\ \Eprint
  {http://arxiv.org/abs/2305.17159} {arXiv:2305.17159 [hep-th]} \BibitemShut
  {NoStop}%
\bibitem [{\citenamefont {Apruzzi}\ \emph {et~al.}(2023)\citenamefont
  {Apruzzi}, \citenamefont {Bonetti}, \citenamefont {Garc{\'\i}a~Etxebarria},
  \citenamefont {Hosseini},\ and\ \citenamefont
  {Sch{\"a}fer-Nameki}}]{apruzzi2023symmetry}%
  \BibitemOpen
  \bibfield  {author} {\bibinfo {author} {\bibfnamefont {F.}~\bibnamefont
  {Apruzzi}}, \bibinfo {author} {\bibfnamefont {F.}~\bibnamefont {Bonetti}},
  \bibinfo {author} {\bibfnamefont {I.}~\bibnamefont {Garc{\'\i}a~Etxebarria}},
  \bibinfo {author} {\bibfnamefont {S.~S.}\ \bibnamefont {Hosseini}}, \ and\
  \bibinfo {author} {\bibfnamefont {S.}~\bibnamefont {Sch{\"a}fer-Nameki}},\
  }\bibfield  {title} {\enquote {\bibinfo {title} {Symmetry tfts from string
  theory},}\ }\href
  {https://link.springer.com/article/10.1007/s00220-023-04737-2} {\bibfield
  {journal} {\bibinfo  {journal} {Communications in mathematical physics}\
  }\textbf {\bibinfo {volume} {402}},\ \bibinfo {pages} {895} (\bibinfo {year}
  {2023})},\ \Eprint {http://arxiv.org/abs/2112.02092} {arXiv:2112.02092
  [hep-th]} \BibitemShut {NoStop}%
\bibitem [{\citenamefont {Bhardwaj}\ \emph
  {et~al.}(2024{\natexlab{d}})\citenamefont {Bhardwaj}, \citenamefont
  {Bottini}, \citenamefont {Pajer},\ and\ \citenamefont
  {Schafer-Nameki}}]{bhardwaj2024gappedphases}%
  \BibitemOpen
  \bibfield  {author} {\bibinfo {author} {\bibfnamefont {L.}~\bibnamefont
  {Bhardwaj}}, \bibinfo {author} {\bibfnamefont {L.~E.}\ \bibnamefont
  {Bottini}}, \bibinfo {author} {\bibfnamefont {D.}~\bibnamefont {Pajer}}, \
  and\ \bibinfo {author} {\bibfnamefont {S.}~\bibnamefont {Schafer-Nameki}},\
  }\href {https://arxiv.org/abs/2310.03784} {\enquote {\bibinfo {title} {Gapped
  phases with non-invertible symmetries: (1+1)d},}\ } (\bibinfo {year}
  {2024}{\natexlab{d}}),\ \Eprint {http://arxiv.org/abs/2310.03784}
  {arXiv:2310.03784 [hep-th]} \BibitemShut {NoStop}%
\bibitem [{\citenamefont {Zhang}\ and\ \citenamefont
  {C\'ordova}(2024)}]{Zhang2024anomaly}%
  \BibitemOpen
  \bibfield  {author} {\bibinfo {author} {\bibfnamefont {C.}~\bibnamefont
  {Zhang}}\ and\ \bibinfo {author} {\bibfnamefont {C.}~\bibnamefont
  {C\'ordova}},\ }\bibfield  {title} {\enquote {\bibinfo {title} {Anomalies of
  $(1+1)$-dimensional categorical symmetries},}\ }\href {\doibase
  10.1103/PhysRevB.110.035155} {\bibfield  {journal} {\bibinfo  {journal}
  {Phys. Rev. B}\ }\textbf {\bibinfo {volume} {110}},\ \bibinfo {pages}
  {035155} (\bibinfo {year} {2024})}\BibitemShut {NoStop}%
\bibitem [{\citenamefont {Ji}\ and\ \citenamefont
  {Wen}(2020)}]{Ji2020categoricalsym}%
  \BibitemOpen
  \bibfield  {author} {\bibinfo {author} {\bibfnamefont {W.}~\bibnamefont
  {Ji}}\ and\ \bibinfo {author} {\bibfnamefont {X.-G.}\ \bibnamefont {Wen}},\
  }\bibfield  {title} {\enquote {\bibinfo {title} {Categorical symmetry and
  noninvertible anomaly in symmetry-breaking and topological phase
  transitions},}\ }\href {\doibase 10.1103/PhysRevResearch.2.033417} {\bibfield
   {journal} {\bibinfo  {journal} {Phys. Rev. Res.}\ }\textbf {\bibinfo
  {volume} {2}},\ \bibinfo {pages} {033417} (\bibinfo {year}
  {2020})}\BibitemShut {NoStop}%
\end{thebibliography}%

\end{document}